\documentclass[a4paper,twoside,11pt]{article}
\usepackage{latexsym,amsfonts,amsthm}
\usepackage[noend]{algorithmic}
\usepackage{algorithm}
\usepackage{breqn}
\usepackage{lscape}
\usepackage{xcolor}
\usepackage{url}
\usepackage[margin=0.5in]{geometry}
\usepackage{amsfonts}

\usepackage{graphicx}
\usepackage{rotating}
\usepackage{tabulary}

\begin{document}

%\tableofcontents
\newtheorem{defn}{Definition}
\newtheorem{lemma}{Lemma}
\newtheorem{theorem}[lemma]{Theorem}
\newtheorem{corollary}[lemma]{Corollary}
\newcommand{\bx}{\hfill \rule{2mm}{2.5mm}}
\def\nl{\medskip \\ \noindent}
\newcommand{\sat}{\small (2,2)\normalsize {\sc -e}\small 3\normalsize {\sc -sat}}

%% Title and Header Information
%\title{The Hospitals / Residents problem with Couples: Complexity and Integer Programming models}
%%
%\author{P\'eter Bir\'o$^1$ \thanks{Supported by the Hungarian Academy of Sciences under its Momentum Programme (LD-004/2010) and also by OTKA grant no.\ K108673.}  , David F. Manlove$^2$ \thanks{Supported by Engineering and Physical Sciences Research Council grant EP/K010042/1.}, Iain McBride$^2$ \thanks{Supported by a SICSA Prize PhD Studentship.}}
%
%

\title{\bf{The Hospitals / Residents problem with Couples: \\ Complexity and Integer Programming models}}
\author{P\'eter Bir\'o$^{1}$ \thanks{Supported by the Hungarian Academy of Sciences under its Momentum Programme (LD-004/2010) and also by OTKA grant no.\ K108673.} , David F.\ Manlove$^{2}$ \thanks{Supported by grant EP/K010042/1 from the Engineering and Physical Sciences Research Council.} , Iain McBride$^{2}$ \thanks{Supported by a SICSA Prize PhD Studentship}\\ \\ 
\small Institute of Economics, Centre for Economic and Regional Studies, \\ \small Hungarian Academy of Sciences, 1112 Buda\"orsi \'ut 45, Budapest, Hungary.  \\ \small Email {\tt biro.peter@krtk.mta.hu}.\\
\\
\small School of Computing Science, Sir Alwyn Williams Building, \\ \small University of Glasgow, Glasgow G12 8QQ, UK.  \\ \small Email {\tt david.manlove@glasgow.ac.uk, i.mcbride.1@research.gla.ac.uk}. \\
\\ }
%\small School of Computing Science, Sir Alwyn Williams Building, \\ \small University of Glasgow, Glasgow G12 8QQ, UK.  \\ \small Email {\tt i.mcbride.1@research.gla.ac.uk}.}
\date{ }

\maketitle

\begin{abstract}

The Hospitals / Residents problem with Couples ({\sc hrc}) is a generalisation of the classical Hospitals / Resident problem ({\sc hr}) that is important in practical applications because it models the case where couples submit joint preference lists over pairs of (typically geographically close) hospitals. In this report we present new NP-completeness results for the problem of deciding whether a stable matching exists, in highly restricted instances of {\sc hrc}. Further, we present an Integer Programming (IP) model for {\sc hrc} and extend it the case where preference lists can include ties. Also, we describe an empirical study of an IP model for {\sc hrc} and its extension to the case where preference lists can include ties.  This model was applied to randomly generated instances and also real-world instances arising from previous matching runs of the Scottish Foundation Allocation Scheme, used to allocate junior doctors to hospitals in Scotland.

\end{abstract}

\section{Introduction}

\label{section:introduction}

\subsection{The Hospitals / Residents problem.} 

The \emph{Hospitals / Residents problem} ({\sc hr}) is a many-to-one allocation problem. An instance of {\sc hr} consists of two groups of agents -- one containing \emph{hospitals} and one containing \emph{residents}. Every hospital expresses a linear preference over some subset of the residents, its \emph{preference list}. The residents in a hospital's preference list are its \emph{acceptable} partners. Further, every hospital has a \emph{capacity}, $c_j$, the maximum number of posts it has available to match with residents. Every resident expresses a linear preference over some subset of the hospitals, his \emph{acceptable} hospitals. 

The preferences expressed in this fashion are reciprocal: if a resident $r_i$ is acceptable to a hospital $h_j$, then $h_j$ is also acceptable to $r_i$, and vice versa. A many-to-one \emph{matching} between residents and hospitals is sought, which is a set of acceptable resident-hospital pairs such that each resident appears in at most one pair and each hospital $h_j$ at most $c_j$ pairs. If a resident $r_i$ appears in some pair of $M$, $r_i$ is said to be \emph{assigned} in $M$ and \emph{unassigned} otherwise. Any hospital assigned fewer residents than its capacity in some matching $M$ is \emph{under-subscribed} in $M$.

A matching is \emph{stable} if it admits no \emph{blocking pair}. Following the definition used in \cite{GS62}, a blocking pair consists of a mutually acceptable resident-hospital pair $(r, h)$ such that both of the following hold: (i) either $r$ is unassigned, or $r$ prefers $h$ to his assigned hospital; (ii) either $h$ is under-subscribed in the matching, or $h$ prefers $r$ to at least one of its assigned residents. Were such a pair to exist, they could form a private arrangement outside of the matching, undermining its integrity \cite{Rot84}.

It is known that every instance of {\sc hr} admits at least one stable matching and such a matching may be found in time linear in the size of the instance \cite{GS62}. Also, for an arbitrary {\sc hr} instance $I$, any resident that is assigned in one stable matching in $I$ is assigned in all stable matchings in $I$, moreover any hospital that is under-subscribed in some stable matching in $I$ is assigned exactly the same set of residents in every stable matching in $I$ \cite{GS85, Rot86, Rot84}. 

{\sc hr} can be viewed as an abstract model of the matching process involved in a centralised matching scheme such as the National Resident Matching Program (NRMP) \cite{ZZZ5} through which graduating medical students are assigned to hospital posts in the USA. A similar process was used until recently to match medical graduates to Foundation Programme places in Scotland, called the Scottish Foundation Allocation Scheme (SFAS) \cite{Irv98}. Analogous allocation schemes having a similar underlying problem model exist around the world, both in the medical sphere, e.g. in Canada \cite{ZZZ6}, Japan \cite{ZZZ10}, and beyond, e.g. in higher education allocation in Hungary \cite{Bir08}.

\subsubsection{The Hospitals / Residents problem with Couples}

Centralised matching schemes such as the NRMP and the SFAS have had to evolve to accommodate couples who wish to be allocated to (geographically) compatible hospitals. The requirement to take into account the joint preferences of couples has been in place in the NRMP context since 1983 and since 2009 in the case of SFAS. In schemes where the agents may be involved in couples, the underlying allocation problem can modelled by the so-called \emph{Hospitals / Residents problem with Couples} ({\sc hrc}).

As in the case of {\sc hr}, an instance of {\sc hrc} consists of a set of \emph{hospitals} $H$ and a set of \emph{residents} $R$. The residents in $R$ are partitioned into two sets, $S$ and $S^{ \prime}$. The set $S$ consists of \emph{single} residents and the set $S^{\prime }$ consists of those residents involved in \emph{couples}. There is a set $C = \{ (r_i, r_j): r_i, r_j \in S^{\prime} \}$ of \emph{couples} such that each resident in $S^{\prime}$ belongs to exactly one pair in $C$.

Each single resident $r_i \in S$ expresses a linear preference order over his acceptable hospitals. Each pair of residents $(r_i, r_j)\in C$ expresses a joint linear preference order over a subset $A$ of $H \times H$ where $(h_p, h_q)\in A$ represents the joint assignment of $r_i$ to $h_p$ and $r_j$ to $h_q$. The hospital pairs in $A$ represent those joint assignments that are \emph{acceptable} to $(r_i, r_j)$, all other joint assignments being \emph{unacceptable} to $(r_i, r_j)$.

Each hospital $h_j \in H$ expresses a linear preference order over those residents who find $h_j$ acceptable, either as a single resident or as part of a couple. As in the {\sc hr} case, each hospital $h_j \in H$ has a \emph{capacity}, $c_j$. 

A many-to-one \emph{matching} between residents and hospitals is sought, which is defined as for {\sc hr} with the additional restriction that each couple $(r_i, r_j)$ is either jointly unassigned, meaning that both $r_i$ and $r_j$ are unassigned, or jointly assigned to some pair $(h_k, h_l)$ that $(r_i, r_j)$ find acceptable.As in {\sc hr}, we seek a \emph{stable} matching, which guarantees that no resident and hospital, and no couple and pair of hospitals, have an incentive to deviate from their assignments and become assigned to each other.
 
Roth \cite{Rot84} considered stability in the {\sc hrc} context although did not define the concept explicitly. Whilst Gusfield and Irving \cite{GI89} defined stability in {\sc hrc}, their definition neglected to deal with the case that both members of a couple may wish to be assigned to the same hospital. Manlove and McDermid \cite{MM10} extended their definition to deal with this possibility. Henceforth, we refer to Manlove and McDermid's stability definition as \emph{MM-stability}. We  now define this concept formally.

%However, a variety of stability definitions do exist in the {\sc hrc} context \cite{BIS11, GI89, MM10}. Unless explicitly stated otherwise, the definition of stability applied in the work which follows is that described by McDermid and Manlove in \cite{MM10} (MM-stability), shown in Definition \ref{stability:MM}, which expands the earlier definition from Gusfield and Irving \cite{GI89} (GI-stability) by considering also those cases in which a resident pair $(r_i, r_j)$ may express a preference for a hospital pair $(h_k, h_k)$. 

\begin{defn}

A matching $M$ is MM-stable if none of the following holds:  
\begin{enumerate}
\item The matching is blocked by a hospital $h_j$ and a single resident $r_i$, as in the classical HR problem.
\item The matching is blocked by a couple $(r_i , r_j)$ and a hospital $h_k$ such that \emph{either}
\begin{enumerate}
\item[(a)] $(r_i , r_j)$ prefers $(h_k, M(r_j))$ to $(M(r_i), M(r_j))$, and $h_k$ is either under-subscribed in $M$ or prefers $r_i$ to some member of $M(h_k)\backslash \{r_j\}$ \emph{or}
\item[(b)] $(r_i , r_j)$ prefers $(M(r_i), h_k )$ to $(M(r_i), M(r_j))$, and $h_k$ is either under-subscribed in $M$ or prefers $r_j$ to some member of $M(h_k)\backslash \{r_i\}$
\end{enumerate}
\item The matching is blocked by a couple $(r_i, r_j)$ and (not necessarily distinct) hospitals $h_k\neq M(r_i)$, $h_l\neq M(r_j)$; that is, $(r_i, r_j)$ prefers the joint assignment $(h_k, h_l)$ to $(M(r_i), M(r_j))$, and \emph{either}
\begin{enumerate}
\item[(a)] $h_k\neq h_l$, and $h_k$ (respectively $h_l$) is either under-subscribed in $M$ or prefers $r_i$ (respectively $r_j$) to 
at least one of its assigned residents in $M$; \emph{or}
\item[(b)] $h_k=h_l$, and $h_k$ has at least two free posts in $M$, i.e., $c_k-|M(h_k)|\geq 2$; \emph{or}
\item[(c)] $h_k=h_l$, and $h_k$ has one free post in $M$, i.e., $c_k-|M(h_k)|=1$, and $h_k$ prefers at least one of
$r_i,r_j$ to some member of $M(h_k)$; \emph{or}
\item[(d)] $h_k=h_l$, $h_k$ is full in $M$, $h_k$ prefers $r_i$ to some $r_s\in M(h_k)$, and $h_k$ prefers $r_j$ to some
$r_t\in M(h_k)\backslash \{r_s\}$.
\end{enumerate}
\end{enumerate}

%\caption{The Manlove-McDermid Stability definition. \cite{MM10}}
\label{stability:MM}
\end{defn}

\color{black}

The majority of the results in this paper for {\sc hrc} are given in terms of MM-stability. A further stability definition due to Bir\'o et al \cite{BIS11} (henceforth \emph{BIS-stability}) can be applied in contexts where the hospitals rank the residents according to an agreed criterion, such as in the SFAS context where this criterion is represented by a score derived for each resident from their academic performance and their completed application. The rationale behind the BIS stability definition being that in a stable matching $M$, if a resident $r$ is not matched to a hospital $h$ then all of the residents who are matched to $h$ in $M$ are strictly preferable to $r$ under the applied criterion. We now define BIS-stability formally as follows.

\begin{defn}
A matching $M$ is BIS-stable if none of the following holds:  
\begin{enumerate}
\item The matching is blocked by a hospital $h_j$ and a single resident $r_i$, as in the classical HR problem.
\item The matching is blocked by a hospital $h_k$ and a resident $r_i$ who is coupled, say with $r_j$; that is \emph{either}
\begin{enumerate}
	\item[(a)] $(r_i,r_j)$ prefers $(h_k, M(r_j))$ to $(M(r_i), M(r_j))$ and \emph{either} 
			\begin{enumerate}
			\item[(i)] $h_k\neq M(r_j)$ and $h_k$ is either under-subscribed in $M$ or prefers $r_i$ to some member of $M(h_k)$ \emph{or}
			\item[(ii)] $h_k = M(r_j)$ and $h_k$ is either under-subscribed in $M$ or prefers both $r_i$ and $r_j$ to some member of $M(h_k) \setminus \{r_j\}$
			\end{enumerate}
	\item[(b)] $(r_i,r_j)$ prefers $(M(r_i), h_k)$ to $(M(r_i), M(r_j))$ and \emph{either} (i) or (ii) as above adapted to symmetric case
\end{enumerate}
\item The matching is blocked by a couple $(r_i, r_j)$ and (not necessarily distinct) hospitals $h_k\neq M(r_i)$ and $h_l\neq M(r_j)$; that is, $(r_i, r_j)$ prefers the joint assignment $(h_k, h_l)$ to $(M(r_i), M(r_j))$, and \emph{either}
\begin{enumerate}
\item[(a)] $h_k\neq h_l$, and $h_k$ (respectively $h_l$) is either under-subscribed in $M$ or prefers $r_i$ (respectively $r_j$) to at least one of its assignees in $M$; \emph{or}
\item[(b)] $h_k=h_l$, and $h_k$ has at least two free posts in $M$ \emph{or}
\item[(c)] $h_k=h_l$, and $h_k$ has one free post in $M$ and both $r_i$  and $r_j$ are preferred by $h_k$ to some member of $M(h_k)$ \emph{and} 
\item[(d)] $h_k=h_l$, $h_k$ is full in $M$ and \emph{either}
\begin{enumerate} 
\item[(i)] $h_k$ prefers each of $r_i$ and $r_j$ to some $r_p \in M(h_k)$ who is a member of a couple with some $r_q \in M(h_k)$,
\item[(ii)] the least preferred resident among $r_i$ and $r_j$ (according to $h_k$) is preferred by $h_k$ to two members of $M(h_k)$

\end{enumerate}
\end{enumerate}
\end{enumerate}

%\caption{The Bir\'o, Schlotter and Irving Stability Definition.}
\label{stability:BIS}
\end{defn}

It is notable that, for an arbitrary instance $I$ of {\sc hrc} in which hospitals may have a capacity greater than 1, an MM-stable matching need not be BIS-stable and vice versa. The instances described in Section \ref{section:stabilityComparison} demonstrate this unequivocally. In the restriction of {\sc hrc} in which each hospitals has a capacity of 1, BIS-stability and MM-stability are both equivalent to the stability definition from Gusfield and Irving \cite{GI89} since no couple $(r_i, r_j)$ may express a preference for a hospital pair $(h_k, h_k)$.   

The \emph{Hospitals / Residents Problem with Couples and Ties} ({\sc hrct}) is a generalisation of {\sc hrc} in which hospitals (respectively residents) may find some subsets of their acceptable residents (respectively hospitals) equally preferable. Residents (respectively hospitals) that are found equally preferable by a hospital (respectively resident) are \emph{tied} with each other in the preference list of that hospital (respectively resident). The stability definitions in  Definition \ref{stability:MM} and Definition \ref{stability:BIS} remain unchanged in the {\sc hrct} context.

\subsection{Existing algorithmic results for {\sc hrc}.} 

In contrast with {\sc hr}, an instance of {\sc hrc} need not admit a stable matching \cite{Rot84}. Also an instance of {\sc hrc} may admit stable matchings of differing sizes \cite{AC96}. Further, the problem of deciding whether a stable matching exists in an instance of {\sc hrc} is NP-complete, even in the restricted case where there are no single residents and all of the hospitals have only one available post \cite{Ron90, NH90}. 

In many practical applications of {\sc hrc} the residents' preference lists are short. Let $(\alpha, \beta)$-{\sc hrc} denote the restriction of {\sc hrc} in which each single resident's preference list contains at most $\alpha $ hospitals, each couple's preference list contains at most $\alpha $ pairs of hospitals and each hospital's preference list contains at most $\beta $ residents.  $(\alpha, \beta)$-{\sc hrc} is hard even for small values of $\alpha $ and $\beta $: Manlove and McDermid \cite{MM10} showed that \small $(3,6)$\normalsize {\sc -hrc} is NP-complete. 

%Klaus and Klijn \cite{KK07} showed in an example instance of {\sc hrc}, that there exists a matching $M$ from which we cannot reach a stable matching by a process of satisfying blocking pairs. The process instead leading to a \emph{cycle} of continually satisfying blocking pairs which does not result in a stable matching. %This suggests that any algorithm $A$ which seeks to find a stable matching in an arbitrary instance of HRC when $A$ is based on satisfying blocking pairs cannot be guaranteed to succeed in all instances of HRC.

A further restriction of {\sc hrc} is {\sc hrc-dual-market}, defined as follows. Given an instance $I$ of {\sc hrc}, let the set of all first members of each couple in $I$ be $R_1 \subseteq R$, and the set of second members of each couple in $I$ be $R_2 \subseteq R$. Let the set of acceptable partners of the residents in $R_1$ in $I$ be $H_1 \subseteq H$ and the set of acceptable partners of the residents in $R_2$ in $I$ be $H_2 \subseteq H$. If in $I$, $H_1 \cap H_2 = \emptyset$ and no single resident has acceptable partners in both $H_1$ and $H_2$ then we define $I$ to be an instance of {\sc hrc-dual-market} consisting of the two disjoint markets $R_1 \cup H_1$ and $R_2 \cup H_2$. The problem of deciding whether an instance of {\sc hrc-dual-market} admits a stable matching is also known to be NP-complete \cite{NH88} even if the instance contains no single residents and the hospitals all have capacity one.

Since the existence of an efficient algorithm for finding a stable matching, or reporting that none exists, in an instance of {\sc hrc} is unlikely, in practical applications such as SFAS and NRMP, stable matchings are found by applying heuristics \cite{BIS11, RP97}. However, neither the SFAS heuristic, nor the NRMP heuristic guarantee to terminate and output a stable matching, even in instances where a stable matching does exist. Hence, a method which guarantees to find a maximum cardinality stable matching in an arbitrary instance of {\sc hrc}, where one exists, might be of considerable interest. For further results on {\sc hrc} the reader is referred to \cite{BK11} and \cite{DM13}. 

For further results in {\sc hrc} the reader is referred to \cite{BK11} and \cite{DM13}. 

\subsection{Linear Programming, Integer Programming and Constraint Programming techniques applied to {\sc hr} and its variants}

Vande Vate \cite{VV89} described a Linear Programming (LP) formulation for the Stable Marriage problem, the one-to-one variant of {\sc hr} in which all hospitals have a capacity of 1, the numbers of residents and hospitals are the same, and each of the residents finds every hospital acceptable. Rothblum \cite{Rot92} generalised this model to the {\sc hr} context for arbitrary instances.  Ba\"{\i}ou and Balinski \cite{BB00a} formulated an LP model for {\sc hr} which Fleiner \cite{Fle03a} further generalised to the many-to-many version of {\sc hr}, a variant in which both hospitals and residents may have capacities exceeding one.

Podhradsky \cite{Pod10} empirically investigated the performance of approximation algorithms for {\sc max-smti} (the 1-1 restriction of {\sc max-hrt}, the NP-hard problem of finding a maximum cardinality stable matching given an instance of {\sc hrt}) and compared them against one another and against an IP formulation for {\sc max-smti}. Kwanashie and Manlove \cite{KM14} described an Integer Programming (IP) model for the Hospitals \ Residents problem with Ties ({\sc hrt}). 

Manlove et al. \cite{MOPU07} and Eirinakis et al \cite{EMMM07} applied Constraint Programming (CP) techniques to {\sc hr} while O'Malley \cite{OMa07} described a CP formulation for {\sc hrt}. Subsequently, Eirinakis et al \cite{EMMM12} gave a generalised CP formulation for many-to-many {\sc hr}. The reader is referred to Ref. \cite[Sections 2.4 \& 2.5]{DM13} for more information about previous work involving the application of IP and CP techniques applied to allocation problems such as {\sc hr}.

\subsection{Contribution of this work} 

In Section \ref{section:complexityresults} of this paper we present a collection of new hardness results for {\sc hrc}. We begin in Section \ref{subsection:22NPcomplete} by presenting a new NP-completeness result for the problem of deciding whether there exists a stable matching in an instance of \small $(2,2)$\normalsize {\sc -hrc} where there are no single residents and all hospitals have capacity 1. This is the most restricted case of {\sc hrc} currently known for which NP-completeness holds. A natural way to try to cope with this complexity is to approximate a matching that is `as stable as possible', i.e., admits the minimum number of blocking pairs \cite{ABM06}. Let {\sc min-bp-hrc} denote the problem of finding a matching with the minimum number of blocking pairs, given an instance of {\sc hrc}, and let \small $(\alpha ,\beta )$\normalsize {\sc -min-bp-hrc} denote the restriction to instances of \small $(\alpha ,\beta )$\normalsize {\sc -hrc}. In Section \ref{subsection:22HRC_Inapprox} we prove that \small $(2,2)$\normalsize {\sc -min-bp-hrc} is not approximable within $N^{1- \varepsilon}$, where $N$ is the number of residents in a given instance, for any $\varepsilon > 0$, unless P=NP. 

In Section \ref{subsection:23NPcomplete} we show that \small $(2,3)$\normalsize {\sc -hrc} is also NP-complete even when all of the hospitals have capacity 1 and there are no single residents and the preference list of each couple and hospital are derived from a strictly ordered master list of pairs of hospitals and residents respectively. Further, in Section \ref{subsection:33NPcomplete} we show that deciding whether an instance of \small $(3,3)$\normalsize {\sc -hrc-dual-market} (which is an instance of {\sc hrc-dual-market} in which the residents', couples' and hospitals' preference lists have a maximum length of three) admits a stable matching, is NP-complete even when all hospitals have capacity 1 and the preference lists of all residents, couples and hospitals are derived from a master list of hospitals, hospital pairs and residents respectively.

Further in Section \ref{section:IPModelsHRC} we present a description of the first IP model for finding a maximum cardinality stable matching or reporting that none exists in an arbitrary instance of {\sc hrc}. Then in Section \ref{section:IPexperiments} we present experimental results obtained from a Java implementation of the IP model for {\sc hrc} applied to randomly generated instances of {\sc hrc} constructed in a manner consistent with the SFAS format. We further extend this model to find a maximum cardinality stable matching in the more general {\sc hrct} context and apply it to real-world instances arising from previous matching runs of SFAS, the allocation scheme previously used to allocate junior doctors to hospitals in Scotland. 

Section \ref{section:stabilityComparison} of the report presents a cloning methodology for {\sc hrc} which can be used to construct an instance of one-to-one {\sc hrc} from an instance of many-to-one {\sc hrc} where the MM-stable matchings in the many-to-one instance are in correspondence to the MM-stable matchings in the one-to-one instance. Finally, in Section \ref{section:conclusions} we present some conclusions based on the work presented.

\section{Complexity Results for {\sc hrc}}

\label{section:complexityresults}

\subsection{Introduction}

In this section we present hardness results for finding and approximating stable matchings in instances of {\sc hrc}. We begin in Section \ref{subsection:22NPcomplete} by establishing NP-completness for the problem of deciding whether a stable matching exists in a highly restricted instance of {\sc hrc}. 

We then turn to {\sc min-bp-hrc} in Section \ref{subsection:22HRC_Inapprox}. Clearly Theorem \ref{22-HRC} in Section \ref{subsection:22NPcomplete} implies that this problem is NP-hard. By chaining together instances of \small $(2,2)$\normalsize {\sc -hrc} constructed in the proof of Theorem \ref{22-HRC}, we arrive at a gap-introducing reduction which establishes a strong inapproximability result for {\sc min-bp-hrc} under the same restrictions as in Theorem \ref{22-HRC} in Section \ref{subsection:22NPcomplete}.

Using a similar reduction to the one used to show NP-completeness in instances of \small $(2,2)$\normalsize {\sc -hrc}, in Section \ref{subsection:23NPcomplete} we show that, given an instance of \small $(2,3)$\normalsize {\sc -hrc}, the problem of deciding whether the instance admits a stable matching is NP-complete even under the restriction that all of the hospitals have capacity 1, there are no single residents and the preference lists for all of the single residents, couples and hospitals are derived from a master list of hospitals, hospital pairs and residents respectively.

Finally in Section \ref{subsection:33NPcomplete} we show that, given an instance of \small $(3,3)$\normalsize {\sc -hrc-dual-market}, the problem of deciding whether there exists a stable matching is NP-complete. Again, we show that the result holds under the restriction that the hospitals have capacity 1 and that the result also holds under the further restriction that the preference lists for all of the single residents, couples and hospitals are derived from a master list of hospitals, hospital pairs and residents respectively.

\subsection{Complexity results for \small $(2,2)$\normalsize {\sc -hrc}}
\subsubsection{NP-completeness result for \small $(2,2)$\normalsize {\sc -hrc}}
\label{subsection:22NPcomplete}

\begin{theorem}
\label{22-HRC}
Given an instance of $(2,2)$\normalsize {\sc -hrc}, the problem of deciding whether there exists a stable matching is NP-complete. The result holds even if there are no single residents and each hospital has capacity 1.
\end{theorem}

\begin{proof}

The proof of this result uses a reduction from a restricted version of {\sc sat}.  More specifically, let \sat\ denote the problem of deciding, given a Boolean formula $B$ in CNF over a set of variables $V$, whether $B$ is satisfiable, where $B$ has the following properties: (i) each clause contains exactly 3 literals and (ii) for each $v_i\in V$, each of literals $v_i$ and $\bar{v_i}$ appears exactly twice in $B$. Berman et al.\ \cite{BKS03} showed that \sat\ is NP-complete.

The problem \small $(2,2 )$\normalsize {\sc -hrc} is clearly in NP, as a given assignment may be verified to be a stable matching in polynomial time. To show NP-hardness, let $B$ be an instance of \sat . Let $V=\{v_1,v_2,\dots,v_n\}$ and $C=\{c_1,c_2,\dots,c_m\}$ be the set of variables and clauses respectively in $B$. Then for each $v_i\in V$, each of literals $v_i$ and $\bar{v_i}$ appears exactly twice in $B$. Also $|c_j|=3$ for each $c_j\in C$. (Hence $m=\frac{4n}{3}$.) We form an instance $I$ of \small $(2,2 )$\normalsize {\sc -hrc} as follows.

The set of residents in $I$ is $A\cup B\cup X\cup Y$ where $A=\bigcup_{i=1}^{n} A_i$, $A_i=\{a_i^r : 1\leq r\leq 2 \}$ ($1\leq i\leq n$), $B=\bigcup_{i=1}^{n} B_i$, $B_i=\{b_i^r : 1\leq r\leq 2 \}$ ($1\leq i\leq n$), $X=\bigcup_{j=1}^{m} X_j$, $X_j=\{x_j^s : 1\leq s\leq 3 \}$ ($1\leq j\leq m$) and $Y=\bigcup_{j=1}^{m} Y_j$, $Y_j=\{y_j^s : 1\leq s\leq 3 \}$ ($1\leq j\leq m$). There are no single residents in $I$ and the pairing of the residents into couples is as shown in Figure \ref{preflists22}.

The set of hospitals in $I$ is $H\cup T$, where $H=\bigcup_{i=1}^{n} H_i$, $H_i=\{h_i^r : 1\leq r\leq 6 \}$ ($1\leq i\leq n$) and $T=\bigcup_{j=1}^{m} T_j$, $T_j=\{t_j^r : 1\leq r\leq 6 \}$ ($1\leq j\leq m$) and each hospital has capacity 1. The preference lists of the resident couples and hospitals in $I$ are shown in Figure \ref{preflists22}. 

%In the joint preference list of a couple $(x_j^s, y_j^s)$  $(x_j^s\in X, y_j^s\in Y)$ the symbol $c(h_i^r)$  $(3\leq r\leq 6)$ denotes the hospital $h_i^r\in H$ such that if $r = 3$ (respectively $r=5$) then the first (respectively second) occurrence of literal $v_i$ appears at position $s$ of clause $c_j$ in $B$. Similarly, if $r = 4$ (respectively $r=6$) then the first (respectively second) occurrence of literal $\bar{v}_i$ appears at position $s$ of clause $c_j$ in $B$.

In the joint preference list of a couple $(x_j^s, y_j^s) ~ (1\leq j\leq m, 1\leq s\leq 3)$ the symbol $h(x^s_j)$ is defined as follows. If the $r^{th}$ occurrence $(1\leq r\leq 2)$ of literal $v_i$ occurs at position $s$ of $c_j$ then $h(x^s_j) = h^{2r+1}_i$. If the $r^{th}$ occurrence $(1\leq r\leq 2)$ of literal $\bar{v}_i$ occurs at position $s$ of $c_j$ then $h(x^s_j) = h^{2r+2}_i$.

In the preference list of a hospital $h^{2r+1}_i ~ (1\leq r\leq 2)$, the symbol $x(h^{2r+1}_i)$ denotes the resident $x^s_j$ such that the $r^{th}$ occurrence of literal $v_i$ occurs at position $s$ of clause $c_j$. Similarly in the preference list of a hospital $h^{2r+2}_i ~ (1\leq r\leq 2)$, the symbol $x(h^{2r+2}_i)$ denotes the resident $x^s_j$ such that the $r^{th}$ occurrence of literal $\bar{v}_i$ occurs at position $s$ of clause $c_j$.

\begin{figure}
\[
\begin{array}{rll}

(a_i^1, b_i^1) : & (h_i^1, h_i^3) ~~ (h_i^2, h_i^4) & (1\leq i\leq n) \vspace{1mm} \\ 
(a_i^2, b_i^2) : & (h_i^2, h_i^5) ~~ (h_i^1, h_i^6) & (1\leq i\leq n) \vspace{1mm} \\ 

\\

(x_j^1, y_j^1) : & (h(x_j^1), t_j^4) ~~ (t_j^1, t_j^3) & (1\leq j\leq m) \vspace{1mm} \\ 
(x_j^2, y_j^2) : & (h(x_j^2), t_j^5) ~~ (t_j^2, t_j^1) & (1\leq j\leq m) \vspace{1mm} \\ 
(x_j^3, y_j^3) : & (h(x_j^3), t_j^6) ~~ (t_j^3, t_j^2) & (1\leq j\leq m) \vspace{1mm} \\ 

\\

h_i^1 : & a_i^2 ~~ a_i^1 & (1\leq i\leq n) \vspace{1mm}\\
h_i^2 : & a_i^1 ~~ a_i^2 & (1\leq i\leq n) \vspace{1mm}\\
h_i^3 : & b_i^1 ~~ x(h^3_i) & (1\leq j\leq m) \vspace{1mm}\\   % s is a position in the clause
h_i^4 : & b_i^1 ~~ x(h^4_i) & (1\leq j\leq m) \vspace{1mm}\\
h_i^5 : & b_i^2 ~~ x(h^5_i) & (1\leq j\leq m) \vspace{1mm}\\
h_i^6 : & b_i^2 ~~ x(h^6_i) & (1\leq j\leq m) \vspace{1mm}\\

\\

t_j^1 : & x_j^1 ~~ y_j^2 & (1\leq j\leq m) \vspace{1mm}\\
t_j^2 : & x_j^2 ~~ y_j^3 & (1\leq j\leq m) \vspace{1mm}\\
t_j^3 : & x_j^3 ~~ y_j^1 & (1\leq j\leq m) \vspace{1mm}\\

\\

t_j^4 : & y_j^1 & (1\leq j\leq m) \vspace{1mm}\\
t_j^5 : & y_j^2 & (1\leq j\leq m) \vspace{1mm}\\
t_j^6 : & y_j^3 & (1\leq j\leq m) \vspace{1mm}\\

\\

\end{array}
\]
\caption{Preference lists in $I$, the constructed instance of \small $(2,2)$\normalsize {\sc -hrc}.}
\label{preflists22}
\end{figure}

For each $i$ ($1\leq i\leq n$), let $T_i=\{(a_i^1, h_i^2), (a_i^2, h_i^1), (b_i^1, h_i^4), (b_i^2, h_i^6), (x(h^3_i), h^3_i), (x(h^5_i), h^5_i) \}$ and $F_i=\{(a_i^1, h_i^1), $ $(a_i^2, h_i^2), (b_i^1, h_i^3), (b_i^2, h_i^5), (x(h^4_i), h^4_i), (x(h^6_i), h^6_i) \}$.

We claim that $B$ is satisfiable if and only if $I$ admits a stable matching. 

Let $f$ be a satisfying truth assignment of $B$.  Define a matching $M$ in $I$ as follows.  For each variable $v_i\in V$, if $v_i$ is true under $f$, add the pairs in $T_i$ to $M$, otherwise add the pairs in $F_i$ to $M$. Let $j ~ (1\leq j\leq m)$ be given. Then $c_j$ contains at least one literal that is true under $f$. Suppose $c_j$ contains exactly one literal that is true under $f$. Let $s$ be the position of $c_j$ containing a true literal. In this case add the pairs $\{ (x^{s+1}_j, t^{s+1}_j), (y^{s+1}_j, t^{s}_j) \}$ (where addition is taken modulo three) to $M$. Now suppose $c_j$ contains exactly two literals that are true under $f$. Let $s$ be the position of $c_j$ containing a false literal, and add the pairs $\{ (x^{s}_j, t^{s}_j), (y^{s}_j, t^{s+2}_j) \}$ (where addition is taken modulo three) to $M$. If $c_j$ contains 3 literals which are true under $f$ no additional pairs need be added.

No resident pair $(a_i^1, b_i^1)$ or $(a_i^2, b_i^2)$ may be involved in a blocking pair of $M$, as no matching in which $(a_i^1, b_i^1)$ is matched with $(h_i^2, h_i^4)$ is blocked by $(a_i^1, b_i^1)$ with $(h_i^1, h_i^3)$, and equally no matching in which $(a_i^2, b_i^2)$ is matched with $(h_i^1, h_i^6)$ is blocked by $(a_i^2, b_i^2)$ with $(h_i^2, h_i^5)$.

No resident pair $(x_j^s, y_j^s) ~ (1\leq s\leq 3)$ may block $M$ with $(h(x_j^s), t_j^{s+3})$ (where addition is taken modulo three). To prove this observe that all $h^r_i$ are matched in $M$ and hence if some $h^r_i$ is not matched to its corresponding $x(h^r_i)$ then $h^r_i$ must be matched to the member of $B_i$ in first place on its preference list. Thus $(x_j^s, y_j^s)$ may not block $M$ with $(h(x_j^s), t_j^{s+3})$.

No resident pair $(x_j^s, y_j^s) ~ (1\leq s\leq 3)$ may be involved in a blocking pair of $M$ with $(t_j^s, t_j^{s+2})$ (where addition is taken modulo three). Clearly $(x_j^s, y_j^s)$ may only block $M$ with $(t_j^s, t_j^{s+2})$ if $(x_j^s, y_j^s)$ is unmatched in $M$. From the construction, this may only be the case if $c_j$ contains exactly one literal that is true under $f$. In this case, $(x_j^{s+2}, y_j^{s+2})$ is matched with $(t_j^{s+2}, t_j^{s+4})$ (where addition is taken modulo 3) and thus $(x_j^s, y_j^s) ~ (1\leq s\leq 3)$ may not block $M$ with $(t_j^s, t_j^{s+2})$, since $t_j^{s+2}$ prefers $x_j^{s+2}$ to $y^s_j$.

Hence $M$ is a stable matching in $I$.
Conversely suppose that $M$ is a stable matching in $I$. We form a truth assignment $f$ in $B$ as follows.  

For any $i$  $(1\leq i\leq n)$, if $(a^1_i, b^1_i)$ is unmatched then $M$ is blocked by $(a^1_i, b^1_i)$ with $(h^2_i, h^4_i)$. Similarly, if $(a^2_i, b^2_i)$ is unmatched then $M$ is blocked by $(a^2_i, b^2_i)$ with $(h^1_i, h^6_i)$.  Hence either $\{(a_i^1, h_i^2),(b^1_i, h_i^4), (a_i^2, h_i^1),(b^2_i, h_i^6)\}\subseteq M$ or $\{(a_i^1, h_i^1),(b^1_i, h_i^3), (a_i^2, h_i^2),(b^2_i, h_i^5)\}\subseteq M$.

Now let $c_j$ be a clause in $C$ ($1\leq j\leq m$). Suppose, $( x^s_j, h(x^s_j) )\notin M$ for all $s ~ (1\leq s\leq 3)$. Clearly, at most one couple $(x^s_j, y^s_j)$ may be matched to the hospital pair in second place on its preference list.  Since no $(x^s_j, y^s_j)$ is matched to the pair in first place in its preference list one of the remaining two unmatched $(x^s_j, y^s_j)$'s must block with the hospital pair in second place on its preference list, a contradiction. Thus $\{ ( x^s_j, h(x^s_j)), ( y^s_j, t^{s+3}_j ) ) \} \subseteq M$ for some $s ~ (1\leq s\leq 3)$ by the stability of $M$

Hence, for each $j ~ (1\leq j\leq m)$, let $s ~ (1\leq s\leq 3)$ be such that $(x^s_j, y^s_j)$ is matched with $(h(x^s_j), t^{s+3}_j)$. Let $h^r_i = h(x^s_j)$. If $r\in \{ 3,5 \}$ then we set $f(v_i) = T$. Thus, variable $v_i$ is true under $f$ and hence clause $c_j$ is true under $f$ since the literal $v_i$ occurs in $c_j$. Otherwise, $r\in \{ 4,6 \}$ and we set $f(v_i) = F$. Thus, variable $v_i$ is false under $f$ and hence clause $c_j$ is true under $f$ since the literal $\bar{v}_i$ occurs in $c_j$. 

This assignment of truth values is well-defined, for if $(h_i^r,t_j^{s+3})\in M$ for $r\in \{3,5\}$ then $\{(b_i^1,h_i^4),(b_i^2,h_i^6)\}\subseteq M$, so neither $h_i^4$ nor $h_i^6$ is partnered with a member of $X$ in $M$.  Similarly if $(h_i^r,t_j^{s+3})\in M$ for $r\in \{4,6\}$ then $\{(b_i^1,h_i^3),(b_i^2,h_i^5)\}\subseteq M$, so neither $h_i^3$ nor $h_i^5$ is partnered with a member of $X$ in $M$. Hence $f$ is a satisfying truth assignment of $B$.
\end{proof}

\subsubsection{Inapproximability of \small $(2,2)$\normalsize {\sc -min-bp-hrc}}
\label{subsection:22HRC_Inapprox}

\begin{figure}
\[
\begin{array}{rll}

(a, b) : & (z_1, z_2) &  \vspace{1mm} \\ 
c : & z_1 ~~ z_2 \vspace{1mm}\\
\\

z_1 : & a ~~ c &  \vspace{1mm}\\
z_2 : & c ~~ b &  \vspace{1mm}\\

\end{array}
\]
\caption{A small instance of $(2,2)$\normalsize {\sc -hrc} which admits no stable matching.}
\label{smalloneblockingpairinstance}
\end{figure}

Let $A$ be an instance of \small $(2,2)$\normalsize {\sc -hrc} as shown in Figure \ref{smalloneblockingpairinstance}. In $A$ the residents are $a$, $b$ and $c$, the hospitals are $z_1$ and $z_2$ and each hospital has a capacity of 1. The instance $S$ admits three non-empty matchings, namely

$$M_1 = \{ (a, z_1 ), (b, z_2) \}$$
$$M_2 = \{ (c, z_1 ) \}$$
$$M_3 = \{ (c, z_2 ) \}$$

Clearly, none of the matchings is stable. Further each of the non-empty matchings in $A$ is blocked by precisely one blocking pair. $M_1$ is blocked in $A$ by only $(c, z_2)$, $M_2$ is blocked in $A$ by only $((a,b), (z_1, z_2))$ and $M_3$ is blocked in $A$ by only $(c,z_1 )$. Thus, $A$ admits no stable matching but amongst all of the non-empty matchings admitted by $A$, the only number of blocking pairs possible (and thus the minimum) is 1. We shall use this property of the instance $A$ in the proofs that follow. The following theorem allows us to prove that, unless P=NP, there is a limit on the approximability of {\sc min-bp-hrc} in a highly restricted case.

\begin{theorem}
\label{22-HRC-APPROX}
{\sc min-bp-hrc} is not approximable within $n_1^{1- \varepsilon}$, where $n_1$ is the number of residents in a given instance, for any $\varepsilon > 0$ unless P=NP even if there are no single residents and the resident couples and hospitals have preference lists of length $\leq 2$.
\end{theorem}

\begin{proof}

Let $B$ be an instance of \sat\ and let $I$ be the corresponding instance of \small $(2,2)$\normalsize {\sc -hrc} as constructed in Theorem \ref{22-HRC}. We show how to modify $I$ in order to obtain an extended instance $I^{\prime \prime}$ of \small $(2,2)$\normalsize {\sc -hrc} as follows. 

Choose $c = \lceil 2 / \varepsilon \rceil$ and $k= n^c$. Now, let $I_1 , I_2 , \ldots , I_k$ be $k$ disjoint copies of the instance $I$. Let $I^{\prime }$ be the \small $(2,2)$\normalsize {\sc -hrc} instance formed by taking the union of the sub-instances $I_1 , I_2 , \ldots , I_k$. Let $I^{\prime \prime}$ be the instance constructed by adding the instance $A$ of \small $(2,2)$\normalsize {\sc -hrc} described in Figure \ref{smalloneblockingpairinstance} to $I^{\prime }$.

If $B$ admits a satisfying truth assignment then by Theorem \ref{22-HRC}, $I$ admits a stable matching and clearly each copy of $I$ must also admit a stable matching. Thus $I^{\prime }$ must admit a stable matching. Moreover, since any non-empty matching admitted by $A$ has exactly one blocking pair, a matching exists in $I^{\prime \prime}$ which has exactly one blocking pair.

If $B$ admits no satisfying truth assignment, then by Theorem \ref{22-HRC}, $I$ admits no stable matching. We claim that any matching admitted by $I^{\prime \prime}$ must be blocked by $k+1$ or more blocking pairs. Since $I$ admits no stable matching, any matching in $I$ must have at least one blocking pair. Thus each $I_r ~ (1\leq r\leq k)$ admits only matchings with one or more blocking pair. Since the only non-empty matchings admitted by $A$ have a single blocking pair, any matching admitted by $I^{\prime \prime }$ must have at least $k+1$ blocking pairs.

The number of residents in $I^{\prime \prime}$ is $n_1 = 4nk + 6mk + 3$. From the construction of $I$ in Theorem \ref{22-HRC} we know that $4n = 3m$ and thus $n_1 \leq 12nk + 3$. We lose no generality by assuming that $n \geq 3$. Thus $n_1 \leq 13nk = 13n^{c+1}$.

Moreover, 
%
%
%$$\dfrac{N} {13} \leq n^{c+1}$$
% which implies 
% $$\left (\dfrac{N}{13}\right )  ^{1/c+1} \leq n$$
%Since $k = n^c$ we may show that 
%$$(\dfrac{N}{13})^{c/c+1} \leq k$$
%and hence 

\begin{equation} \label{HRC22Inapp-Eq1} \displaystyle 13^{-c/(c+1)}n_1^{c/(c+1)} \leq k. \end{equation}

Now we know that $n_1 \geq k = n^c$. We lose no generality by assuming that $n \geq 13$ and hence $n_1 \geq 13^c$. It follows that

\begin{equation} \label{HRC22Inapp-Eq2} \displaystyle n_1^{-1/(c+1)} \leq 13^{-c/(c+1)}. \end{equation}

Thus it follows from Inequality \ref{HRC22Inapp-Eq1} and \ref{HRC22Inapp-Eq2} that

% So 
%$$N^{-1} \leq 13^{-c}$$
%and thus 
%$$N^{-1/c+1} \leq 13^{-c/c+1}$$
%Hence 
%$$N^{-1/c+1} N^{c/c+1} \leq 13^{-c/c+1}N^{c/c+1}$$
%As we have previously shown that 
%$$13^{-c/c+1}N^{c/c+1} \leq k$$
%we may show that 
\begin{equation} \label{HRC22Inapp-Eq3} n_1^{c-1/c+1} = n_1^{c/(c+1)} n_1^{-1/(c+1)} \leq 13^{-c/(c+1)} n_1^{c/(c+1)} \leq k \end{equation}

We now show that $n_1^{1- \varepsilon} \leq n_1^{c-1/c+1}$. Observe that $c \geq 2 / \varepsilon$ and thus $c + 1 \geq 2 / \varepsilon$. Hence 
%
%$$\varepsilon \geq \dfrac{2}{ c+1}$$
%
% and thus %$$1 - \varepsilon \leq 1 - \dfrac{2}{c+1}$$ and 
$$1 - \varepsilon \leq \dfrac{c+1 - 2}{c+1} = \dfrac{c-1}{c+1}$$ 

and hence by Inequality \ref{HRC22Inapp-Eq3}, $n_1^{1- \varepsilon} \leq k$.

Assume that $X$ is an approximation algorithm for \small $(2,2)$\normalsize {\sc -hrc} with a performance guarantee of $n_1^{1 - \varepsilon } \leq k$. Let $B$ be an instance of \sat\ and construct an instance $I^{\prime \prime }$ of \small $(2,2)$\normalsize {\sc -hrc} from $B$ as described above. If $B$ admits a satisfying truth assignment then $X$ must return a matching in $I^{\prime \prime }$ which admits $\leq k$ blocking pairs. Otherwise, $B$ does not admit a satisfying assignment and $X$ must return a matching which admits $\geq k + 1$ blocking pairs. Thus algorithm $X$ may be used to determine whether $B$ admits a satisfying truth assignment in polynomial time, a contradiction. Hence, no such polynomial approximation algorithm can exist unless $P = NP$.
\end{proof}

\subsection{Complexity results for \small $(2,3)$\normalsize {\sc -hrc} with master lists}
\label{subsection:23NPcomplete}

\begin{theorem}
\label{23-HRC}
Given an instance of $(2,3 )$\normalsize {\sc -hrc}, the problem of deciding whether the instance supports a stable matching is NP-complete. The result holds even if there are no single residents and each hospital has capacity 1.
\end{theorem}

\begin{proof}

The proof of this result uses a reduction from a restricted version of {\sc sat}.  More specifically, let \sat\ denote the problem of deciding, given a Boolean formula $B$ in CNF over a set of variables $V$, whether $B$ is satisfiable, where $B$ has the following properties: (i) each clause contains exactly 3 literals and (ii) for each $v_i\in V$, each of literals $v_i$ and $\bar{v_i}$ appears exactly twice in $B$. Berman et al.\ \cite{BKS03} showed that \sat\ is NP-complete.

The problem \small $(2,3 )$\normalsize {\sc -hrc} is clearly in NP, as a given assignment may be verified to be a stable matching in polynomial time. 

To show NP-hardness, let $B$ be an instance of \sat.  Let $V=\{v_1,v_2,\dots,v_n\}$ and $C=\{c_1,c_2,\dots,c_m\}$ be the set of variables and clauses respectively in $B$. Then for each $v_i\in V$, each of literals $v_i$ and $\bar{v_i}$ appears exactly twice in $B$. Also $|c_j|=3$ for each $c_j\in C$. (Hence $m=\frac{4n}{3}$.) We form an instance $I$ of \small $(2,3 )$\normalsize {\sc -hrc} as follows.

The set of residents in $I$ is $X\cup P\cup Q$, where $X=\bigcup_{i=1}^{n} X_i$, $X_i=\{x_i^r, \bar{x}_i^r : 1\leq r\leq 2 \}$ ($1\leq i\leq n$), $P=\bigcup_{j=1}^m P_j$, $P_j=\{p_j^s : 1\leq s\leq 3\}$ ($1\leq j\leq m$), $Q=\bigcup_{j=1}^m Q_j$, $Q_j=\{q_j^s : 1\leq s\leq 3\}$ ($1\leq j\leq m$). 

There are no single residents in $I$ and the pairing of the residents into couples is as shown in Figure \ref{preflists23}.

The set of hospitals in $I$ is $H\cup Y$, where $H=\bigcup_{i=1}^{n} H_i$, $H_i=\{h_i^r : 1\leq r\leq 2 \}$ ($1\leq i\leq n$) and $Y=\{y_j^s : 1\leq s\leq 3\}$  $(1\leq j\leq m)$ and each hospital has capacity 1. The preference lists of the resident couples and hospitals in $I$ are shown in Figure \ref{preflists23}. 

In the joint preference list of a couple $(x_i^1, x_i^2)$  $(x_i^1\in X, x_i^2\in X)$ the symbol $y(x_i^r)$  $(1\leq r\leq 2)$ denotes the hospital $y_j^s\in Y$ such that the $r$th occurrence of literal $v_i$ appears at position $s$ of clause $c_j$ in $B$.

Similarly in the joint preference list of a couple $(\bar{x}_i^1, \bar{x}_i^2)$  $(\bar{x}_i^1\in X, \bar{x}_i^2\in X)$ the symbol $y(\bar{x}_i^r)$  $(1\leq r\leq 2)$ denotes the hospital $y_j^s\in Y$ such that the $r$th occurrence of literal $\bar{v}_i$ appears at position $s$ of clause $c_j$ in $B$.

Also in the preference list of a hospital $y_j^s\in Y$, if literal $v_i$ (respectively $\bar{v}_i$) appears at position $s$ of clause $c_j\in C$, the symbol $x(y_j^s)$ denotes the resident $x_i^1$ or $x_i^2$ (respectively  $\bar{x}_i^1$ or $\bar{x}_i^2$) according as this is the first or second occurrence of the literal in $B$.

\begin{figure}
\[
\begin{array}{rll}

(x_i^1, x_i^2) : & (h_i^1, h_i^2) ~~ (y(x_i^1), y(x_i^2)) & (1\leq i\leq n) \vspace{1mm} \\ 
(\bar{x}_i^1, \bar{x}_i^2) : & (h_i^1, h_i^2) ~~ (y(\bar{x}_i^1), y(\bar{x}_i^2)) & (1\leq i\leq n) \vspace{1mm} \\ 

(p_j^1, q_j^1) : & (y_j^1, y_j^2) & (1\leq j\leq m) \vspace{1mm} \\ 
(p_j^2, q_j^2) : & (y_j^2, y_j^3) & (1\leq j\leq m) \vspace{1mm} \\ 
(p_j^3, q_j^3) : & (y_j^3, y_j^1) & (1\leq j\leq m) \vspace{1mm} \\

\\

h_i^1 : & x_i^1 ~~ \bar{x}_i^1 & (1\leq i\leq n) \vspace{1mm}\\
h_i^2 : & \bar{x}_i^2 ~~ x_i^2 & (1\leq i\leq n) \vspace{1mm}\\

y_j^1 : & x(y^1_j) ~~ p_j^1 ~~ q_j^3 & (1\leq j\leq m) \vspace{1mm}\\
y_j^2 : & x(y^2_j) ~~ p_j^2 ~~ q_j^1 & (1\leq j\leq m) \vspace{1mm}\\
y_j^3 : & x(y^3_j) ~~ p_j^3 ~~ q_j^2 & (1\leq j\leq m) \vspace{1mm}\\

\end{array}
\]
\caption{Preference lists in $I$, the constructed instance of $(2,3)$\normalsize {\sc -hrc}.}
\label{preflists23}
\end{figure}

For each $i$ ($1\leq i\leq n$), let $T_i=\{(x_i^1, y(x_i^1)), (x_i^2, y(x_i^2)), (\bar{x}_i^1, h_i^1),(\bar{x}_i^2, h_i^2)\}$ and 
$F_i=\{(x_i^1, h_i^1), (x_i^2, h_i^2),(\bar{x}_i^1, y(\bar{x}_i^1),(\bar{x}_i^2, y(\bar{x}_i^2) \}$.

We claim that $B$ is satisfiable if and only if $I$ admits a stable matching.

Let $f$ be a satisfying truth assignment of $B$.  Define a matching $M$ in $I$ as follows.  For each variable $v_i\in V$, if $v_i$ is true under $f$, add the pairs in $T_i$ to $M$, otherwise add the pairs in $F_i$ to $M$.

Let $j$  $(1\leq j\leq m)$ be given. Then $c_j$ contains at least one literal that is true under $f$. Suppose this literal occurs at position $s$ of $c_j$  $(1\leq s \leq 3)$, then $(x(y^s_j), y^s_j)\in M$. If no other literal in $c_j$ is true then add the pairs $\{ (p^{s+1}_j, y^{s+1}_j), (q^{s+1}_j, y^{s+2}_j)\}$ to $M$ (where addition is taken modulo 3).

No resident pair $(x_i^1, x_i^2)$ or $(\bar{x}_i^1, \bar{x}_i^2)$ may block $M$, as no matching in which $(x_i^1, x_i^2)$ is matched with $(h_i^1, h_i^2)$ is blocked by $(\bar{x}_i^1, \bar{x}_i^2)$ with $(h_i^1, h_i^2)$, and equally no matching in which $(\bar{x}_i^1, \bar{x}_i^2)$ is matched with $(h_i^1, h_i^2)$ is blocked by $(x_i^1, x_i^2)$ with $(h_i^1, h_i^2)$.

No resident pair $(p_j^s, q_j^s)$ may block $M$ as, if $(p_j^s, q_j^s)$ is not matched to $(y_j^s, y_j^{s+1})$ (where addition is taken modulo 3), then at least one of $y^s_j$ or $y^{s+1}_j$ is matched to its first choice and therefore $(p_j^s, q_j^s)$ may not block $M$ with $(y_j^s, y_j^{s+1})$.

Hence $M$ is a stable matching in $I$.

Conversely suppose that $M$ is a stable matching in $I$. We form a truth assignment $f$ in $B$ as follows.

For any $i$  $(1\leq i\leq n)$, if $h^1_i$ and $h^2_i$ are unmatched then $M$ is blocked by ($x^1_i, x^2_i)$ with $(h^1_i, h^2_i)$. Hence either $\{(x_i^1, h_i^1),(x^2_i, h_i^2)\}\subseteq M$ or $\{(\bar{x}_i^1, h_i^1),(\bar{x}^2_i, h_i^2)\}\subseteq M$. Suppose $(x_i^1, x^2_i)$ are unmatched. Then $(x_i^1, x^2_i)$ blocks with $(y(x_i^1), y(x^2_i))$. Similarly $(\bar{x}_i^1, \bar{x}^2_i)$ must be matched. Hence $M\cap (X_i\times (H\cup Y)) = T_i$ or $M\cap (X_i\times (H\cup Y)) = F_i$. In the former case set $f(x_i ) = T_i$ and in the latter case set $f(x_i ) = F_i $.

Now let $c_j$ be a clause in $C$ ($1\leq j\leq m$). Suppose $(x(y^s_j), y^s_j)\notin M$ for all $s$  $(1\leq s\leq 3)$. If $(p^1_j, q^1_j)$ is matched to $(y^1_j, y^2_j)$ then $(p^2_j, q^2_j)$ blocks with $(y^2_j, y^3_j)$. If $(p^2_j, q^2_j)$ is matched to $(y^2_j, y^3_j)$ then $(p^3_j, q^3_j)$ blocks $M$ with $(y^3_j, y^1_j)$. If $(p^3_j, q^3_j)$ is matched to $(y^3_j, y^1_j)$ then $(p^1_j, q^1_j)$ blocks $M$ with $(y^1_j, y^2_j)$. If $\{ (p^s_j, y^s_j), (q^s , y_j , ^{s+1} ) \} \not \subseteq M$ for $1\leq s\ leq 3$ (where addition is taken modulo 3) then $(p^1 _j, q^1 _j)$ blocks $M$ with $(y^1_j , y^2_j )$. Therefore $(x(y^s_j), y^s_j)\in M$ for some $s$  $(1\leq s\leq 3)$ by the stability of $M$.

If $x(y^s_j)=x^r_i$ then $(x^r_i, y^s_j)\in T_i$ and therefore $v_i$ is true under $f$ and hence $c_j$ is true under $f$. Otherwise $x(y^s_j)=\bar{x}^r_i$, so $(\bar{x}^r_i, y^s_j)\in F_i$ and $\bar{v}_i$ is false under $f$ and hence $c_j$ is true under $f$.

Hence $f$ is a satisfying truth assignment of $B$.

\end{proof}

\begin{corollary}
\label{23ml}
Given an instance of \small $(2,3)$\normalsize {\sc -hrc}, the problem of deciding whether a stable matching exists is NP-complete. The result holds even in the case where there are no single residents, each hospital has capacity 1 and the preference list of each couple and hospital are derived from a strictly ordered master list of pairs of hospitals and residents respectively.
\end{corollary}
\begin{proof}

Consider the instance $I$ of \small $(2,3)$\normalsize {\sc -hrc} as constructed in the proof of Theorem \ref{23-HRC}, and let $n$ and $m$ be defined as in the proof of Theorem \ref{23-HRC}. The master lists shown in Figures \ref{masterlistrescouples23} and \ref{masterlisthospitals23} indicate that the preference list of each resident couple and hospital may be derived from a master list of hospital pairs and residents respectively. Since there are no single residents in $I$, no preferences are expressed for individual hospitals in $I$.

As we have previously shown that \sat\ may be reduced to the instance $I$ in polynomial time, the corollary has been proven. 

\begin{figure}
\[
\begin{array}{cll}
L^1_i :
(h_{i}^{1 }, h_{i}^{2}) ~

(y(x_{i}^{1}), y(x_{i}^{2})) ~

(y(\bar{x}_{i}^{1}), y(\bar{x}_{i}^{2}) ) &

&

\vspace{1mm}

(1\leq i\leq n) \\
\\

L^2_i :

(y^1_j , y^2_j) ~
(y^2_j , y^3_j) ~
(y^3_j , y^1_j) &

&

\vspace{1mm}

(1\leq j\leq m)
\\
\\

\mbox{Master List} : L^1_1 ~ L^1_2\dots L^1_{n} ~ L^2_1 ~ L^2_2\dots L^2_{m}

\end{array}
\]
\caption{Master list of preferences for resident couples in $(2,3)$\normalsize {\sc -hrc} instance $I$.}
\label{masterlistrescouples23}
\end{figure}

\begin{figure}
\[
\begin{array}{cll}

L^3_i : 
x_{i}^1 ~
\bar{x}_{i}^1 ~

\bar{x}_{i}^2 ~
x_{i}^2 &

&

\vspace{1mm}

(1\leq i\leq n) \\

\\

L^4_j : p_j^1 ~ p_j^2 ~ p_j^3 &

&

\vspace{1mm}

(1\leq j\leq m)\\ \\

L^5_i : q_j^1 ~ q_j^2 ~ q_j^3 &

&

\vspace{1mm}

(1\leq j\leq m)\\

\\
\mbox{Master List} : L^3_1 ~ L^3_2\dots L^3_n ~ L^4_1 ~ L^4_2\dots L^4_{m} ~  L^5_1 ~ L^5_2\dots L^5_{m} 
\end{array}
\]
\caption{Master list of preferences for hospitals in $(2,3)$\normalsize {\sc -hrc} instance $I$.}
\label{masterlisthospitals23}
\end{figure}
\end{proof}

\subsection{Complexity results for \small $(3,3)$\normalsize {\sc -hrc-dual-market} with master lists}
\label{subsection:33NPcomplete}

\begin{theorem}
\label{33-COM-HRC}
Given an instance of \small $(3,3)$\normalsize {\sc -hrc}, the problem of deciding whether the instance supports a complete stable matching is NP-complete. The result holds even if all hospitals have capacity 1.
\end{theorem}

\begin{proof}
The problem is clearly in NP, as a given assignment may be verified to be a complete, stable matching in polynomial time. 

We use a reduction from a restricted version of {\sc sat}.  More specifically, let \sat\ denote the problem of deciding, given a Boolean formula $B$ in CNF over a set of variables $V$, whether $B$ is satisfiable, where $B$ has the following properties: (i) each clause contains exactly 3 literals and (ii) for each $v_i\in V$, each of literals $v_i$ and $\bar{v_i}$ appears exactly twice in $B$. Berman et al.\ \cite{BKS03} showed that \sat\ is NP-complete.

Let $B$ be an instance of \sat. We construct an instance $I$ of \small $(3,3)$\normalsize {\sc -hrc} using a similar reduction to that employed by Irving et al. \cite{IMO06}. Let $V=\{v_0,v_1,\dots,v_{n-1}\}$ and $C=\{c_1,c_2,\dots,c_m\}$ be the set of variables and clauses respectively in $B$.  Then for each $v_i\in V$, each of literals $v_i$ and $\bar{v_i}$ appears exactly twice in $B$. Also $|c_j|=3$ for each $c_j\in C$. (Hence $m=\frac{4n}{3}$.) All hospitals in $I$ have capacity 1.

The set of residents in $I$ is $X\cup K\cup P\cup Q\cup T$, where $X=\cup_{i=0}^{n-1} X_i$, $X_i=\{x_{4i+r} : 0\leq r\leq 3\}$ ($0\leq i\leq n-1$), $K=\cup_{i=0}^{n-1} K_i$, $K_i=\{k_{4i+r} : 0\leq r\leq 3\}$ ($0\leq i\leq n-1$), $P=\cup_{j=1}^m P_j$, $P_j=\{p_j^r : 1\leq r\leq 6\}$ ($1\leq j\leq m$), $Q=\{q_j : c_j\in C\}$ and $T=\{t_j : c_j\in C\}$. The residents in $Q\cup T$ are single and the residents in $X\cup K\cup P$ are involved in couples.

The set of hospitals in $I$ is $Y\cup L\cup C'\cup Z$, where $Y=\cup_{i=0}^{n-1} Y_i$, $Y_i=\{y_{4i+r} : 0\leq r\leq 3\}$ ($0\leq i\leq n-1$), $L=\cup_{i=0}^{n-1} L_i$, $L_i=\{l_{4i+r} : 0\leq r\leq 3\}$ ($0\leq i\leq n-1$), $C'=\{c_j^s : c_j\in C\wedge 1\leq s\leq 3\}$ and $Z=\{z_j^r : 1\leq j\leq m\wedge 1\leq r\leq 5\}$.

In the joint preference list of a couple $(x_{4i+r}, k_{4i+r}) (x_{4i+r}\in X, k_{4i+r}\in K)$, if $r\in \{0,1\}$, the symbol $c(x_{4i+r})$ 
denotes the hospital $c_j^s\in C'$ such that the $(r+1)$th occurrence of literal $v_i$ appears at position $s$ of $c_j$.  Similarly if $r\in \{2,3\}$ then the symbol $c(x_{4i+r})$ denotes the hospital $c_j^s\in C'$ such that the $(r-1)$th occurrence of literal $\bar{v_i}$ appears at position $s$ of $c_j$. Also in the preference list of a hospital $c_j^s\in C'$, if literal $v_i$ appears at position $s$ of clause $c_j\in C$, the symbol $x(c_j^s)$ denotes the resident $x_{4i+r-1}$ where $r=1,2$ according as this is the first or second occurrence of literal $v_i$ in $B$.  Otherwise
if literal $\bar{v_i}$ appears at position $s$ of clause $c_j\in C$, the symbol $x(c_j^s)$ denotes the resident $x_{4i+r+1}$ where $r=1,2$ according as this is the first or second occurrence of literal $\bar{v_i}$ in $B$.  

The preference lists of the residents and hospitals in $I$ are shown in Figure \ref{preflists33} and diagramatically in Figures \ref{image:clausegadget} and \ref{image:variablegadget}. Clearly each preference list is of length at most 3.

\begin{figure}
\[
\begin{array}{rll}
(x_{4i}, k_{4i}) : & (y_{4i}, l_{4i}) ~~ (c(x_{4i}), l_{4i+1}) ~~ (y_{4i+1}, l_{4i+1}) & (0\leq i\leq n-1)\\
(x_{4i+1}, k_{4i+1}) : & (y_{4i+1}, l_{4i+1}) ~~ (c(x_{4i+1}), l_{4i+2}) ~~ (y_{4i+2}, l_{4i+2}) & (0\leq i\leq n-1)\\
(x_{4i+2}, k_{4i+2}) : & (y_{4i+3}, l_{4i+3}) ~~ (c(x_{4i+2}), l_{4i+2}) ~~ (y_{4i+2}, l_{4i+2}) & (0\leq i\leq n-1)\\
(x_{4i+3}, k_{4i+3}) : & (y_{4i}, l_{4i}) ~~ (c(x_{4i+3}), l_{4i+3}) ~~ (y_{4i+3}, l_{4i+3}) & (0\leq i\leq n-1)\\
(p_j^1, p_j^4) : & (z^1_j, z^2_j)  ~~ (c_j^1, z^3_j) & (1\leq j\leq m) \vspace{1mm}\\
(p_j^2, p_j^5) : & (z^1_j, z^2_j)  ~~ (c_j^2, z^4_j) & (1\leq j\leq m) \vspace{1mm}\\
(p_j^3, p_j^6) : & (z^1_j, z^2_j)  ~~ (c_j^3, z^5_j) & (1\leq j\leq m) \vspace{1mm}\\ 
q_j : & c_j^1 ~~ c_j^2 ~~ c_j^3 & (1\leq j\leq m)\\ 
t_j : & z^3_j ~~ z^4_j ~~ z^5_j & (1\leq j\leq m) \vspace{1mm} \\ \\

y_{4i} : & x_{4i} ~~ x_{4i+3} & (0\leq i\leq n-1)\\
y_{4i+1} : & x_{4i+1} ~~ x_{4i} & (0\leq i\leq n-1)\\
y_{4i+2} : & x_{4i+1} ~~ x_{4i+2} & (0\leq i\leq n-1)\\
y_{4i+3} : & x_{4i+2} ~~ x_{4i+3} & (0\leq i\leq n-1)\\
l_{4i} : & k_{4i+3} ~~ k_{4i} & (0\leq i\leq n-1)\\
l_{4i+1} : & k_{4i} ~~ k_{4i+1} & (0\leq i\leq n-1)\\
l_{4i+2} : & k_{4i+2} ~~ k_{4i+1} & (0\leq i\leq n-1)\\
l_{4i+3} : & k_{4i+3} ~~ k_{4i+2} & (0\leq i\leq n-1)\\
z^1_j : & p_j^1 ~~ p_j^2 ~~ p_j^3 & (1\leq j\leq m) \vspace{1mm}\\
z^2_j : & p_j^6 ~~ p_j^5 ~~ p_j^4 & (1\leq j\leq m) \vspace{1mm}\\
z^3_j : & p_j^4 ~~ t_j & (1\leq j\leq m) \vspace{1mm}\\
z^4_j : & p_j^5 ~~ t_j & (1\leq j\leq m) \vspace{1mm}\\
z^5_j : & p_j^6 ~~ t_j & (1\leq j\leq m) \vspace{1mm} \\
c_j^s : & p_j^s ~~ x(c_j^s) ~~ q_j & (1\leq j\leq m\wedge 1\leq s\leq 3) \vspace{1mm}\\

\end{array}
\]
\caption{Preference lists in $I$, the constructed instance of \small $(3,3)$\normalsize {\sc -hrc}.}
\label{preflists33}
\end{figure}

\begin{figure}
\[
\includegraphics[width = \linewidth]{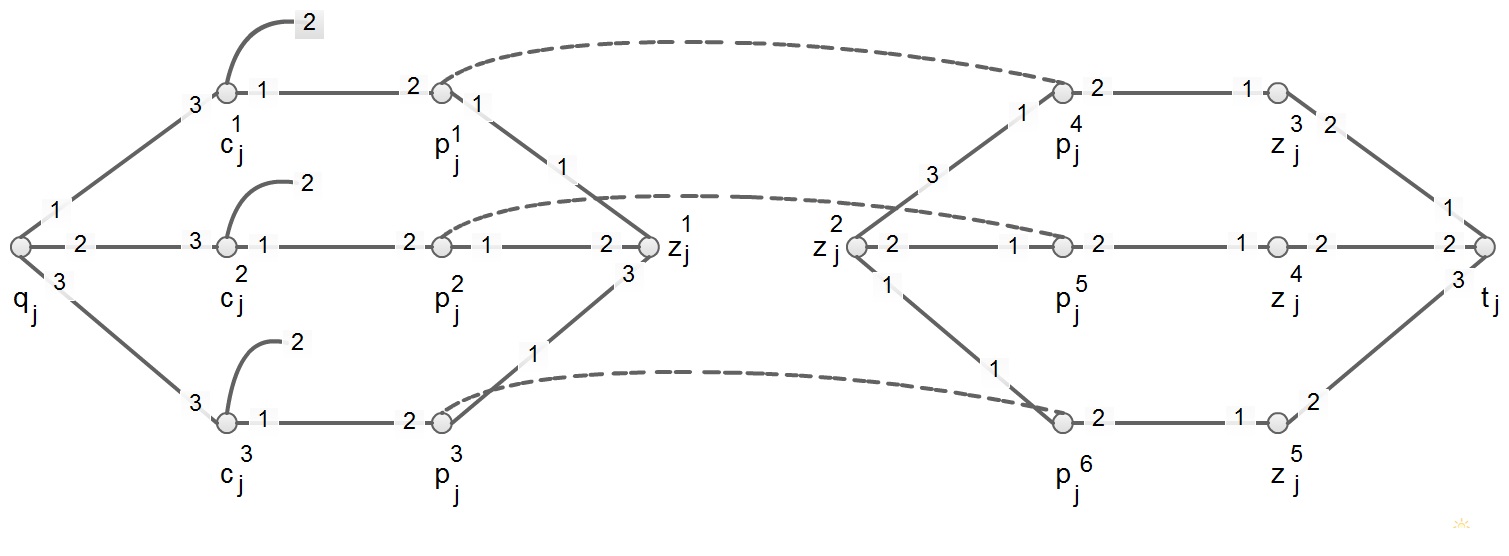}
\]
\caption{Diagram of subset of agents in $I$, the constructed instance of \small $(3,3)$\normalsize {\sc -hrc}.}
\label{image:clausegadget}
\end{figure}

\begin{figure}
\[
\includegraphics[scale = 0.375]{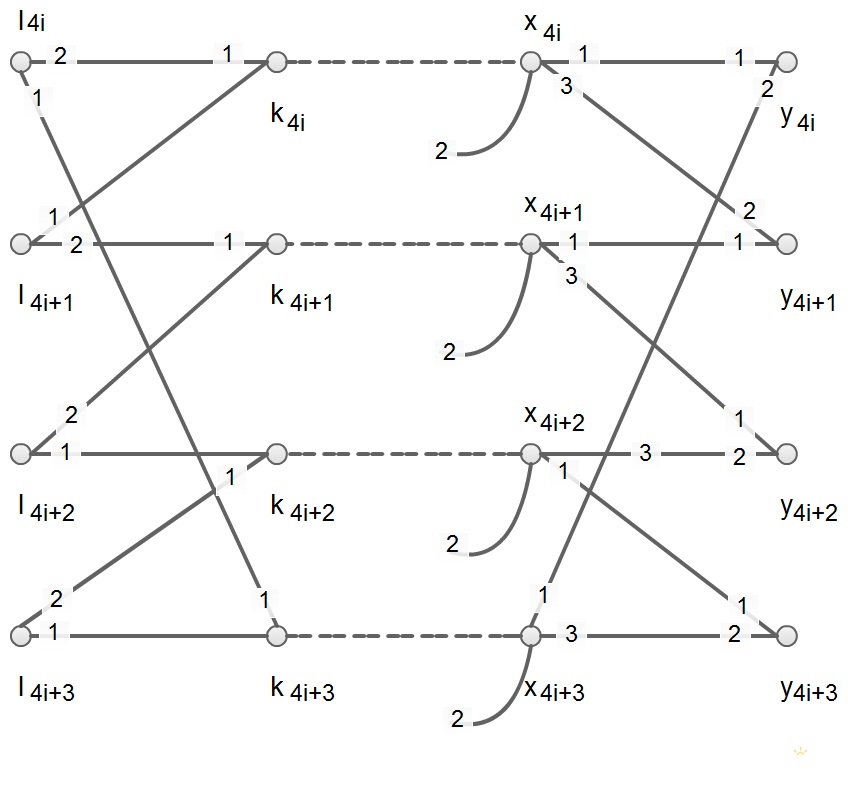}
\]
\caption{Diagram of subset of agents in $I$, the constructed instance of \small $(3,3)$\normalsize {\sc -hrc}.}
\label{image:variablegadget}
\end{figure}

For each $i$ ($0\leq i\leq n-1$), let $T_i=\{(x_{4i+r},y_{4i+r}) : 0\leq r\leq 3\} \cup \{(k_{4i+r},l_{4i+r}) : 0\leq r\leq 3\}$ and 
$F_i=\{(x_{4i+r},y_{4i+r+1})\} : 0\leq r\leq 3\}\cup \{(k_{4i+r},l_{4i+r+1})\} : 0\leq r\leq 3\}$, where addition is taken modulo $4$.

We claim that $B$ is satisfiable if and only if $I$ admits a complete stable matching. 

Firstly suppose that $B$ is satisfiable and let $f$ be a satisfying truth assignment of $B$.  Define a complete matching $M$ in $I$ as follows.  For each variable $v_i\in V$, if $v_i$ is true under $f$, add the pairs in $T_i$ to $M$, otherwise add the pairs in $F_i$ to $M$.

Let $j (1\leq j\leq m)$ be given. Then $c_j$ contains at least one literal that is true under $f$. Suppose this literal occurs at position $s$ of $c_j$  $(1\leq s \leq 3)$. Then add the pairs ($q_j, c^s_j$), ($p^s_j, z^1_j$) and ($p_j^{s+3}, z^2_j$) to $M$. For each $b (1\leq b\leq 3, b \neq s)$ add the pairs ($p^b_j,c^b_j$) and ($p_j^{b+3}, z_j^{b+2}$) to $M$. Finally add the pair ($t_j, z_j^{s+2}$) to $M$.

No resident in $Q$ may form a blocking pair (since he can only potentially prefer a hospital in $C$, which ranks him last) nor can a resident in $T$ (since he can only potentially prefer a hospital $z_j (3\leq j\leq 5)$
who ranks him last).

Suppose that ($x_{4i+r}, k_{4i+r}$) blocks $M$ with $(c(x_{4i+r}), l_{4i+a})$, where $0\leq i\leq n-1$, $0\leq r\leq 3$ and $1\leq a\leq 3$. Then ($x_{4i+r}, k_{4i+r}$) are jointly matched with their third choice pair.

\noindent\emph{Case (i)}: $r\in \{0,1\}$. Then $f(v_i) = F$ because $(x_{4i+r}, y_{4i+r+1})\in M$ and $(k_{4i+r}, l_{4i+r+1})\in M$. Let $c^s_j = c(x_{4i+r}) (1\leq s\leq 3 and 1\leq j\leq m)$. As $v_i$ does not make $c_j$ true then $(p_j^s, c_j^s)\in M$. This means that $c_j^s$ has its first choice and cannot form part of a blocking pair,
a contradiction.

\noindent\emph{Case (ii)}: $r\in \{2,3\}$ Then $f(v_i) = T$ because $(x_{4i+r}, y_{4i+r})\in M$ and $(k_{4i+r}, l_{4i+r})\in M$. Let $c^s_j = c(x_{4i+r}) (1\leq s\leq 3 and 1\leq j\leq m)$. As $\bar{v}_i$ does not make $c_j$ true then $(p_j^s, c_j^s)\in M$. This means that $c_j^s$ has its first choice and cannot form part of a blocking pair,
a contradiction.

Now suppose $(p^s_j, p_j^{s+3})$ blocks $M$. Then $(p^s_j, c^s_j)\in M$ and $(p_j^{s+3}, z_j^{s+2})\in M$, and $(p_j^s, p_j^{s+3})$ jointly prefer $(z^1_j, z^2_j)$ to their partners. At most one of $\{z^1_j, z^2_j\}$  can prefer the relevant member of $\{p^s_j, p^{s+3}_j\}$ to their partner, whilst the other would prefer their current partner in $M$ and thus could not form part of a blocking pair, a contradiction. Hence, $M$ is a complete stable matching in $I$.

Conversely suppose that $M$ is a complete stable matching in $I$. We form a truth assignment $f$ in $B$ as follows.  For each $i$ ($0\leq i\leq n-1$), 
$M\cap ( (X_i\times Y_i)\cup (K_i\times L_i) )$ is a perfect matching of  ($X_i\cup Y_i) \cup (K_i\cup L_i$). If $M\cap ( (X_i\times Y_i)\cup (K_i\times L_i))=T_i$, set $v_i$ to be true under $f$.  Otherwise $M\cap ( (X_i\times Y_i)\cup (K_i\times L_i))=F_i$, in which case we set $v_i$ to be false under $f$.

Now let $c_j$ be a clause in $C$ ($1\leq j\leq m$).  There exists some $s$ ($1\leq s\leq 3$) such that $(q_j,c_j^s)\in M$.  Let $x_{4i+r}=x(c_j^s)$, for some $i$ ($0\leq i\leq n-1$) and $r$ ($0\leq r\leq 3$).  If $r\in \{0,1\}$ then $(x_{4i+r},y_{4i+r})\in M$ (or equivalently $(k_{4i+r}, l_{4i+r})\in M$) by the stability 
of $M$.  Thus variable $v_i$ is true under $f$, and hence clause $c_j$ is true under $f$, since literal $v_i$ occurs in $c_j$.  If $r\in \{2,3\}$ then $(x_{4i+r},y_{4i+r+1})\in M$ (or equivalently $(k_{4i+r}, l_{4i+r+1})\in M$) (where addition is taken modulo $4$) by the stability of $M$.  Thus variable $v_i$ is false under $f$, and hence clause $c_j$ is true under $f$, since literal $\bar{v_i}$ occurs in $c_j$.  Hence $f$ is a satisfying truth assignment of $B$.

Thus, the claim holds and $B$ is satisfiable if and only if $I$ admits a complete, stable matching.
\end{proof}

%
%
%\begin{theorem}
%\label{33-HRC}
%The decision problem of whether a stable matching exists in an instance of $(3,3)$\normalsize {\sc -hrc} is NP-complete. The result holds even if each hospital has capacity 1 and the preference list of each single resident, couple and hospital are derived from a strictly ordered master list of hospitals, pairs of hospitals and residents respectively.
%\end{theorem}
%

Let $I$ be the instance of \small $(3,3)$\normalsize {\sc -com-hrc} as constructed in the proof of Theorem \ref{33-COM-HRC} from an instance of \sat. We add additional residents and hospitals to $I$ to obtain a new instance $I^{\prime }$ of \small $(3,3)$\normalsize {\sc -com-hrc} as follows. For every $y_{4i+r} \in Y$ add further residents $U = \cup_ {i=0}^{n-1} U_i, U_i = \{ u^s _{4i+r} : 1\leq s \leq 5 \wedge 0\leq r \leq 3 \}$ and further hospitals $H = \cup_ {i=0}^{n-1} H_i, H_i = \{ h^s _{4i+r} : 1\leq s \leq 4 \wedge 0\leq r \leq 3 \}$. The preference lists of the agents so added for a single $y_{4i+r} \in Y$ are shown in Figure \ref{preflists:addedgadget} and diagrammatically in Figure \ref{image:addedgadget} where $\Phi _{4i+r}$ represents those preferences expressed by $y_{4i+r}$ in $I^{\prime } \setminus ( U \cup H)$.

\begin{figure}
\[
\begin{array}{rll}
(u^1_{4i+r}, u^2_{4i+r}) : & (h^1 _{4i+r}, h^2 _{4i+r}) ~~ ~~ & (0\leq i\leq n-1, 0\leq r\leq 3)\\
(u^3_{4i+r}, u^4_{4i+r}) : & (h^1 _{4i+r}, h^4 _{4i+r}) ~~ (h^3 _{4i+r}, h^2 _{4i+r})~~ & (0\leq i\leq n-1, 0\leq r\leq 3)\\
u^5 _{4i+r} : & y_{4i+r} ~~ h^1 _{4i+r} ~~ & (0\leq i\leq n-1, 0\leq r\leq 3) \vspace{1mm} \\ \\

y_{4i+r} : & \Phi _{4i+r} ~~ u^5 _{4i+r} & (0\leq i\leq n-1, 0\leq r \leq 3) \vspace{1mm} \\
h^1_{4i+r} : & u^5 _{4i+r} ~~ u^1 _{4i+r} ~~ u^3 _{4i+r} & (0\leq i\leq n-1, 0\leq r \leq 3) \vspace{1mm}\\
h^2_{4i+r} : & u^4 _{4i+r} ~~ u^2 _{4i+r} & (0\leq i\leq n-1, 0\leq r \leq 3) \vspace{1mm}\\
h^3_{4i+r} : & u^3 _{4i+r} & (0\leq i\leq n-1, 0\leq r \leq 3) \vspace{1mm}\\
h^4_{4i+r} : & u^4 _{4i+r} & (0\leq i\leq n-1, 0\leq r \leq 3) \vspace{1mm}\\

\end{array}
\]
\caption{Added agents in $I^{\prime }$, the constructed instance of \small $(3,3)$\normalsize {\sc -hrc}.}
\label{preflists:addedgadget}
\end{figure}

\begin{figure}
\[
\includegraphics[scale = 0.35]{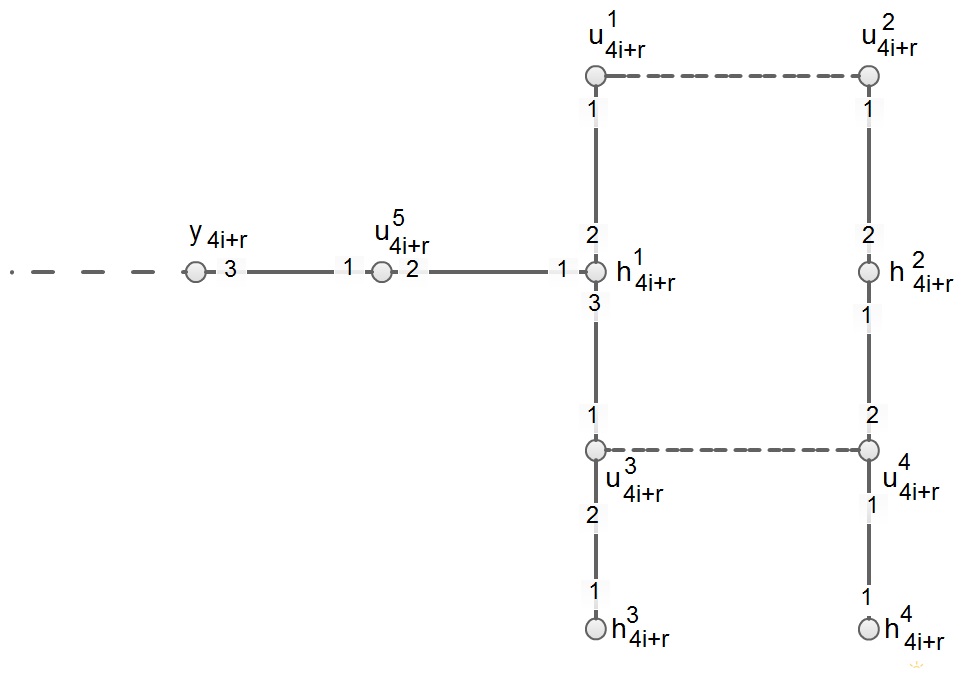}
\]
\caption{Diagram of added agents in $I^{\prime }$, the constructed instance of \small $(3,3)$\normalsize {\sc -hrc}.}
\label{image:addedgadget}
\end{figure}

\begin{lemma}
\label{y4i+rmustbematched}
In any stable matching $M$ in $I^{\prime }$, for every $y_{4i+r} \in Y$, $M(y_{4i+r}) \in X$.  
\end{lemma}

\begin{proof}
Suppose not. Then $y_{4i+r}$ is either unmatched in $M$ or $M(y_{4i+r}) = u^5_{4i+r}$. If $y_{4i+r}$ is unmatched then $( u^5 _{4i+r}, y_{4i+r})$ blocks $M$ in $I^{\prime }$, a contradiction. Hence,  $M(y_{4i+r}) = u^5_{4i+r}$ and thus $(u^5_{4i+r}, h^1_{4i+r}) \notin M$. 

Assume that $h^1_{4i+r}$ is unmatched in $M$. Then $(u^1_{4i+r}, u^2_{4i+r})$ is unmatched in $M$ and $M((u^3_{4i+r}, u^4_{4i+r})) \neq (h^1_{4i+r}, h^4_{4i+r})$. Thus, either $(u^3_{4i+r}, u^4_{4i+r})$ is unmatched in $M$ or $M ((u^3_{4i+r}, u^4_{4i+r}) ) = (h^3_{4i+r}, h^2_{4i+r})$. If $(u^3_{4i+r}, u^4_{4i+r})$ is unmatched in $M$ then $(u^1_{4i+r}, u^2_{4i+r})$ blocks $M$ with $(h^1_{4i+r}, h^2_{4i+r})$ in $I^{\prime }$, a contradiction. Hence $(u^3_{4i+r}, u^4_{4i+r})$ must be matched with $(h^3_{4i+r}, $ $ h^2_{4i+r})$ in $M$. However, now $(u^3_{4i+r}, u^4_{4i+r})$ blocks $M$ with $(h^1_{4i+r}, h^4_{4i+r})$ in $I^{\prime}$, a contradiction. 

Now, assume $h^1_{4i+r}$ is matched in $M$. Further assume that $h^1_{4i+r}$ is matched in $M$ through the joint matching of $(u^1_{4i+r}, u^2_{4i+r})$ with $(h^1 _{4i+r},$ $ h^2 _{4i+r})$. However, in this case, $M$ is blocked by $(u^3_{4i+r}, u^4_{4i+r})$ with $(h^3 _{4i+r}, h^2 _{4i+r})$ in $I^{\prime }$, a contradiction. Thus, $h^1_{4i+r}$ must be matched in $M$ through the joint matching of $(u^3_{4i+r}, u^4_{4i+r})$ with $(h^1 _{4i+r},$ $ h^4 _{4i+r})$. However, in this case, $M$ is blocked by $(u^1_{4i+r}, u^2_{4i+r})$ with $(h^1 _{4i+r}, h^2 _{4i+r})$ in $I^{\prime }$, a contradiction.

%Hence, in any stable matching $M$ in $I^{\prime }$, $M(u^5_{4i+r}) = h^1_{4i+r}$ (further $M((u^3_{4i+r}, u^4_{4i+r})) =  (h^3 _{4i+r},$ $ h^2 _{4i+r})$). Thus, $M(y_{4i+r}) \in X$ in any stable matching $M$ in $I^{\prime }$.

\end{proof}

We now show through the following two Lemmata that if $M^{\prime }$ is a stable matching in $I^{\prime }$ and $$M = M^{\prime } \setminus \{ ( u^p_{4i+r}, h^q_{4i+r} ) : 0\leq i\leq n-1, 0\leq r\leq 3, 1\leq p\leq 5, 1\leq q\leq 5 \}$$ then $M$ is a complete stable matching in $I$, the reduced \small $(3,3)$\normalsize {\sc -hrc} instance obtained by removing all of the agents in $U\cup H$ from $I^{\prime }$.

\begin{lemma}
\label{Z-OR-PUNIONT}
No hospital in $Z\cup C$ may be unmatched and no resident in $P\cup T\cup Q$ may be unmatched in any stable matching $M$ in $I^{\prime }$.
\end{lemma}

\begin{proof}
Assume $z^1_j$ is unmatched in $M$ for some $j ~ (1\leq j \leq m)$. Thus $(p_j^s, z^2_j) \notin M$  $(4\leq s\leq 6)$ as $z^2_j$ must also be unmatched. Hence, $(z^1_j, z^2_j)$ are jointly unmatched and find $(p^1_j, p^4_j)$ acceptable. Further, $(p^1_j, p^4_j)$ jointly prefers $(z^1_j, z^2_j)$ to any other pair. Hence $(z^1_j, z^2_j)$ blocks $M$ with $(p^1_j, p^4_j)$, a contradiction. Thus, $z^1_j$ must be matched in any stable matching admitted by $I^{\prime }$. By a similar argument, the same holds for $z^2_j$.

Assume $t_j$  is unmatched in $M$ for some $j ~(1\leq j \leq m)$. If some $z^s_j ~ (3\leq s \leq 5)$ is unmatched, then $(t^s_j , z^s_j)$ blocks $M$ in $I^{\prime }$, a contradiction. Thus, $\{(p^4_j, z^3_j), (p^5_j, z^4_j), (p^6_j, z^5_j) \}\subseteq M$. Then $z^2_j$ is unmatched, a contradiction. Thus, $t_j$ must be matched in any matching admitted by $I^{\prime }$.

%Assume $t_j$  is unmatched in $M$ for some $j ~(1\leq j \leq m)$. Thus, either some $z^s_j$  $(3\leq s\leq 5)$ is unmatched or $\{(p^4_j, z^3_j), (p^5_j, z^4_j), (p^6_j, z^5_j) \}\subseteq M$. Assume some $z^s_j$  $(3\leq s\leq 5)$ is unmatched. Hence $z^s_j$  $(3\leq s\leq 5)$ forms a blocking pair with $t_j$, a contradiction. Thus, $\{(p^4_j, z^3_j), (p^5_j, z^4_j), (p^6_j, z^5_j) \}\subseteq M$. However, this implies that $(z^1_j, z^2_j)$ blocks $M$ in $I^{\prime }$  with $(p^1_j, p^4_j)$, a contradiction.  Thus, $t_j$ must be matched in any matching admitted by $I^{\prime }$.

Assume some resident $p^s_j $ is unmatched in $M$ for some $j ~ (1\leq j \leq m)$ and $s ~ (1\leq s\leq 3)$. Then $(p^s_j, p^{s+3}_j)$ is unmatched. Hence, $(p^s_j, p^{s+3}_j)$ blocks $M$ in $I^{\prime }$ with $(c^s_j, z^{s+2}_j)$, a contradiction. Thus, all $p^s_j ~ (1\leq s\leq 6)$ must be matched in any stable matching admitted by $I^{\prime }$.

Assume some $z^s_j$  $(3\leq s\leq 5)$ is unmatched in $M$. Thus either $t_j$ is unmatched in $M$ or $p^{s+1}_j$ is unmatched in $M$. As shown previously, $t_j$ cannot be unmatched so $(t_j, z^b_j) \in M$ for some $b\in \{3,4,5\} \setminus \{s\}$. Also, as shown previously $p_j ^{s+1}$ cannot be unmatched so $(p_j ^{s+1} , z^2_j)\in M$, but then $p^{b+1} _j$ is unmatched in $M$, a contradiction. Thus, all $z^s_j ~ (3\leq s\leq 5)$ must be matched in any stable matching admitted by $I^{\prime }$

Observe that no $c^s_j ~ (1\leq s\leq 3, 1\leq j\leq m)$ can be matched to $x(c^s_j )$, for otherwise $M(y_{4i+r}) \notin X$ for some $y_{4i+r} \in Y$ a contradiction to Lemma \ref{y4i+rmustbematched}. Since $z^1_j$ must be matched to some $p^s_j ~ (1\leq s\leq 3)$ and since no resident in $P$ may be unmatched, then for $s^{\prime } \in \{1,2,3 \}$ such that $s^{\prime } \neq s$, $c^{s^{\prime }} _j$ must be matched with the corresponding $p^{s^{\prime }} _j$.  Thus, $(q_j, c^s_j) \in M$ for otherwise $(q_j , c^s_j)$ blocks $M$ in $I^{\prime }$. Thus all residents in $Q$ and hospitals in $C$ must be matched in any stable matching admitted by $I^{\prime }$.
\end{proof}

\begin{lemma}
\label{XUNIONKORYUNIONL}
No hospital in $L\cup Y$ may be unmatched and no resident in $K\cup X$ may be unmatched in any stable matching $M$ in $I^{\prime }$.
\end{lemma}

\begin{proof} 
By Lemma \ref{y4i+rmustbematched}, $M(y_{4i+r}) \in X$ for all $y_{4i+r} \in Y$. Hence $M(x_{4i+r}) \in Y$ for all $x_{4i+r} \in X$. Since $(x_{4i+r} , k_{4i+r})$ are a couple for all $x_{4i+r} \in X$, it follows that $M(k_{4i+r}) \in L$ for all $k_{4i+r} \in K$ and $M(l_{4i+r}) \in K$ for all $l_{4i+r} \in L$.
\end{proof}

The proof of the previous three Lemmas allows us to now state the following more general theorem.

\begin{theorem}
\label{33HRCisNPcomplete}
Given an instance, $I^{\prime }$ of \small $(3,3)$\normalsize {\sc -hrc}, the problem of deciding whether $I^{\prime }$ admits a stable matching is NP-complete.
\end{theorem}

\begin{proof}
Let $B$ be an instance of  \sat. Construct an instance $I$ of \small $(3,3)$\normalsize {\sc -hrc} as described in the proof of Theorem \ref{33-COM-HRC} and as illustrated in Figure \ref{preflists33} and extend this instance as described above and as illustrated in Figure \ref{preflists:addedgadget} to obtain the instance $I^{\prime }$ of \small $(3,3)$\normalsize {\sc -hrc}.

Let $f$ be a satisfying truth assignment of $B$. Define a matching $M$ in $I$ as in the proof of Theorem \ref{33-COM-HRC}. Define a matching $M^{\prime }$ in $I^{\prime }$ as follows:

$$M^{\prime } = M \cup \{ (u^5_{4i+r}, h^1_{4i+r}),(u^3_{4i+r}, h^3_{4i+r}), (u^4_{4i+r}, h^2_{4i+r}) : 0\leq i \leq n-1, 0\leq r\leq 3 \} $$

As shown previously no agent in $X\cup K\cup L\cup P\cup Q\cup T\cup Z\cup C$ can block $M^{\prime }$ in $I$. By Lemma \ref{y4i+rmustbematched}, $M^{\prime }(y_{4i+r}) \in X$ for any stable matching $M^{\prime }$ in $I^{\prime }$. Thus, it must be the case that $M^{\prime }(u^5_{4i+r}) = h^1_{4i+r}$ otherwise $(u^5_{4i+r}, h^1_{4i+r})$ would block $M^{\prime }$. Further $M^{\prime }((u^3_{4i+r}, u^4_{4i+r})) =  (h^3 _{4i+r},$ $ h^2 _{4i+r})$ or $M^{\prime }$ must admit a blocking pair amongst the agents in the subinstance $S$. Thus, no agent in $Y\cup U\cup H$ can block $M^{\prime }$ in $I$. Thus $M^{\prime }$ is a stable matching in $I$.

\medskip

Conversely, suppose that $M^{\prime }$ is a stable matching in $I$. By Lemma \ref{y4i+rmustbematched}, every $y_{4i+r}$ is matched in $M^{\prime }$ to a resident in $X$. By Lemmas \ref{Z-OR-PUNIONT} and \ref{XUNIONKORYUNIONL} every agent in $K\cup X\cup P\cup T\cup Q\cup Z\cup C$ is matched in $M^{\prime }$.

Now let $$M = M^{\prime } \setminus \{ ( u^p_{4i+r}, h^q_{4i+r} ) : 0\leq i\leq n-1, 0\leq r\leq 3, 1\leq p\leq 5, 1\leq q\leq 5 \}$$

Then $M$ is a complete stable matching in $I$, the reduced \small $(3,3)$\normalsize {\sc -hrc} instance obtained by removing all of the agents in $U\cup H$ from $I^{\prime }$. By the proof of Theorem \ref{33-COM-HRC} we can obtain a satisfying truth assignment for $B$ from $M$.
\end{proof}

%For an instance $I$ of {\sc hrc} consisting of a set of residents $R$ and a set of hospitals $H$, let the set of all first members of each couple in $I$ be $R_1 \subseteq R$, and the set of second members of each couple in $I$ be $R_2 \subseteq R$. Let the set of acceptable partners of the residents in $R_1$ in $I$ be $H_1 \subseteq H$ and the set of acceptable partners of the residents in $R_2$ in $I$ be $H_2 \subseteq H$. If in $I$, $H_1 \cap H_2 = \emptyset$ and no single resident has acceptable partners in both $H_1$ and $H_2$ then we define $I$ to be an instance of {\sc hrc-dual-market} consisting of the two disjoint markets $R_1 \cup H_1$ and $R_2 \cup H_2$. 

Recall that an instance of \small $(3,3)$\normalsize {\sc -hrc-dual-market} is an instance of {\sc hrc-dual-market} in which no resident, couple or hospital has a preference list of length greater than 3. We now show that the instance described in Theorem \ref{33HRCisNPcomplete} represents a dual market and thus we are able to show that deciding whether a stbale matching exists in an instance of \small $(3,3)$\normalsize {\sc -hrc-dual-market} is also NP-complete.

\begin{corollary}
\label{33comdmml}
Given an instance of \small $(3,3)$\normalsize {\sc -hrc-dual-market}, the problem of deciding whether a stable matching exists is NP-complete. The result holds even if each hospital has capacity 1 and the preference list of each single resident, couple and hospital is derived from a strictly ordered master list of hospitals, pairs of hospitals and residents respectively.
\end{corollary}
\begin{proof}

Let $I^{\prime }$ be the instance of \small $(3,3)$\normalsize {\sc -hrc} constructed in the proof of Theorem \ref{33HRCisNPcomplete}. The residents in $I^{\prime }$ can be partitioned into two disjoint sets, $R_1 = X\cup P_1\cup Q\cup U_1$ where $P_1 = \{ p^s_j : 1\leq s\leq 3 \}$ and $U_1 = \{ u^s_{4i+r} : s \in \{1, 3, 5\} \}$  and $R_2 = K \cup P_2\cup T\cup U_2$ where $P_2 =  \{ p^s_j : 4\leq s\leq 6 \}$ and $U_2 = \{ u^s_{4i+r} : s \in \{2, 4\} \}$. Further, the hospitals in $I^{\prime }$ may also be partitioned into two disjoint sets,  $A_1 = Y\cup Z_1\cup C\cup H_1$ where $Z_1 = \{ z^1_j : 1\leq j\leq m \}$ and $H_1 = \{ h^s_{4i+r} : 0\leq i\leq n-1, 0 \leq r\leq 3, s \in \{ 1,3 \} \}$  and $A_2 = L\cup Z_2 \cup H_2$ where $Z_2 = \{ z^s_j : 2\leq s\leq 5 \wedge 1\leq j \leq m \}$ and $H_2 = \{ h^s_{4i+r} : 0\leq i\leq n-1, 0 \leq r\leq 3, s \in \{ 2,4 \} \}$.

A resident $r\in R_i$ finds acceptable only those hospitals in $A_i$ and a hospital $h\in A_i$ finds acceptable only those residents in $R_i$. From this construction it can be seen that the instance $I^{\prime }$ represents a dual market consisting of two disjoint markets, ($R_1\cup A_1$) and ($R_2\cup A_2$).

The master lists shown in Figures \ref{masterlistrescouples},  \ref{masterlistresidents} and \ref{masterlisthospitals} indicate that the preference list of each single resident, couple and hospital may be derived from a master list of hospital pairs, residents and hospitals respectively. Thus, the result follows immediately from Theorem \ref{33HRCisNPcomplete}.
\end{proof}
\begin{landscape}
\begin{figure}
\[
\begin{array}{cll}
L^1_i :
(y_{4i}, l_{4i}), 
(c(x_{4i}), l_{4i+1}),

(y_{4i+1}, l_{4i+1}), 
(c(x_{4i+1}), l_{4i+2}), 
(c(x_{4i+3}), l_{4i+3}), 
(y_{4i+3}, l_{4i+3}),
(c(x_{4i+2}), l_{4i+2}),
(y_{4i+2}, l_{4i+2}) & (0 \leq i \leq n-1)\\
\\
L^2_{i,r} :
(h^1_{4i+r}, h^2_{4i+r}),
(h^1_{4i+r}, h^4_{4i+r}),
(h^3_{4i+r}, h^2_{4i+r}) ~~~~~ (0 \leq i \leq n-1, 0 \leq r \leq 3)

~~~~~

L^3_j : 
(z^1_j, z^2_j),  
(c^1_j, z^3_j),
(c^2_j, z^4_j),
(c^3_j, z^5_j) & (1 \leq j \leq m)
\\
\\
\mbox{Master List} : L^1_0 ~ L^1_1\dots L^1_{n-1}~ L^2_{0,0} ~ L^2_{0,1} \ldots L^2_{0,3} ~ L^2_{1,0} ~ L^2_{1,1} \ldots L^2_{(n-1),3} ~  L^3_1  L^3_2\dots L^3_m

\end{array}
\]
\caption{Master list of preferences for resident couples in \small $(3,3)$\normalsize {\sc -com-hrc} instance $I$.}
\label{masterlistrescouples}
\end{figure}

\begin{figure}
\[
\begin{array}{cll}

L^4_j :
c^1_j,
c^2_j,
c^3_j,
z^3_j,
z^4_j,
z^5_j

~~~~~ (1 \leq j \leq m)

~~~~~

L^5_{i,r} :

y_{4i+r},
h^1_{4i+r} & (0 \leq i \leq n-1, 0 \leq r \leq 3)
\\
\\
\mbox{Master List} : L^4_1 L^4_2\dots L^4_m L^5_{0,0} ~ L^5_{0,1} \ldots L^5_{0,3} ~ L^5_{1,0} ~ L^5_{1,1} \ldots L^5_{(n-1),3}
\end{array}
\]
\caption{Master list of preferences for single residents in \small $(3,3)$\normalsize {\sc -com-hrc} instance $I$.}
\label{masterlistresidents}
\end{figure}

\begin{figure}
\[
\begin{array}{cll}

L^6_i : 
x_{4i+1},
x_{4i},
x_{4i+2},
x_{4i+3} ~~~~~ (0 \leq i \leq n-1)

~~~~~

L^7_i : 
k_{4i+3},
k_{4i},
k_{4i+2},
k_{4i+1} ~~~~~ (0 \leq i \leq n-1)\\

\\
L^8_j : 
p^1_j,
p^2_j,
p^3_j 

~~~~~ (0 \leq j \leq m)
~~~~~

L^9_j :
p^6_j,
p^5_j,
p^4_j 

~~~~~ (0 \leq j \leq m)

\\
\\

L^{10}_j : q_1, q_2\dots q_j ~~~~~ (0 \leq j \leq m)

~~~~~~

L^{11}_j : t_1, t_2\dots t_j  ~~~~~ (0 \leq j \leq m)\\

\\ 
L^{12}_{i,r} :
u^5_{4i+r},
u^1_{4i+r},
u^3_{4i+r},
u^4_{4i+r},
u^2_{4i+r} ~~~~~ (0 \leq i \leq n-1, 0 \leq r \leq 3)
\\
\\
\mbox{Master List} : L^8_1 ~ L^8_2 \dots L^8_m  L^6_0  L^6_1\dots L^6_{n-1}  L^7_0 L^7_1\dots L^7_{n-1}  L^9_1 ~ L^9_2\dots L^9_m ~ L^{10}_1\ldots L^{10}_m ~  L^{11}_1\ldots L^{11}_m ~ L^{12}_{0,0} ~ L^{12}_{0,1} \ldots L^{12}_{0,3} ~L^{12}_{1,0} ~ L^{12}_{1,1} \ldots L^{12}_{1,3} ~ L^{12}_{2,0} \ldots L^{12}_{(n-1),3}
\end{array}
\]
\caption{Master list of preferences for hospitals in \small $(3,3)$\normalsize {\sc -com-hrc} instance $I$.}
\label{masterlisthospitals}
\end{figure}
\end{landscape}

\section{Integer programming models for Hospitals / Residents problem variants}

\label{section:IPModels}

In this section we present a range of IP models for {\sc hr}, {\sc hrc} and {\sc hrct}. We begin by giving an IP model for {\sc hr} in Section \ref{section:IPModelsHR}. This model is then extended in Section \ref{section:IPModelsHRC} to the {\sc hrc} context. We further provide a worked example in Section \ref{section:IPModelsExample} with a view to demonstrating how the IP model for {\sc hrc} may be constructed from a small example instance. We then extend the IP model for {\sc hrc} to the more general {\sc hrct} context in Section \ref{section:IPModelsHRCT}. Detailed proofs of the correctness of both the {\sc hr}, {\sc hrc} and {\sc hrct} models are also presented in the relevant sections. 
%The text in bold before the definition of a constraint shows the section of the MM-stability definition with which the constraint corresponds. Hence, a constraint preceded by '\textbf{Stability 1}' is intended to prevent blocking pairs described by part 1 of the MM-stability definition shown in Definition \ref{stability:MM} in Section \ref{section:introduction}.

\def\myline{
 
    %\vspace{-0.9em}
    \line(0,1){3.5}
    \hspace{-0.71em}
     
     \line(1,0){50}
 			\line(0,1){3.5}

}

\def\Xline{
 
    %\vspace{-0.9em}
    \line(0,1){3.5}
    \hspace{-0.71em}
     \line(1,0){17.5}
     
     ~X
  
     \line(1,0){17.5}
 			\line(0,1){3.5}
}

\def\X2line{
 
    %\vspace{-0.9em}
    \line(0,1){3.5}
    \hspace{-0.71em}
     \line(1,0){9.5}
     
     \hspace{-0.3em}
     X
     
    \hspace{-0.55em}
     \line(1,0){12.5}

     \hspace{-0.30em}
     X
  
  	\hspace{-0.51em}
     \line(1,0){8.5}
 			\line(0,1){3.5}
}

\subsection{An IP formulation for the {\sc hr}}
\label{section:IPModelsHR}

\subsubsection{Variables in the IP model for {\sc hr}}
\label{section:IPModelsHRVariables}

The IP model is designed around a series of linear inequalities that establish the absence of blocking pairs. The variables are defined for each resident and for each element on his/her preference list (with the possibility of being unmatched).

Let $I$ be an instance of {\sc hr} with residents $R = \{r_1, r_2,\dots, r_{n_1}\}$ and hospitals $H= \{h_1, h_2,\dots$ $, h_{n_2}\}$ where each resident $r_i\in R$, has a preference list of length $l(r_i)$ consisting of individual hospitals $h_j\in H$. Each hospital $h_j\in H$ has a preference list of individual residents $r_i\in R$ of length $l(h_j)$. Further, each hospital $h_j$ has a capacity $c_j$ representing the number of available posts it has to match with residents.

Let  $J$ be the following Integer Programming (IP) formulation of $I$. In $J$, for each $i ~ (1\leq i\leq n_1)$ and $p ~(1\leq p\leq l(r_i))$, define a variable $x_{i,p}$  such that

\[ x_{i,p} = \left\{ \begin{array}{ll}
         1 & \mbox{if $r_i$ is assigned to their $p^{th}$ choice hospital}\\
        0 & \mbox{otherwise}\end{array} \right. \] 

For $p=l(r_i)+1$ define a variable $x_{i,p}$ whose intuitive meaning is that resident $r_i$ is unassigned. Therefore we also have 

\[ x_{i,l(r_i)+1} = \left\{ \begin{array}{ll}
         1 & \mbox{if $r_i$ is unassigned}\\
        0 & \mbox{otherwise}\end{array} \right. \]

Let $X = \{ x_{i,p} : 1\leq i \leq n_1 \wedge 1\leq p \leq l(r_i) + 1 \}$. Let $pref(r_i, p)$ denote the hospital at position $p$ of $r_i$'s preference list where $1\leq i\leq n_1$ and $1\leq p\leq l(r_i)$. Further for an acceptable resident-hospital pair $(r_i, h_j)$, let $rank(h_j, r_i) =q$ denote the rank which hospital $h_j$ assigns resident $r_i$, where $1\leq j\leq n_2$, $1\leq i\leq n_1$ and $1\leq q \leq l(h_j)$.

%Let the preference list of $h_j$ for $(1 \leq j \leq n_2)$ be ordered in non-decreasing order of rank, such that if $rank(h_j, r_u) < rank(h_j, r_v)$ then $r_u$ precedes $r_v$ in $h_j$'s preference list. If $rank(h_j, r_u) = rank(h_j, r_v)$ then $r_u$ and $r_v$ may be arbitrarily ordered with respect to each other in the preference list of $h_j$.

\subsubsection{Constraints in the IP model for {\sc hr}}
\label{section:IPModelsHRConstraints}

The following constraint simply confirms that each variable $x_{i,p}$ must be binary valued for all $i ~(1\leq i\leq n_1)$ and $p ~ (1\leq p\leq l(r_i)+1)$: 

\begin{equation} \label{constraint:HR2_1} \displaystyle x_{i,p} \in \{0,1\} \end{equation}

As each resident $r_i\in R$ is either assigned to a single hospital or is unassigned,  we introduce the following constraint for all $i ~(1\leq i\leq n_1)$:

\begin{equation} \label{constraint:HR2_2} \displaystyle \sum\limits_{p=1}^{l(r_i)+1} x_{i,p} = 1 \end{equation}

Since a hospital $h_j$ may be assigned to at most $c_j$ residents, $x_{i, p} = 1$ where $pref(r_i,p) = h_j$ for at most $c_j$ residents. We thus obtain the following constraint for all $j ~ (1\leq j\leq n_2)$:

\begin{equation} \label{constraint:HR2_3} \displaystyle \sum\limits_{i=1}^{n_1} \sum\limits_{p=1}^{l(r_i)}  \{x_{i,p} \in X :  pref(r_i, p)=h_j \} \leq c_j \end{equation}

In a stable matching $M$ in $I$, if a single resident $r_i\in R$ has a worse partner than some hospital $h_j\in H$ where $pref(r_i, p)=h_j$ and $rank(h_j, r_i) =q$ then $h_j$ must be fully subscribed with better partners than $r_i$. Therefore, either $\sum\limits_{p^{\prime}=p+1}^{l(r_i)+1} x_{i,p^{\prime}}= 0$ or $h_j$ is fully subscribed with better partners than $r_i$ and $\sum\limits_{q^{\prime}=1}^{q-1} \{x_{i^{\prime },p^{\prime \prime}} :   rank(h_j,r_{i^{\prime }}) = q^{\prime } \wedge pref(r_{i^{\prime }}, p^{\prime \prime}) = h_j) \} = c_j$.

Thus, for each $i ~ (1\leq i\leq n_1)$ and $p ~ (1\leq p\leq l(r_i))$ we obtain the following constraint where $pref(r_i, p) = h_j$ and $rank(h_j, r_i)=q$:

\begin{equation} \label{constraint:HR2_4} \displaystyle c_j \sum\limits_{p^{\prime}=p+1}^{l(r_i)+1} x_{i,p^{\prime}} \leq \sum\limits_{q^{\prime}=1}^{q-1} \{x_{i^{\prime },p^{\prime \prime}} : rank(h_j, r_{i^{\prime }}) = q^{\prime } \wedge  pref(r_{i^{\prime }, p^{\prime \prime} })= h_j) \} \end{equation}

Objective Function - A maximum cardinalilty matching $M$ in $I$ is a stable matching in which the largest number of residents is matched amongst all of the stable matchings admitted by $I$. To maximise the size of the stable matching found we apply the following objective function:

\begin{equation} \label{constraint:HR2_5} \displaystyle \max \sum\limits_{i=1}^{n} \sum\limits_{p=1}^{l(r_i)} x_{i,p} \end{equation}

\subsubsection{Proof of correctness the IP model for HR}

\begin{theorem}\label{IPHR2}
Given an instance $I$ of {\sc hr}, let $J$ be the corresponding IP model as defined in Section \ref{section:IPModelsHRVariables} and Section {section:IPModelsHRConstraints}. A stable matching in $I$ is exactly equivalent to a feasible solution to $J$.
\end{theorem}

\begin{proof}

Consider a stable matching $M$ in $I$. From $M$ we form an assignment of values to the variables \textbf{x} as follows. 

Initially set $x_{i,p}=0$ for all $i ~ (1\leq i\leq n_1)$ and $p ~ (1\leq p\leq l(r_i)+1)$. Then for each $(r_i, h_j)\in M$ set $x_{i,p}=1$, where $h_j =pref(r_i, p)$. If $r_i$ is unassigned then set $x_{i, l(r_i)+1}=1$. As each resident has a single partner or is unassigned (but not both), for a given $i ~ (1\leq i\leq n_1)$, for exactly one value of $p$ in the range $1\leq p\leq r(i)+1$, $x_{i,p}=1$, and for each other value of $p$ in the same range, $x_{i,p}=0$, and therefore Constraint \ref{constraint:HR2_2} holds in the assignment derived from $M$. Since each hospital is assigned in $M$ to at most $c_j$ acceptable residents, Constraint \ref{constraint:HR2_3} also holds in the assignment derived from $M$.

Let $(r_i, h_j)$ be an acceptable pair not in $M$. For $(r_i, h_j)$ to block, $r_i$ must have a partner worse than rank $p$ while simultaneously $h_j$ is either under-subscribed or has a partner worse than rank $q$, where $pref(r_i, p) = h_j$ and $rank(h_j, r_i)=q$. If $r_i$ does not have a partner worse than $h_j$ then $c_j \sum\limits_{p^{\prime}=p+1}^{l(r_i)+1} x_{i,p^{\prime}}=0$ and otherwise $c_j \sum\limits_{p^{\prime}=p+1}^{l(r_i)+1} x_{i,p^{\prime}}=c_j$. If $h_j$ has $c_j$ partners better than $r_i$ then $\sum\limits_{q^{\prime}=1}^{q-1} \{x_{i^{\prime },p^{\prime \prime}} :  rank(h_j, r_{i^{\prime }}) = q^{\prime } \wedge  pref(r_{i^{\prime }, p^{\prime \prime} })= h_j) \} = c_j$ otherwise $\sum\limits_{q^{\prime}=1}^{q-1} \{x_{i^{\prime },p^{\prime \prime}} :  rank(h_j, r_{i^{\prime }}) = q^{\prime } \wedge  pref(r_{i^{\prime }, p^{\prime \prime} })= h_j) \} < c_j$. 

Suppose $c_j \sum\limits_{p^{\prime}=p+1}^{l(r_i)+1} x_{i,p^{\prime}} =c_j$. Then $r_i$ is unassigned or has a worse partner than $h_j$ in $M$. Thus, by the stability of $M$, $h_j$ is full and all of $h_j$'s partners are better than $r_i$. Hence $\sum\limits_{q^{\prime}=1}^{q-1} \{x_{i^{\prime },p^{\prime \prime}} :rank(h_j, r_{i^{\prime }}) = q^{\prime } \wedge  pref(r_{i^{\prime }}, p^{\prime \prime} )= h_j) \} = c_j$ and Constraint \ref{constraint:HR2_4} is satisfied by the assignment derived from $M$.

Now, suppose $c_j \sum\limits_{p^{\prime}=p+1}^{l(r_i)+1} x_{i,p^{\prime}} =0 $. Then $r_i$ has a better partner than $h_j$ in $M$. Since $\sum\limits_{q^{\prime}=1}^{q-1} \{x_{i^{\prime },p^{\prime \prime}} :  rank(h_j, r_{i^{\prime }}) = q^{\prime } \wedge  pref(r_{i^{\prime }}, p^{\prime \prime} )= h_j) \} \geq 0$, Constraint \ref{constraint:HR2_4} is satisfied by the assignment derived from $M$.

As all of the constraints in $J$ hold for an assignment derived from a stable matching $M$, a stable matching $M$ in $I$ represents a feasible solution to $J$.

Conversely, consider a feasible  solution, $\langle \textbf{x}, \textbf{y} \rangle$,  to $J$. From such a solution we form a set of pairs, $M$, as follows.

Initially let $M = \emptyset$. For each $i$   $(1\leq i\leq n_1)$ and $p$  $(1\leq p\leq l(r_i))$ if $x_{i,p}=1$ then add $(r_i, h_j)$ to $M$ where $h_j= pref(r_i, p)$. As $\langle \textbf{x} \rangle$ satisfies constraints \ref{constraint:HR2_1}, \ref{constraint:HR2_2} and \ref{constraint:HR2_3}, each resident in $M$ must have exactly one partner or be unassigned (but not both) and each hospital $h_j$ in $M$ must have at most $c_j$ partners. Therefore the set of pairs $M$ created from the solution $\langle \textbf{x} \rangle$ to $J$, is a matching in $I$.

We now show that $M$ is stable. Assume $(r_i, h_j)$ blocks $M$. Let $pref(r_i,p) = h_j$ and $rank(h_j, r_i) =q$. Therefore, $r_i$ is unassigned or, has a worse partner than $h_j$ and $h_j$ is under-subscribed or prefers $r_i$ to some member of $M(h_j)$. However, this implies that in $J$,  $c_j \sum\limits_{p^{\prime}=p+1}^{l(r_i)+1} x_{i,p^{\prime}} = c_j$ and $\sum\limits_{q^{\prime}=1}^{q-1} \{x_{i^{\prime },p^{\prime \prime}} : rank(h_j, r_{i^{\prime }}) = q^{\prime } \wedge  pref(r_{i^{\prime }, p^{\prime \prime} })= h_j) \} < c_j$ and therefore Constraint \ref{constraint:HR2_4} is not satisfied in $J$, a contradiction. Therefore no such $(r_i, h_j)$ can block $M$.
\end{proof}

\subsection{An IP formulation for {\sc hrc}}

\label{section:IPModelsHRC}

\subsubsection{Variables in the IP model for {\sc hrc}}

\label{section:IPModelsHRCVariables}

The IP model extends the model for {\sc hr} presented in Section \ref{section:IPModelsHR}. This extended model is designed around a series of linear inequalities that establish the absence of blocking pairs according to each of the different parts of Definition \ref{stability:MM}. The variables are defined for each resident, whether single or a member of a couple, and for each element on his/her preference list (with the possibility of being unmatched). A further consistency constraint ensures that each member of a couple obtains hospitals from the same pair in their list, if assigned. Finally, the objective of the IP is to maximise the size of a stable matching, if one exists. The model presented is more complex than existing IP formulations in the literature for stable matching problems \cite{VV89, Rot86, Pod10, KM14} simply because of the number of blocking pair cases in Definition \ref{stability:MM} required to adequately take account of couples.

Let $I$ be an instance of {\sc hrc} with residents $R = \{r_1, r_2,\dots, r_{n_1}\}$ and hospitals $H= \{h_1, h_2,\dots, h_{n_2}\}$. Without loss of generality, suppose residents $r_1, r_2\ldots r_{2c}$ are in couples. Again, without loss of generality, suppose that the couples are $(r_{2i-1}, r_{2i})$  $(1\leq i\leq c)$.  Suppose that the joint preference list of  a couple $c_i =  (r_{2i-1}, r_{2i})$ is:

$$c_i ~ : ~ (h_{\alpha _1}, h_{\beta _1}),(h_{\alpha _2}, h_{\beta _2})\ldots (h_{\alpha _l}, h_{\beta _l})$$ 
From this list we create the following projected preference list for resident $r_{2i-1}$: 
$$r_{2i-1} ~ : ~ h_{\alpha _1}, h_{\alpha _2}\ldots h_{\alpha _l}$$ 
and the following projected preference list for resident $r_{2i}$: 
$$r_{2i} ~ :  ~ h_{\beta _1}, h_{\beta _2}\ldots h_{\beta _l}$$

Clearly, the projected preference list of the residents $r_{2i-1}$ and $r_{2i}$ are the same length as the preference list of the couple $c_i = (r_{2i-1}, r_{2i})$. Let $l(c_i)$ denote the lengths of the preference list of $c_i$ and let $l(r_{2i-1})$ and $l(r_{2i})$ denote the lengths of the projected preference lists of $r_{2i-1}$ and $r_{2i}$ respectively. Then $l(r_{2i-1}) = l(r_{2i}) = l(c_i)$. A given hospital $h_j$ may appear more than once in the projected preference list of a linked resident in a couple $c_i = (r_{2i-1}, r_{2i})$.

Let the single residents be $r_{2c+1}, r_{2c+2}\ldots r_{n_1}$, where each single resident $r_i$, has a preference list of length $l(r_i)$ consisting of individual hospitals $h_j\in H$.

Each hospital $h_j\in H$ has a preference list of individual residents $r_i\in R$ of length $l(h_j)$. Further, each hospital $h_j\in H$ has capacity $c_j \geq 1$, the number of residents with which it may match.

Let  $J$ be the following Integer Programming (IP) formulation of $I$. In $J$, for each $i ~ (1\leq i\leq n_1)$ and $p ~(1\leq p\leq l(r_i))$, define a variable $x_{i,p}$  such that

\[ x_{i,p} = \left\{ \begin{array}{ll}
         1 & \mbox{if $r_i$ is assigned to their $p^{th}$ choice hospital}\\
        0 & \mbox{otherwise}\end{array} \right. \] 

For $p=l(r_i)+1$ define a variable $x_{i,p}$ whose intuitive meaning is that resident $r_i$ is unassigned. Therefore we also have 

\[ x_{i,l(r_i)+1} = \left\{ \begin{array}{ll}
         1 & \mbox{if $r_i$ is unassigned}\\
        0 & \mbox{otherwise}\end{array} \right. \]

Let $X = \{ x_{i,p} : 1\leq i \leq n_1 \wedge 1\leq p \leq l(r_i) + 1 \}$. Let $pref(r_i, p)$ denote the hospital at position $p$ of a single resident $r_i$'s preference list or on the projected preference list of a resident belonging to a couple where $1\leq i\leq n_1$ and $1\leq p\leq l(r_i)$. %Further let $pref(h_j, q)$ denote the resident at position $q$ of $h_j$'s preference list where $1\leq j\leq n_2$ and $1\leq q\leq l(h_j)$.% 
Let $pref( ( r_{2i}, r_{2i-1}), p)$ denote the hospital pair at position $p$ on the joint preference list of $(r_{2i-1}, r_{2i})$.

For an acceptable resident-hospital pair $(r_i, h_j)$, let $rank(h_j, r_i) = q$ denote the rank which hospital $h_j$ assigns resident $r_i$ where $1\leq j\leq n_2$, $1\leq i\leq n_1$ and $1\leq q \leq l(h_j)$. Thus, $rank(h_j, r_i)$ is equal to the number of residents that $h_j$ prefers to $r_i$ plus one.

%Let the preference list of $h_j$ for $(1 \leq j \leq n_2)$ be ordered in non-decreasing order of rank, such that if $rank(h_j, r_u) < rank(h_j, r_v)$ then $r_u$ precedes $r_v$ in $h_j$'s preference list. If $rank(h_j, r_u) = rank(h_j, r_v)$ then $r_u$ and $r_v$ may be arbitrarily ordered with respect to each other in the preference list of $h_j$.

Further, for $i ~ (1\leq i\leq n_1)$, $j ~ (1\leq j\leq n_2)$, $p ~ (1\leq p\leq l(r_i))$ and $q ~ (1\leq q\leq l(h_j))$ let the set $R(h_j, q)$ contain resident integer pairs $(r_i, p)$ such that $rank(h_j, r_i) = q$ and $pref(r_i, p)=h_j$. Hence: 

%$$R(h_j, q) = \{ (r_i, p)\in R \times \mathbb{Z} :  rank(h_j, r_i) = q \wedge p \in \{ p^{\prime } : 1\leq p^{\prime } \leq l(r_{i}) \wedge pref(r_i, p^{\prime }) = h_j \} \}$$

$$R(h_j, q) = \{ (r_i, p)\in R \times \mathbb{Z} :  rank(h_j, r_i) = q \wedge 1\leq p\leq l(r_i )\wedge pref(r_i, p) = h_j \}$$

Intuitively, the set $R(h_j, q)$ contains the resident-position pairs $(r_i, p)$ such that $r_i$ is assigned a rank of $q ~ (1\leq q \leq l(h_j))$ by $h_j$ and $h_j$ is in position $p ~ (1\leq p \leq l(r_i))$ on $r_i$'s preference list. 

%\footnote{In the HRC context described here the set $R(h_j, q)$ will always have cardinality of 1 for $j ~ (1\leq j \leq n_2)$ and $q ~ (1\leq q \leq l(h_j))$, however when ties are allowed in the hospital preference lists in the HRCT context it need not be the case that $R(h_j, q)$ has cardinality of 1.}

When considering the exact nature of a blocking pair in the description that follows, the stability definition due to Manlove and McDermid \cite{MM10} (MM-stability) is applied in all cases. The text in bold before the definition of a constraint shows the section of the MM-stability definition with which the constraint corresponds. Hence, a constraint preceded by `\textbf{Stability 1}' is intended to prevent blocking pairs described by part 1 of the MM-stability definition shown in Definition \ref{stability:MM} in Section \ref{section:introduction}.

\subsubsection{Constraints in the IP model for {\sc hrc}}

\label{section:IPModelsHRCConstraints}

The following constraint simply confirms that each variable $x_{i,p}$ must be binary valued for all $i ~(1\leq i\leq n_1)$ and $p ~ (1\leq p\leq l(r_i)+1)$: 

\begin{equation} \label{constraint:HRC2_1} \displaystyle x_{i,p} \in \{0,1\} \end{equation}

Similarly, the following constraint confirms that each variable $\alpha _{j, q}$ must be binary valued for all $j ~(1\leq j\leq n_2)$ and $q ~(1\leq q\leq l(h_j))$:

\begin{equation} \displaystyle \alpha _{j,q} \in \{0,1\} \end{equation}

Also, the following constraint confirms that each variable $\beta _{j, q}$ must be binary valued for all $j ~(1\leq j\leq n_2)$ and $q ~(1\leq q\leq l(h_j))$:

\begin{equation} \displaystyle \beta _{j,q} \in \{0,1\} \end{equation}

As each resident $r_i\in R$ is either assigned to a single hospital or is unassigned,  we introduce the following constraint for all $i ~(1\leq i\leq n_1)$:

\begin{equation} \label{constraint:HRC2_2} \displaystyle \sum\limits_{p=1}^{l(r_i)+1} x_{i,p} = 1 \end{equation}

Since a hospital $h_j$ may match with at most $c_j$ residents, $x_{i, p} = 1$ where $pref(r_i,p) = h_j$ for at most $c_j$ residents. We thus obtain the following constraint for all $j ~ (1\leq j\leq n_2)$:

\begin{equation} \label{constraint:HRC2_3} \displaystyle \sum\limits_{i=1}^{n_1} \sum\limits_{p=1}^{l(r_i)} \{x_{i,p} \in X :  pref(r_i, p)=h_j \} \leq c_j \end{equation}

For each couple $(r_{2i-1}, r_{2i})$, if resident $r_{2i-1}$ is assigned to the hospital in position $p$ in their projected preference list then $r_{2i}$ must also be assigned to the hospital in position $p$ in their projected preference list. We thus obtain the following constraint for all $1\leq i\leq c$ and $1\leq p\leq l(r_{2i-1})+1$:

\begin{equation} \label{constraint:HRC2_4} \displaystyle x_{2i-1,p} = x_{2i,p} \end{equation}

\textbf{Stability 1} - In a stable matching $M$ in $I$, if a single resident $r_i\in R$ has a worse partner than some hospital $h_j\in H$ where $pref(r_i, p)=h_j$ and $rank(h_j, r_i) =q$ then $h_j$ must be fully subscribed with better partners than $r_i$. Therefore, either $\sum\limits_{p^{\prime}=p+1}^{l(r_i)+1} x_{i,p^{\prime}}= 0$ or $h_j$ is fully subscribed with better partners than $r_i$ and $\sum\limits_{q^{\prime}=1}^{q-1} \{x_{i^{\prime },p^{\prime \prime}} \in X :   ( r_{i^{\prime }, p^{\prime \prime}}) \in R(h_{j}, q^{\prime }) \} = c_j$.

Thus, for each $i ~ (2c+1\leq i\leq n_1)$ and $p ~ (1\leq p\leq l(r_i))$ we obtain the following constraint where $pref(r_i, p) = h_j$ and $rank(h_j, r_i)=q$:

%\begin{equation} \label{constraint:HRC2_5} \displaystyle c_j \sum\limits_{p^{\prime}=p+1}^{l(r_i)+1} x_{i,p^{\prime}} \leq \sum\limits_{q^{\prime}=1}^{q-1} \{x_{i^{\prime },p^{\prime \prime}} \in X :   rank(h_j,r_{i^{\prime }}) = q^{\prime } \wedge pref(r_{i^{\prime }}, p^{\prime \prime}) = h_j) \} \end{equation}

\begin{equation} \label{constraint:HRC2_5} \displaystyle c_j \sum\limits_{p^{\prime}=p+1}^{l(r_i)+1} x_{i,p^{\prime}} \leq \sum\limits_{q^{\prime}=1}^{q-1} \{x_{i^{\prime },p^{\prime \prime}} \in X :   ( r_{i^{\prime }, p^{\prime \prime}}) \in R(h_{j}, q^{\prime }) \} \end{equation}

\textbf{Stability 2(a)} - In a stable matching $M$ in $I$, if a couple $c_i=(r_{2i-1}, r_{2i})$ prefers hospital pair $(h_{j_1}, h_{j_2})$ (which is at position $p_1$ on $c_i$'s preference list) to $( M(r_{2i-1}), M(r_{2i}) )$ (which is at position $p_2$) then it must not be the case that, if $h_{j_2} = M(r_{2i})$ then $h_{j_1}$ is under-subscribed or prefers $r_{2i-1}$ to one of its partners in $M$. In the special case in which $pref(r_{2i-1}, p_1)= pref(r_{2i}, p_1) = h_{j_1}$ it must not be the case that, if $h_{j_1} = h_{j_2} = M(r_{2i})$ then $h_{j_1}$ is under-subscribed or prefers $r_{2i-1}$ to one of its partners in $M$ other than $r_{2i}$.

Thus, for the general case, we obtain the following constraint for all $i ~ (1\leq i\leq c)$ and $p_1, p_2$ $(1\leq p_1 < p_2 \leq l(r_{2i-1}))$ such that $pref(r_{2i}, p_1) = pref(r_{2i}, p_2)$  and $rank(h_{j_1}, r_{2i-1})=q$: 

\begin{equation} \label{constraint:HRC2_6} \displaystyle c_{j_1} x_{2i, p_2} \leq \sum\limits_{q^{\prime }=1}^{q-1} \{ x_{i^{\prime },p^{\prime \prime}} \in X: ( r_{i^{\prime }, p^{\prime \prime}}) \in R(h_{j_1}, q^{\prime })\} \end{equation}

However, for the special case in which $pref(r_{2i-1}, p_1)= pref(r_{2i}, p_1) = h_{j_1}$ we obtain the following constraint for all $i ~ (1\leq i\leq c)$ and $p_1, p_2$ where $(1\leq p_1 < p_2 \leq l(r_{2i-1}))$ such that $pref(r_{2i}, p_1) = pref(r_{2i}, p_2)$  and $rank(h_{j_1}, r_{2i-1})=q$:

\begin{equation} \label{constraint:HRC2_7} \displaystyle ( c_{j_1} - 1) x_{2i, p_2} \leq \sum\limits_{q^{\prime }=1}^{q-1} \{ x_{i^{\prime },p^{\prime \prime}} \in X : q^{\prime } \neq rank(h_{j_1}, r_{2i}) \wedge ( r_{i^{\prime }, p^{\prime \prime}}) \in R(h_{j_1}, q^{\prime })\} \end{equation}

\textbf{Stability 2(b)} - A similar constraint is required for the odd members of each couple. 

Thus, for the general case, we obtain the following constraint for all $i ~ (1\leq i\leq c)$ and $p_1, p_2$ where $(1\leq p_1 < p_2 \leq l(r_{2i}))$ such that $pref(r_{2i-1}, p_1) = pref(r_{2i-1}, p_2)$  and  $rank(h_{j_2}, r_{2i})=q$:

\begin{equation} \label{constraint:HRC2_8} \displaystyle c_{j_2} x_{2i-1, p_2} \in X \leq \sum\limits_{q^{\prime }=1}^{q-1} \{ x_{i^{\prime },p^{\prime \prime}} : ( r_{i^{\prime }, p^{\prime \prime}}) \in R(h_{j_2}, q^{\prime })\} \end{equation}

Again, for the special case in which $pref(r_{2i-1}, p_1)= pref(r_{2i}, p_1) = h_{j_2}$ we obtain the following constraint for all $i ~ (1\leq i\leq c)$ and $p_1, p_2$ where $(1\leq p_1 < p_2 \leq l(r_{2i}))$ such that $pref(r_{2i-1}, p_1) = pref(r_{2i-1}, p_2)$  and  $rank(h_{j_2}, r_{2i})=q$:

\begin{equation} \label{constraint:HRC2_9} \displaystyle ( c_{j_1} - 1) x_{2i-1, p_2} \leq \sum\limits_{q^{\prime }=1}^{q-1} \{ x_{i^{\prime },p^{\prime \prime}} \in X : q^{\prime } \neq rank(h_{j_2}, r_{2i-1}) \wedge ( r_{i^{\prime }, p^{\prime \prime}}) \in R(h_{j_2}, q^{\prime })\} \end{equation}

For all $j ~(1\leq j\leq n_2)$ and $q ~(1\leq q\leq l(h_j))$ define a new constraint such that:

\begin{equation} \label{definition:alpha} \displaystyle \alpha _{j, q} \geq  1 - \dfrac { \sum\limits_{q^{\prime }=1}^{q -1} \{ x_{i^{\prime },p^{\prime \prime}} \in X : ( r_{i^{\prime }, p^{\prime \prime}}) \in R(h_{j}, q^{\prime })  \} } {c_{j} } \end{equation}

Thus, if $h_{j}$ is full with assignees better than rank  $q$ then $\alpha _{j, q}$ may take the value 0 or 1. However, if $h_{j}$ is not full with assignees better than rank $q$ then $\alpha _{j, q} = 1$.

For all $j ~(1\leq j\leq n_2)$ and $q ~(1\leq q\leq l(h_j))$ define a new constraint such that:

\begin{equation} \label{definition:beta} \displaystyle \beta _{j, q} \geq 1 - \dfrac { \sum\limits_{q^{\prime }=1}^{q -1} \{ x_{i^{\prime },p^{\prime \prime}} \in X: ( r_{i^{\prime }, p^{\prime \prime}}) \in R(h_{j}, q^{\prime })  \}  } { ( c_{j} - 1 ) } \end{equation}

Thus, if $h_{j}$ has $c_j -1$ or more assignees better than rank  $q$ then $\beta _{j, q}$ may take the value 0 or 1. However, if $h_{j}$ has less than $c_j -1$ assignees better than rank $q$ then $\beta _{j, q} = 1$.

\textbf{Stability 3(a)} - In a stable matching $M$ in $I$, if a couple $c_i=(r_{2i-1}, r_{2i})$  is assigned to a worse pair than hospital pair $(h_{j_1}, h_{j_2})$ (where $h_{j_1} \neq h_{j_2}$) it must be the case that for some $t\in \{1,2\}$, $h_{j_t}$ is full and prefers its worst assignee to $r_{2i-2+t}$.

%every post of at least one of $h_{j_1}$ or $h_{j_2}$ has a strictly better partner than $r_{2i-1}$ and $r_{2i}$ respectively. Let $rank(h_{j_1}, r_{2i-1}) = q_1$ and $rank(h_{j_2}, r_{2i}) = q_2$.

Thus we obtain the following constraint for all $ i ~ (1\leq i\leq c)$ and $p ~ (1\leq p\leq l(r_{2i-1}))$ where $h_{j_1} = pref(r_{2i-1}, p)$, $h_{j_2} = pref(r_{2i}, p)$, $h_{j_1} \neq h_{j_2}$, $rank(h_{j_1}, r_{2i-1}) =q_1$ and $rank(h_{j_2}, r_{2i}) =q_2$:

\begin{equation} \label{constraint:HRC2_10} \displaystyle \sum\limits_{p^{\prime }=p+1}^{l(r_{2i-1})+1} x_{2i-1, p^{\prime }} + \alpha _{j_1, q_1} + \alpha _{j_2, q_2} \leq 2 \end{equation}

\textbf{Stability 3(b)} - In a stable matching $M$ in $I$, if a couple $c_i=(r_{2i-1}, r_{2i})$ is assigned to a worse pair than $(h_{j}, h_{j})$  where $M(r_{2i-1})\neq h_j$ and $M(r_{2i})\neq h_j$ then $h_{j}$  must not have two or more free posts available. 

\textbf{Stability 3(c)} - In a stable matching $M$ in $I$, if a couple $c_i=(r_{2i-1}, r_{2i})$  is assigned to a worse pair than $(h_{j}, h_{j})$ where $M(r_{2i-1})\neq h_j$ and $M(r_{2i})\neq h_j$ then $h_{j}$ must not prefer at least one of $r_{2i-1}$ or $r_{2i}$ to some assignee of $h_{j}$ in $M$ while having a single free post.

Both of the preceding stability definitions may be modeled by a single constraint. Thus, we obtain the following constraint for $i ~ (1\leq i\leq c)$ and $p ~ (1\leq p\leq l(r_{2i-1}))$  such that $pref(r_{2i-1}, p) = pref(r_{2i},p)$ and $h_j=pref(r_{2i-1}, p)$ where $q = \min \{ rank(h_j, r_{2i}),$ $ rank (h_j, r_{2i-1}) \}$ :

%\begin{equation} \label{constraint:HRC2_8} \displaystyle ( c_j-1) \sum\limits_{p^{\prime }=p+1}^{l(r_{2i-1})+1}  x_{2i-1,p^{\prime }} \leq \sum\limits_{q^{\prime }=1}^{l(h_j)} \{ x_{i^{\prime },p^{\prime \prime}} : (r_{i^{\prime }}, p^{\prime \prime}) \in S(h_j, q^{\prime }) \} \end{equation}

\begin{equation} \label{constraint:HRC2_11} \displaystyle c_j \sum\limits_{p^{\prime }=p+1}^{l(r_{2i-1})+1} x_{2i-1,p^{\prime }} - \dfrac {\sum\limits_{q^{\prime }=1}^{q-1} \{ x_{i^{\prime },p^{\prime \prime}} \in X : ( r_{i^{\prime }, p^{\prime \prime}}) \in R(h_{j}, q^{\prime })\}}{ (c_j -1 )} $$ $$ \leq \sum\limits_{q^{\prime }=1}^{l(h_j)} \{ x_{i^{\prime },p^{\prime \prime}} \in X: (r_{i^{\prime }}, p^{\prime \prime}) \in R(h_j, q^{\prime }) \} \end{equation}

\textbf{Stability 3(d)} - In a stable matching $M$ in $I$, if a couple $c_i=(r_{2i-1}, r_{2i})$ is jointly assigned to a worse pair than $(h_{j}, h_{j})$ where $M(r_{2i-1})\neq h_j$ and $M(r_{2i})\neq h_j$ then $h_{j}$ must not be fully subscribed and also have two assigned partners $r_x$ and $r_y$ (where $x\neq y)$ such that $h_{j}$ strictly prefers $r_{2i-1}$ to $r_x$ and also prefers $r_{2i}$ to $r_{y}$.

For each $(h_{j}, h_{j})$ acceptable to $(r_{2i-1}, r_{2i})$, let $r_{min}$ be the better of $r_{2i-1}$ and $r_{2i}$ according to hospital $h_j$ with $rank(h_j, r_{min}) = q_{min}$. Analogously, let $r_{max}$ be the worse of $r_{2i}$ and $r_{2i-1}$ according to hospital $h_j$ with $rank(h_j, r_{max}) = q_{max}$. Thus we obtain the following constraint for $i ~ (1\leq i\leq c)$ and $p ~ (1\leq p\leq l(r_{2i-1}))$ such that $pref(r_{2i-1}, p) = pref(r_{2i},p) = h_j$.

\begin{equation} \label{constraint:HRC2_12} \displaystyle \sum\limits_{p^{\prime }=p+1}^{l(r_{2i-1})+1} x_{2i-1, p^{\prime }} + \alpha _{j, q_{max}} + \beta _{j, q_{min}} \leq 2 \end{equation}

Objective Function - A maximum cardinalilty matching $M$ in $I$ is a stable matching in which the largest number of residents is matched amongst all of the stable matchings admitted by $I$. To maximise the size of the stable matching found we apply the following objective function:

\begin{equation} \label{constraint:HRC2_13} \displaystyle \max \sum\limits_{i=1}^{n_1} \sum\limits_{p=1}^{l(r_i)} x_{i,p} \end{equation}

\subsubsection{Proof of correctness of constraints in the IP model for {\sc hrc}}

\begin{theorem}\label{IPHRC2}
Given an instance $I$ of {\sc hr}, let $J$ be the corresponding IP model as defined in Section \ref{section:IPModelsHRCVariables} and Section {section:IPModelsHRCConstraints}. A stable matching in $I$ is exactly equivalent to a feasible solution to $J$.
\end{theorem}

\begin{proof}

Consider a stable matching $M$ in $I$. From $M$ we form an assignment of values to the variables  \textbf{x}, \textbf{$ \alpha $}, and \textbf{$ \beta $} as follows. 

Initially set $x_{i,p}=0$ for all $i ~ (1\leq i\leq n_1)$ and $p ~ (1\leq p\leq l(r_i)+1)$. Then for each $(r_i, h_j)\in M$ where $r_i$ is a single resident, set $x_{i,p}=1$, where $pref(r_i, p)= h_j$. If $r_i$ is unassigned then set $x_{i, l(r_i)+1}=1$. If $r_i$ is a linked resident, assume without loss of generality that $r_i = r_{2i-1}$ (respectively $r_{2i}$) then set $x_{2i-1,p} =1$ (respectively $x_{2i,p} =1$) where $pref( (r_{2i-1}, r_{2i}), p) = (h_{j_1}, h_{j_2})$ where $h_{j_1} = M(r_{2i-1})$ and $h_{j_2} = M((r_{2i})$. If $(r_{2i-1}, r_{2i})$ is unassigned then set $x_{2i-1, l(r_{2i-1})+1}=1$ and $x_{2i, l(r_{2i})+1}=1$.

As each resident has a single partner or is unassigned (but not both), for a given $i ~ (1\leq i\leq n_1)$, for exactly one value of $p$ in the range $1\leq p\leq r(i)+1$, $x_{i,p}=1$, and for each other value of $p$ in the same range, $x_{i,p}=0$, and therefore Constraint \ref{constraint:HRC2_2} holds for \textbf{x}. Also, if $x_{i,p}=1$ then for all $p^{\prime }\neq p$ such that $pref(r_i, p) = pref(r_i, p^{\prime })$, $x_{i,p^{\prime }}=0$.

Initially set $\alpha _{j,q}=0$ for all $j ~ (1\leq j\leq n_2)$ and $q ~ (1\leq q\leq l(h_j))$. Then for each $\alpha _{j,q}$, if $$ \sum\limits_{q^{\prime }=1}^{q -1} \{ x_{i^{\prime },p^{\prime \prime}} \in X : ( r_{i^{\prime }, p^{\prime \prime}}) \in R(h_{j}, q^{\prime })  \} < c_j$$ then set $\alpha _{j,q} = 1$. Initially set $\beta _{j,q}=0$ for all $j ~ (1\leq j\leq n_2)$ and $q ~ (1\leq q\leq l(h_j))$. Then for each $\beta _{j,q}$, if $${ \sum\limits_{q^{\prime }=1}^{q -1} \{ x_{i^{\prime },p^{\prime \prime}} \in X: ( r_{i^{\prime }, p^{\prime \prime}}) \in R(h_{j}, q^{\prime })  \}  } < c_j - 1$$ then set $\beta _{j,q} = 1$.

Since, each hospital $h_j$ is assigned to at most $c_j$ acceptable residents in $M$, Constraint \ref{constraint:HRC2_3} also holds for \textbf{x}. 

For each couple $(r_{2i-1}, r_{2i})$ in $I$, let $p ~ (1\leq p \leq l(r_{2i-1}))$ be given. If $r_{2i-1}$ is assigned to $h_{j_1} = pref (r_{2i-1},p)$ in $M$ then $r_{2i}$ is assigned to $h_{j_2} = pref (r_{2i},p)$ in $M$. Similarly, for each couple $(r_{2i-1}, r_{2i})$ in $I$, if $r_{2i-1}$ is not assigned to $h_{j_1} = pref (r_{2i-1},p)$ in $M$ then $r_{2i}$ is not assigned to $h_{j_2} = pref (r_{2i},p)$ in $M$. Therefore, in the assignment derived from $M$, $x_{2i-1,p} = x_{2i,p}$ for all $i ~ (1\leq i\leq c)$ and $p ~ (1\leq p\leq l(r_{2i-1}) + 1)$ (where $l(r_{2i-1}) = l (r_{2i})$) and Constraint \ref{constraint:HRC2_4} is satisfied for \textbf{x}.

% 1.9 - stability 1

Assume \textbf{x} does not satisfy Constraint \ref{constraint:HRC2_5}. For all $i ~ (2c+1\leq i \leq n_1)$, $j ~ (1\leq j \leq n_2)$ and $p ~ (1\leq p\leq l(r_{2i-1}))$,  suppose that $(r_i, h_j)$ is an acceptable pair not in $M$ where $h_j=pref(r_i, p)$ and $rank(h_j, r_i)=q$.  If $c_j \sum\limits_{q^{\prime}=q+1}^{l(h_j)+1} x_{i,p^{\prime}} = 0$ then Constraint \ref{constraint:HRC2_5} is trivially satisfied as $\sum\limits_{q^{\prime}=1}^{q-1} \{x_{i^{\prime },p^{\prime \prime}} \in X :  ( r_{i^{\prime }, p^{\prime \prime}}) \in R(h_{j}, q^{\prime }) \} \geq 0$.

Hence, $c_j \sum\limits_{q^{\prime}=q+1}^{l(h_j)+1} x_{i,p^{\prime}} = c_j$ If $\sum\limits_{q^{\prime}=1}^{q-1} \{x_{i^{\prime },p^{\prime \prime}} \in X :  ( r_{i^{\prime }, p^{\prime \prime}}) \in R(h_{j}, q^{\prime }) \} \geq c_j$ then Constraint \ref{constraint:HRC2_5} is satisfied. Hence $\sum\limits_{q^{\prime}=1}^{q-1} \{x_{i^{\prime },p^{\prime \prime}} \in X :  ( r_{i^{\prime }, p^{\prime \prime}}) \in R(h_{j}, q^{\prime }) \} < c_j$ However, since $c_j \sum\limits_{q^{\prime}=q+1}^{l(h_j)+1} x_{i,p^{\prime}} = c_j$, $r_i$ must be unassigned or have a partner worse than $h_j$. Also, since $\sum\limits_{q^{\prime}=1}^{q-1} \{x_{i^{\prime },p^{\prime \prime}} \in X :  ( r_{i^{\prime }, p^{\prime \prime}}) \in R(h_{j}, q^{\prime }) \} < c_j$, $h_j$ is either under-subscribed or has a partner worse than $r_i$. Thus $(r_i, h_j)$ blocks $M$, a contradiction. Hence, Constraint \ref{constraint:HRC2_5} is satisfied by \textbf{x}.

% HRC2_6 constraint - stability 2a(i)

Assume \textbf{x} does not satisfy Constraint \ref{constraint:HRC2_6}. For all $x_{2i,{p_2}}$ such that $i ~ (1\leq i\leq c)$, $p_1, p_2 ~(1\leq p_1 < p_2 \leq l(r_{2i-1}))$ where $h_{j_1} = pref(r_{2i-1}, p_1)$, $pref(r_{2i}, p_1) = pref(r_{2i}, p_2) = h_{j_2}$  and $rank(h_{j_1}, r_{2i-1})=q$. If $c_{j_1} x_{2i,{p_2}} = 0$ then Constraint \ref{constraint:HRC2_6} is trivially satisfied as $\sum\limits_{q^{\prime }=1}^{q-1} \{ x_{i^{\prime },p^{\prime \prime}} \in X : ( r_{i^{\prime }, p^{\prime \prime}}) \in R(h_{j_1}, q^{\prime })\} \geq 0$. Hence, $c_{j_1} x_{2i,{p_2}} = c_{j_1}$. If  $\sum\limits_{q^{\prime }=1}^{q-1} \{ x_{i^{\prime },p^{\prime \prime}} \in X : ( r_{i^{\prime }, p^{\prime \prime}}) \in R(h_{j_1}, q^{\prime })\} \geq c_{j_1}$ then Constraint \ref{constraint:HRC2_6} is satisfied. Hence, $\sum\limits_{q^{\prime }=1}^{q-1} \{ x_{i^{\prime },p^{\prime \prime}} \in X : ( r_{i^{\prime }, p^{\prime \prime}}) \in R(h_{j_1}, q^{\prime })\} < c_{j_1}$. 

However, since  $c_{j_1} x_{2i,{p_2}} = c_{j_1}$ in \textbf{x}, $(r_{2i-1}, r_{2i})$ is jointly assigned to a worse partner than $(h_{j_1}, h_{j_2})$. Further, since $\sum\limits_{q^{\prime }=1}^{q-1} \{ x_{i^{\prime },p^{\prime \prime}} \in X : ( r_{i^{\prime }, p^{\prime \prime}}) \in R(h_{j_1}, q^{\prime })\} < c_{j_1}$ in \textbf{x}, $h_{j_1}$ is either under-subscribed in $M$ or prefers $r_{2i-1}$ to some member of $M(h_{j_1})$. Thus, $(r_{2i-1}, r_{2i})$ blocks $M$ with $(h_{j_1}, h_{j_2})$, a contradiction. Therefore Constraint \ref{constraint:HRC2_6} holds in the assignment derived from $M$.

% HRC2_7 constraint - stability 2a(ii)

Assume \textbf{x} does not satisfy Constraint \ref{constraint:HRC2_7}. Let there be $x_{2i,{p_2}}$ such that $i ~ (1\leq i\leq c)$ and $p_1, p_2 (1\leq p_1 < p_2 \leq l(r_{2i-1}))$, where $h_{j} = pref(r_{2i-1, p_1}) = pref(r_{2i, p_1})$, $pref(r_{2i}, p_1) = pref(r_{2i}, p_2) = h_j$  and $rank(h_{j}, r_{2i-1})=q$. If $(c_{j} - 1) x_{2i,{p_2}} = 0$ then Constraint \ref{constraint:HRC2_7} is trivially satisfied as $\sum\limits_{q^{\prime }=1}^{q-1} \{ x_{i^{\prime },p^{\prime \prime}} \in X : ( r_{i^{\prime }, p^{\prime \prime}}) \in R(h_{j}, q^{\prime })\} \geq 0$. Hence, $(c_{j} - 1) x_{2i,{p_2}} = c_{j} - 1$. If  $\sum\limits_{q^{\prime }=1}^{q-1} \{ x_{i^{\prime },p^{\prime \prime}} \in X : ( r_{i^{\prime }, p^{\prime \prime}}) \in R(h_{j}, q^{\prime })\} \geq c_{j} - 1$ then Constraint \ref{constraint:HRC2_7} is satisfied. Hence, $\sum\limits_{q^{\prime }=1}^{q-1} \{ x_{i^{\prime },p^{\prime \prime}} \in X : ( r_{i^{\prime }, p^{\prime \prime}}) \in R(h_{j}, q^{\prime })\} < c_{j} - 1$. 

However, since  $(c_{j}-1) x_{2i,{p_2}} = c_{j} - 1$ in \textbf{x}, $(r_{2i-1}, r_{2i})$ is jointly assigned to a worse partner than $(h_{j}, h_{j})$. Further, since $\sum\limits_{q^{\prime }=1}^{q-1} \{ x_{i^{\prime },p^{\prime \prime}} \in X : ( r_{i^{\prime }, p^{\prime \prime}}) \in R(h_{j}, q^{\prime })\} < (c_j - 1)$ in \textbf{x}, $h_{j}$ is either under-subscribed in $M$ or prefers $r_{2i-1}$ to some assignee in $M(h_{j})$ other than $r_{2i}$. Thus, $(r_{2i-1}, r_{2i})$ blocks $M$ with $(h_{j}, h_{j})$, a contradiction. Therefore Constraint \ref{constraint:HRC2_7} holds in the assignment derived from $M$.

% HRC2_8 and HRC2_9 constraint - stability 2a(i) and 2a(ii)

A similar argument for the odd members of the couples in $M$ ensures that Constraints \ref{constraint:HRC2_8} and \ref{constraint:HRC2_9} are also satisfied in \textbf{x}.

% HRC2_10 - stability 3a

Assume \textbf{x} does not satisfy Constraint \ref{constraint:HRC2_10}. There exists $ i ~ (1\leq i\leq c)$ and $p ~ (1\leq p\leq r_{2i-1})$, where $h_{j_1} = pref(r_{2i-1}, p)$, $h_{j_2} = pref(r_{2i}, p)$, $h_{j_1} \neq h_{j_2}$, $rank(h_{j_1}, r_{2i-1}) =q_1$ and $rank(h_{j_2}, r_{2i}) =q_2$, if $\sum\limits_{p^{\prime }=p+1}^{l(r_{2i-1})+1} x_{2i-1, p^{\prime }} =0$ then Constraint \ref{constraint:HRC2_10} must be satisfied in \textbf{x}. Hence, $\sum\limits_{p^{\prime }=p+1}^{l(r_{2i-1})+1} x_{2i-1, p^{\prime }} =1$. If $\alpha _{j_1, q_1} = 0$ (similarly $\alpha _{j_2, q_2} = 0$)  then Constraint \ref{constraint:HRC2_10} must be satisfied in \textbf{x}. Hence $\sum\limits_{p^{\prime }=p+1}^{l(r_{2i-1})+1} x_{2i-1, p^{\prime }} =1$,  $\alpha _{j_1, q_1} =1$ and $\alpha _{j_2, q_2} =1$.

Since $\alpha _{j_1, q_1} =1$, $ { \sum\limits_{q^{\prime }=1}^{q_1 -1} \{ x_{i^{\prime },p^{\prime \prime}} \in X : ( r_{i^{\prime }, p^{\prime \prime}}) \in R(h_{j_1}, q^{\prime })  \}  } < c_j$. Thus $h_{j_1}$ is under-subscribed in $M$ or prefers $r_{2i-1}$ to some assignee in $M(h_{j_1})$. Similarly, if $\alpha _{j_2, q_2} =1$, then $h_{j_2}$ is under-subscribed in $M$ or prefers $r_{2i}$ to some assignee in $M(h_{j_2})$. Also in $M$, $(r_{2i-1}, r_{2i})$ is unassigned or is assigned to a worse partner than $(h_{j_1}, h_{j_2})$. Thus, $(r_{2i-1}, r_{2i})$ blocks $M$ with $(h_{j_1}, h_{j_2})$, a contradiction. Therefore Constraint \ref{constraint:HRC2_10} is satisfied in \textbf{x, $\alpha $, $\beta $}.

%constraint HRC2_11

Let $$\gamma = c_j \sum\limits_{p^{\prime }=p+1}^{l(r_{2i-1})+1} x_{2i-1,p^{\prime }} - \dfrac {\sum\limits_{q^{\prime }=1}^{q-1} \{ x_{i^{\prime },p^{\prime \prime}} \in X : ( r_{i^{\prime }, p^{\prime \prime}}) \in R(h_{j}, q^{\prime })\}}{ (c_j -1 )}$$ Further, let $$\delta = \sum\limits_{q^{\prime }=1}^{l(h_j)} \{ x_{i^{\prime },p^{\prime \prime}} \in X: (r_{i^{\prime }}, p^{\prime \prime}) \in R(h_j, q^{\prime }) \}$$ and let $$\varepsilon = {\sum\limits_{q^{\prime }=1}^{q-1} \{ x_{i^{\prime },p^{\prime \prime}} \in X : ( r_{i^{\prime }, p^{\prime \prime}}) \in R(h_{j}, q^{\prime })\}}$$

Assume \textbf{x} does not satisfy Constraint \ref{constraint:HRC2_11}. Hence, $\gamma > \delta $. If $c_j \sum\limits_{p^{\prime }=p+1}^{l(r_{2i-1})+1} x_{2i-1,p^{\prime }} = 0$ then $\gamma \leq 0$. However, $\delta \geq 0$, a contradiction. Hence $c_j \sum\limits_{p^{\prime }=p+1}^{l(r_{2i-1})+1} x_{2i-1,p^{\prime }} = c_j$.

Clearly, $0 \leq \varepsilon \leq c_j$. Assume $\varepsilon = c_j$. Hence, $\gamma = c_j - ( c_j / c_j - 1) = c_j(c_j -2)/ (c_j -1)$. A simple argument shows that $c_j - 2 < \gamma < c_j - 1$. Thus $\delta \leq c_j - 2$. However, this implies that $h_j$ has two free posts in $M$ and $(r_{2i-1}, r_{2i})$ is unassigned or is assigned to a worse partner than $(h_{j}, h_{j})$. Thus, $(r_{2i-1}, r_{2i})$ blocks $M$ with $(h_j, h_j)$ , a contradiction. Assume $\varepsilon = c_j -1$. Hence $\gamma = c_j -1$. Thus, $\delta \leq c_j - 2$. However, this again implies that $h_j$ has two vacant posts in $M$ and $(r_{2i-1}, r_{2i})$ is unassigned or is assigned to a worse partner than $(h_{j}, h_{j})$. Thus, $(r_{2i-1}, r_{2i})$ blocks $M$ with $(h_j, h_j)$, a contradiction. 

Hence, $\varepsilon < c_j -1$ and thus $c_j -1 < \gamma \leq c_j$. Therefore, $\delta \leq c_j -1$. This implies that $h_j$ has a vacant post in $M$, moreover, $h_j$ prefers $r_{2i-1}$ or $r_{2i}$  to at least one of its assignees and $(r_{2i-1}, r_{2i})$ is unassigned or is assigned to a worse partner than $(h_{j}, h_{j})$. Hence, $(r_{2i-1}, r_{2i})$ blocks $M$ with $(h_j, h_j)$, a contradiction. Therefore Constraint \ref{constraint:HRC2_11} is satisfied in \textbf{x}.

%constraint HRC2_12

Assume \textbf{x, $\alpha $, $\beta $} does not satisfy Constraint \ref{constraint:HRC2_12}. For some $i ~ (1\leq i\leq c)$ and $p ~ (1\leq p\leq l(r_{2i-1}))$ where $pref(r_{2i-1}, p) = pref(r_{2i},p) = h_j$, let $r_{min}$ be the better of $r_{2i}$ and $r_{2i-1}$ according to hospital $h_j$ with $rank(h_j, r_{min}) = q_{min}$. 

%If $\sum\limits_{p^{\prime }=p+1}^{l(r_{2i-1})+1} x_{2i-1, p^{\prime }} =0$ then Constraint \ref{constraint:HRC2_12} must be satisfied in \textbf{x, $\alpha $, $\beta $}. Hence, $\sum\limits_{p^{\prime }=p+1}^{l(r_{2i-1})+1} x_{2i-1, p^{\prime }} =1$. If $\alpha _{j, q_{max}} =0$ then Constraint \ref{constraint:HRC2_12} must be satisfied in \textbf{x, $\alpha $, $\beta $}. Similarly if $\beta _{j, q_{min}} =0$ then Constraint \ref{constraint:HRC2_12} must be satisfied in \textbf{x, $\alpha $, $\beta $}.  

Hence $\sum\limits_{p^{\prime }=p+1}^{l(r_{2i-1})+1} x_{2i-1, p^{\prime }} =1$,  $\alpha _{j, q_{max}} =1$ and $\beta _{j, q_{min}} =1$. Since $\alpha _{j, q_{max}} =1$, $$ { \sum\limits_{q^{\prime }=1}^{q_{max} -1} \{ x_{i^{\prime },p^{\prime \prime}} \in X : ( r_{i^{\prime }, p^{\prime \prime}}) \in R(h_{j}, q^{\prime })  \}  } < c_j$$ Thus $h_{j}$ is under-subscribed in $M$ or prefers $r_{2i-1}$ to some assignee, $r_x$, in $M(h_{j})$. Similarly, if $\beta _{j, q_{min}} =1$ then $${ \sum\limits_{q^{\prime }=1}^{q -1} \{ x_{i^{\prime },p^{\prime \prime}} \in X: ( r_{i^{\prime }, p^{\prime \prime}}) \in R(h_{j}, q^{\prime })  \}  } < c_j -1  $$ Thus $h_{j}$ is under-subscribed in $M$ or prefers $r_{2i}$ to some assignee, $r_y$, in $M(h_{j})$. 

This implies that in $M$, $(r_{2i-1}, r_{2i})$ is assigned to a worse partner than $(h_j, h_j)$, $h_j$ prefers $r_{2i-1}$ to some $r_{x} \in M(h_j)$ and also prefers $r_{2i}$ to some $r_{y} \in M(h_j)$. Moreover, in $M$, $(r_{2i-1}, r_{2i})$ is unassigned or is assigned to a worse partner than $(h_{j}, h_{j})$. Thus, $(r_{2i-1}, r_{2i})$ blocks $M$ with $(h_j, h_j)$, a contradiction. Therefore Constraint \ref{constraint:HRC2_12} is satisfied in \textbf{x, $\alpha $, $\beta $}. As all of the constraints in $J$ hold for an assignment derived from a stable matching $M$, a stable matching $M$ in $I$ represents a feasible solution to $J$.

Conversely, consider a feasible  solution, $\langle \textbf{x} \rangle$,  to $J$. From such a solution we form a set of pairs, $M$, as follows.

Initially let $M = \emptyset$. For each $i$   $(1\leq i\leq n_1)$ and $p$  $(1\leq p\leq l(r_i))$ if $x_{i,p}=1$ then add $(r_i, h_j)$ to $M$ where $h_j= pref(r_i, p)$. As $\langle \textbf{x} \rangle$ satisfies Constraints \ref{constraint:HRC2_1}, \ref{constraint:HRC2_2} and \ref{constraint:HRC2_3}, each resident in $M$ must have exactly one partner or be unassigned (but not both) and each hospital in $M$ must have at most $c_j$ assignees.

As $\langle \textbf{x} \rangle$ satisfies Constraint \ref{constraint:HRC2_4} each resident couple $(r_{2i-1}, r_{2i})$ must be either jointly assigned to a hospital pair $(h_{j_1}, h_{j_2})$ where $pref((r_{2i-1}, r_{2i}, p ) =(h_{j_1}, h_{j_2}))$ for some $p ~ (1\leq p\leq l(r_{2i-1}))$, meaning that $(r_{2i-1}, h_j) \in M$ and $(r_{2i}, h_{j_2}) \in M$, or jointly unassigned meaning that both $r_{2i-1}$ and $r_{2i}$ are unassigned in $M$.

Therefore the set of pairs $M$ created from the solution $\langle$ \textbf{x} \textbf{$\alpha $} \textbf{$\beta $} $\rangle$ to $J$, is a matching in $I$. We now show that $M$ is stable. 

Type 1 Blocking Pair - Assume $(r_i, h_j)$ blocks $M$ (as a type 1 blocking pair), where $r_i$ is a single resident. Let $pref(r_i,p) = h_j$ and $rank(h_j, r_i) =q$. Therefore, $r_i$ is unassigned or, has a worse partner than $h_j$ and $h_j$ is under-subscribed or prefers $r_i$ to some member of $M(h_j)$. However, this implies that in $J$,  $c_j \sum\limits_{p^{\prime}=p+1}^{l(r_i)+1} x_{i,p^{\prime}} = c_j$ and $ \sum\limits_{q^{\prime}=1}^{q-1} \{x_{i^{\prime },p^{\prime \prime}} \in X :   rank(h_j,r_{i^{\prime }}) = q^{\prime } \wedge pref(r_{i^{\prime }}, p^{\prime \prime}) = h_j) \} < c_j$ and therefore Constraint \ref{constraint:HRC2_5} is not satisfied in $J$, a contradiction. Therefore no such $(r_i, h_j)$ can block $M$.

Type 2 Blocking Pair - Assume $(r_{2i-1}, r_{2i})$ blocks $M$ (as a type 2 blocking pair) with $(h_{j_1}, h_{j_2})$ where $pref( (r_{2i-1}, r_{2i}), p_1 ) = (h_{j_1}, h_{j_2})$, $$pref( (r_{2i-1}, r_{2i}), p_2) =  ( M(r_{2i-1}), M(r_{2i}) )$$ for $(1\leq p_1 < p_2 \leq l(r_{2i-1}))$, $pref(r_{2i}, p_1) = pref(r_{2i}, p_2) = h_{j_2}$  and $rank(h_{j_1}, r_{2i-1})=q$. Hence, $r_{2i}$ has the same hospital in positions $p_1$ and $p_2$, and $h_{j_1}$ is under-subscribed or prefers $r_{2i-1}$ to some member of $M(h_{j_1})$.

Further assume $pref(r_{2i-1}, p_1)\neq pref(r_{2i}, p_1)$. Hence $c_{j_1} x_{2i, p_2} = c_{j_1}$ and $$\sum\limits_{q^{\prime }=1}^{q-1} \{ x_{i^{\prime },p^{\prime \prime}} \in X : ( r_{i^{\prime }, p^{\prime \prime}}) \in R(h_{j_1}, q^{\prime })\} < c_{j_1}$$ as $h_{j_1}$ is under-subscribed or prefers $r_{2i-1}$ to some member of $M(h_{j_1})$. Hence in $J$, the RHS of Constraint \ref{constraint:HRC2_6} is at most $(c_{j_1} - 1)$ and the LHS is equal to $c_{j_1}$ and therefore Constraint \ref{constraint:HRC2_6} is not satisfied in $J$, a contradiction. Therefore no such $( (r_{2i-1}, r_{2i}), (h_{j_1}, h_{j_2}) )$ can block $M$.

Thus, $pref(r_{2i-1}, p_1)= pref(r_{2i}, p_1)$. Hence $( c_{j_1} - 1) x_{2i, p_2} = ( c_{j_1} - 1)$ and $\sum\limits_{q^{\prime}=1}^{q-1}  \{ x_{i^{\prime },p^{\prime \prime}} \in X : q^{\prime } \neq rank(h_{j_1}, r_{2i}) \wedge ( r_{i^{\prime }, p^{\prime \prime}}) \in R(h_{j_1}, q^{\prime }) < c_{j_1} - 1)$ as $h_{j_1}$ is under-subscribed or prefers $r_{2i-1}$ to some member of $M(h_{j_1})$ other than $r_{2i}$. Hence in $J$, the RHS of Constraint \ref{constraint:HRC2_7} is at most $c_{j_1} - 2$ and the LHS is equal to $c_{j_1} - 1$ and therefore Constraint \ref{constraint:HRC2_7} is not satisfied in $J$, a contradiction. Therefore no such $( (r_{2i-1}, r_{2i}), (h_{j_1}, h_{j_2}) )$ can block $M$.

A similar argument can be used to show that the odd member of each couple cannot improve in such a blocking pair in $M$ and therefore Constraint \ref{constraint:HRC2_8} and \ref{constraint:HRC2_9} are both satisfied in the assignment derived from $M$.

\noindent Type 3 Blocking Pairs - Suppose that $(r_{2i-1}, r_{2i})$ blocks $M$ (as a type 3 blocking pair) with $(h_{j_1}, h_{j_2})$ where $rank( (r_{2i-1}, r_{2i}), (h_{j_1}, h_{j_2}) )= p$, $rank(h_{j_1}, r_{2i-1}) = q_1$ and $rank(h_{j_2}, r_{2i}) = q_2$. Hence, $(r_{2i-1}, r_{2i})$ is unassigned or prefers $(h_{j_1}, h_{j_2})$ to $( M(r_{2i-1}), M(r_{2i}) )$ where $h_{j_1} \neq M(r_{2i-1})$ and $h_{j_2} \neq M(r_{2i})$.

Type 3a Blocking Pair -  $h_{j_1} \neq h_{j_2}$. Therefore, $h_{j_1}$ is under-subscribed or prefers $r_{2i-1}$ to some member of $M(h_{j_1})$ and $h_{j_2}$ is under-subscribed prefers $r_{2i}$ to some member of $M(h_{j_2})$. However, this implies that both, $$ { \sum\limits_{q^{\prime }=1}^{q_1 -1} \{ x_{i^{\prime },p^{\prime \prime}} \in X : ( r_{i^{\prime }, p^{\prime \prime}}) \in R(h_{j_1}, q^{\prime })  \}  }  <  c_{j_1} $$ and $$  { \sum\limits_{q^{\prime }=1}^{q_2 -1} \{ x_{i^{\prime },p^{\prime \prime}} \in X : ( r_{i^{\prime }, p^{\prime \prime}}) \in R(h_{j_2}, q^{\prime })  \} }  < c_{j_2} $$ in $J$. Hence $\alpha _{j_1, q_1} = 1$, $\alpha _{j_2, q_2} =1$ and $$\sum\limits_{p^{\prime }=p+1}^{l(r_{2i-1})+1} x_{2i-1, p^{\prime }} = 1$$ and thus Constraint \ref{constraint:HRC2_10} is not satisfied in $J$, a contradiction. Therefore no such $( (r_{2i-1}, r_{2i}), $ $ (h_{j_1}, h_{j_2}))$ can block $M$.

Type 3b Blocking Pair - $h_{j_1} = h_{j_2} = h_j $ and $h_j$ has two unassigned posts in $M$. 

However, this implies that in $J$, $$\sum\limits_{q^{\prime }=1}^{l(h_j)} \{ x_{i^{\prime },p^{\prime \prime}} \in X : (r_{i^{\prime }}, p^{\prime \prime}) \in R(h_j, q^{\prime }) \leq c_j -2$$ and also $${\sum\limits_{q^{\prime }=1}^{q-1} \{ x_{i^{\prime },p^{\prime \prime}} : ( r_{i^{\prime }, p^{\prime \prime}}) \in R(h_{j}, q^{\prime })\}} \leq c_j - 2 $$ Further, $$ c_j \sum\limits_{p^{\prime }=p+1}^{l(r_{2i-1})+1} x_{2i-1,p^{\prime }} =c_j$$ since $r_{2i-1}$ and $r_{2i}$ prefer $h_j$ to $M(r_{2i-1})$ and $M(r_{2i})$ respectively. Hence in $J$, the RHS of Constraint \ref{constraint:HRC2_11} is at most $c_j - 2$ and the LHS is greater than $c_j - 1$ and therefore Constraint \ref{constraint:HRC2_11} is not satisfied in $J$, a contradiction. Therefore no such $( (r_{2i-1}, r_{2i}), (h_j, h_j) )$ can block $M$.

Type 3c Blocking Pair - $h_{j_1} = h_{j_2} = h_j $ and $h_j$ has a vacant post in $M$ and $h_j$ also prefers $r_{2i-1}$ or $r_{2i}$ to some other member of $M(h_j)$. Let $q = \min\{rank(h_{j},r_{2i-1}), rank(h_{j},r_{2i})\}$.

However, this implies that in $J$,$\sum\limits_{q^{\prime }=1}^{l(h_j)} \{ x_{i^{\prime },p^{\prime \prime}} \in X : (r_{i^{\prime }}, p^{\prime \prime}) \in R(h_j, q^{\prime }) \leq c_j -1$.

Since $h_j$ prefers $r_{2i-1}$ or $r_{2i}$ to some other member of $M(h_j)$ and $h_j$ also has a free post $$\sum\limits_{q^{\prime }=1}^{q-1} \{ x_{i^{\prime },p^{\prime \prime}} : ( r_{i^{\prime }, p^{\prime \prime}}) \in R(h_{j}, q^{\prime }) \leq (c_j -2)$$ Since $(r_{2i-1}, r_{2i})$ is unassigned or prefers $(h_{j_1}, h_{j_2})$ to $( M(r_{2i-1}), M(r_{2i}) )$, $$ c_j \sum\limits_{p^{\prime }=p+1}^{l(r_{2i-1})+1} x_{2i-1,p^{\prime }} =c_j$$ Hence, $$ {\sum\limits_{q^{\prime }=1}^{q-1} \{ x_{i^{\prime },p^{\prime \prime}} \in X : ( r_{i^{\prime }, p^{\prime \prime}}) \in R(h_{j}, q^{\prime })\}} < c_j -1$$ Thus in $J$, the RHS of Constraint \ref{constraint:HRC2_11} is at most $c_j - 1$ and the LHS is greater than $c_j - 1$ and therefore Constraint \ref{constraint:HRC2_11} is not satisfied in $J$, a contradiction. Therefore no such $( (r_{2i-1}, r_{2i}), (h_j, h_j) )$ can block $M$.

Type 3d Blocking Pair - $h_{j_1} = h_{j_2} = h_j $, $h_{j}$ is fully subscribed and also has two assignees $r_x$ and $r_y$ (where $x\neq y$ and neither $x$ nor $y$ is equal to $r_{2i-1}$ or $r_{2i}$) such that $h_{j}$ prefers $r_{2i-1}$ to $r_x$ and also prefers $r_{2i}$ to $r_{y}$. Let $r_{min}$ be the better of $r_{2i}$ and $r_{2i-1}$ according to hospital $h_j$ with $rank(h_j, r_{min}) = q_{min}$. Analogously, let $r_{max}$ be the worse of $r_{2i}$ and $r_{2i-1}$ according to hospital $h_j$ with $rank(h_j, r_{max}) = q_{max}$.

However, this implies that both $${ \sum\limits_{q^{\prime }=1}^{q_{min} -1} \{ x_{i^{\prime },p^{\prime \prime}} \in X : ( r_{i^{\prime }, p^{\prime \prime}}) \in R(h_{j}, q^{\prime })  \}  }  <  c_{j} - 1 $$ and $${ \sum\limits_{q^{\prime }=1}^{q_{max} -1} \{ x_{i^{\prime },p^{\prime \prime}} \in X : ( r_{i^{\prime }, p^{\prime \prime}}) \in R(h_{j}, q^{\prime })  \} }  <  c_{j} $$ in $J$.  Hence $\beta _{j, q_{min}} =1$ and $\alpha _{j, q_{max}} = 1$. Also $$\sum\limits_{p^{\prime }=p+1}^{l(r_{2i-1})+1} x_{2i-1, p^{\prime }} = 1$$ and thus Constraint \ref{constraint:HRC2_12} is not satisfied in $J$, a contradiction. Therefore no such $( (r_{2i-1}, r_{2i}),$ $ (h_{j_1}, h_{j_2}))$ can block $M$.

\end{proof}

\subsection{Creating the IP model for an example {\sc hrc} instance}
\label{section:IPModelsExample}

\begin{figure}[h]
\[
\begin{array}{rll}

\multicolumn{3}{c}{Residents} \\
\hline
\\
(r_{1}, r_{2}) : & (h_{1}, h_{2}) ~~ (h_{2}, h_{1}) ~~ (h_{2}, h_{3}) & \\ \\
 
r_3 : & h_1 ~~ h_3 ~~  & \\ 
r_4 : & h_2 ~~ h_3 ~~  & \\ 
r_5 : & h_2 ~~ h_1 ~~  & \\ 
r_6 : & h_1 ~~ h_2 ~~  & \\  \\

\multicolumn{3}{c}{Hospitals} \\
\hline
\\

h_1 : ~ 2 ~ : & r_1 ~~ r_3 ~~ r_2 ~~ r_6 ~~ r_5 \\
h_2 : ~ 2 ~ : & r_2 ~~ r_6 ~~ r_1 ~~ r_4 ~~ r_5 \\
h_3 : ~ 2 ~ : & r_4 ~~ r_3 ~~ r_2 \\
\end{array}
\]
\caption{Example instance of HRC.}
\label{instance:example}
\end{figure}

Let $I$ be the example instance of {\sc hrc} shown in Figure \ref{instance:example} where the capacity of each hospital in $I$ is shown after the first colon, followed by the preference list. We shall consider the creation of the corresponding IP model $J$ for the example instance $I$. For each resident $r_i \in I ~ (1\leq i\leq 6)$ construct a vector $x_i$ consisting of $l(r_i)+1$ binary variables, $x_{i,p} ~ (1\leq p\leq l(r_i)+1)$, as shown in Figure \ref{model:example_variables}, and apply the constraints as described in Section \ref{section:IPModelsHRC}. Thus, we form an IP model $J$ derived from $I$.

\begin{figure}[h]
\[
\begin{array}{rll}
x_1 : & \langle ~ x_{1,1} ~~ x_{1,2} ~~ x_{1,3} ~~ x_{1,4} ~  \rangle & \\ 
x_2 : & \langle ~ x_{2,1} ~~ x_{2,2} ~~ x_{2,3} ~~ x_{2,4} ~  \rangle & \\ 
x_3 : & \langle ~ x_{3,1} ~~ x_{3,2} ~~ x_{2,3} ~ \rangle & \\ 
x_4 : & \langle ~ x_{4,1} ~~ x_{4,2} ~~ x_{4,3} ~ \rangle & \\ 
x_5 : & \langle ~ x_{5,1} ~~ x_{5,2} ~~ x_{5,3} ~ \rangle & \\ 
x_6 : & \langle ~ x_{6,1} ~~ x_{6,2} ~~ x_{6,3} ~ \rangle & \\  
\end{array}
\]
\caption{Variables created in $J$ from the example instance of {\sc hrc} shown in Figure \ref{instance:example}.}
\label{model:example_variables}
\end{figure}

Let $\mathbf{x} ^u$ denote the assignment of values to the variables in the IP model $J$ shown in Figure \ref{model:example_unstable}. We will show that $\mathbf{x} ^u$ is not a feasible solution to the IP model $J$ and thus, by Theorem \ref{IPHRC2}, does not correspond to a stable matching in $I$. However, as all instantiations of Constraints \ref{constraint:HRC2_1} - \ref{constraint:HRC2_4} hold for $\mathbf{x} ^u$, $\mathbf{x} ^u$ does correspond to a matching in $I$, namely $M_u = \{ (r_1, h_2 ) , (r_2, h_3 ), (r_3 , h_1 ),$ $ (r_4 , h_3 ), (r_5 , h_1 ), (r_6 , h_2 ) \}$. We shall demonstrate that several constraints in $J$ are violated by $\mathbf{x} ^u$ and that these constraints correspond to blockings pairs of $M_u$ in $I$.

\begin{figure}[h]
\[
\begin{array}{rll}
x_1 : & \langle ~ 0 ~~ 0 ~~ 1 ~~ 0 ~  \rangle & \\ 
x_2 : & \langle ~ 0 ~~ 0 ~~ 1 ~~ 0 ~  \rangle & \\ 
x_3 : & \langle ~ 1 ~~ 0 ~~ 0 ~  \rangle & \\ 
x_4 : & \langle ~ 0 ~~ 1 ~~ 0 ~  \rangle & \\ 
x_5 : & \langle ~ 0 ~~ 1 ~~ 0 ~  \rangle & \\ 
x_6 : & \langle ~ 0 ~~ 1 ~~ 0 ~  \rangle & \\  
\end{array}
\]
\caption{The assignment of values, $\mathbf{x} ^u$, to the variables in the IP Model $J$ corresponding to the unstable matching $M_u$ in $I$, the example instance of {\sc hrc}.}
\label{model:example_unstable}
\end{figure}

Inequality \ref{constraint:type1blockingpairexample} represents the instantiation of Constraint \ref{constraint:HRC2_5} in the case that $i=6$ and $p=1$. The LHS of Inequality \ref{constraint:type1blockingpairexample} is the product of the capacity of $h_1$ and the values of the variables that represent $r_6$ being matched to a worse partner than $h_1$ or being unmatched. The RHS of Inequality \ref{constraint:type1blockingpairexample} is the summation of the values of the variables that indicate whether $h_1$ is matched to partners it prefers to $r_6$.

\begin{equation} \label{constraint:type1blockingpairexample} \displaystyle c_1 ( x_{6,2} +x_{6,3} ) \leq x_{1,1} + x_{3,1} + x_{2,2} \end{equation}

The acceptable pair $(r_6 , h_1 )$ is a Type 1 blocking pair of $M_u$ in $I$. %Correspondingly one of the instantiations of Constraint \ref{constraint:HRC2_5} is not satisfied by $\mathbf{x} ^u$. Inequality \ref{constraint:type1blockingpairexample} is such an instantiation in $J$ in the case that $i=6$ and $p=1$. 
In this case the LHS of Inequality \ref{constraint:type1blockingpairexample} equals 2 and the RHS of Inequality \ref{constraint:type1blockingpairexample} equals 1. Hence Inequality \ref{constraint:type1blockingpairexample} is not satisfied in $\mathbf{x} ^u$ and thus $\mathbf{x} ^u$ is not a feasible solution to $J$.

\medskip

Inequality \ref{constraint:type2blockingpairexample} represents the instantiation of Constraint \ref{constraint:HRC2_8} in the case that $p1=2$, $p2=3$ and $r=3$. In this case the LHS of Inequality \ref{constraint:type2blockingpairexample} is the product of the capacity of $h_1$ and the value of the variable that represents $r_1$ being matched at position 3 on its projected preference list (and thus, since no instance of Constraint \ref{constraint:HRC2_4} is violated, $(r_1, r_2 )$ being jointly matched to the pair in position 3 on its joint projected preference list). The RHS of Inequality \ref{constraint:type2blockingpairexample} is the summation of the values of variables which indicate whether $h_1$ is matched to partners it prefers to $r_2$.

\begin{equation} \label{constraint:type2blockingpairexample} \displaystyle c_1 ( x_{1,3} ) \leq x_{1,1} + x_{3,1} \end{equation}

The acceptable pair $( (r_1 , r_2 ), (h_2 , h_1 ))$ is a Type 2 blocking pair of $M_u$ in $I$. %Correspondingly one of the instantiations of Constraints \ref{constraint:HRC2_6} - \ref{constraint:HRC2_9}, in this case Constraint \ref{constraint:HRC2_8}, is not satisfied by $\mathbf{x} ^u$. Equation \ref{constraint:type2blockingpairexample} is such an instantiation in $J$ in the case that $p1=1$, $p2=3$ and $r=3$. 
In this case the LHS of Inequality \ref{constraint:type2blockingpairexample} equals 2 and the RHS of Inequality \ref{constraint:type2blockingpairexample} equals 1. Hence Inequality \ref{constraint:type2blockingpairexample} is not satisfied in $\mathbf{x} ^u$ and thus $\mathbf{x} ^u$ is not a feasible solution to $J$.

\medskip

Inequality \ref{constraint:type3ablockingpairexample} represents the instantiation of Constraint \ref{constraint:HRC2_10} in the case that $i=1$ and $p=1$. In this case the summation on the LHS of Inequality \ref{constraint:type3ablockingpairexample} is over the variables that represent $r_1$ being matched to a worse partner than $h_1$ in position 1 on its projected preference or unmatched. (Since no instance of Constraint \ref{constraint:HRC2_4} is violated in $J$ these variables equally represent $(r_1 , r_2 )$ being jointly matched to a worse joint partner than $(h_1 , h_2 )$ or being jointly unmatched). Also, $\alpha _{1,1}$ is a variable constrained to take a value of 1 in the case that $h_1$ prefers less than $c_1$ residents to $r_1$. Similarly, $\alpha _{2,1}$ is a variable constrained to take a value of 1 in the case that $h_2$ prefers less than $c_2$ residents to $r_1$.

\begin{equation} \label{constraint:type3ablockingpairexample} \displaystyle ( x_{1,2} +  x_{1,3} +  x_{1,4} ) + \alpha _{1,1} + \alpha _{2,1} \leq 2 \end{equation}

The acceptable pair $( (r_1 , r_2 ), (h_1 , h_2 ))$ is a Type 3a blocking pair of $M_u$ in $I$. %Correspondingly one of the instantiations of Constraint \ref{constraint:HRC2_10} is not satisfied by $\mathbf{x} ^u$. Equation \ref{constraint:type3ablockingpairexample} is such an instantiation in the case that $i=1$ and $p=1$. 
In this case the summation on the LHS of Inequality \ref{constraint:type3ablockingpairexample} equals 1. Also, since $r_1$ is in first position on $h_1$'s preference list and thus $h_1$ prefers no other assignees to $r_1$, $\alpha _{1,1} \geq ( 1 -  ( 0 / 2))$ and hence $\alpha _{1,1} = 1$. Similarly, $\alpha _{2,1} \geq ( 1 -  ( 0 / 2))$ since $r_2$ is in first position on $h_2$'s preference list and hence $\alpha _{2,1} = 1$. Thus Inequality \ref{constraint:type3ablockingpairexample} is not satisfied in $\mathbf{x} ^u$ and $\mathbf{x} ^u$ is not a feasible solution to $J$.

\begin{figure}[h]
\[
\begin{array}{rll}
x_1 : & \langle ~ 1 ~~ 0 ~~ 0 ~~ 0 ~  \rangle & \\ 
x_2 : & \langle ~ 1 ~~ 0 ~~ 0 ~~ 0 ~  \rangle & \\ 
x_3 : & \langle ~ 1 ~~ 0 ~~ 0 ~  \rangle & \\ 
x_4 : & \langle ~ 0 ~~ 1 ~~ 0 ~  \rangle & \\ 
x_5 : & \langle ~ 0 ~~ 0 ~~ 1 ~  \rangle & \\ 
x_6 : & \langle ~ 0 ~~ 1 ~~ 0 ~  \rangle & \\  
\end{array}
\]
\caption{The assignment of values, $\mathbf{x} ^s$, to the variables in the IP model $J$ corresponding to the stable matching $M_s$ in $I$, the example instance of {\sc hrc}.}
\label{model:example_stable}
\end{figure}

Let $\mathbf{x} ^s$ denote the assignment of values to the variables in the IP model $J$ shown in Figure \ref{model:example_stable}. $\mathbf{x} ^s$ is a feasible solution to the IP model $J$ and as such does correspond with a stable matching in $I$, namely $M_s = \{ (r_1, h_1 ) , (r_2, h_2 ), (r_3 , h_1 ), (r_4 , h_3 ), (r_6 , h_2 ) \}$.

Consider a potential blocking pair of $M_s$. $(r_5, h_2)$ is an acceptable pair in $I$ and $r_5$ is unmatched in $M_s$. Inequality \ref{constraint:notblockingpairexample1} represents the instantiation of Constraint \ref{constraint:HRC2_5} in the case that $i=5$ and $p=1$. The LHS of Inequality \ref{constraint:notblockingpairexample1} is the product of the capacity of $h_2$ and the values of the variables that represent $r_5$ being matched to a worse partner than $h_2$. The RHS of Inequality \ref{constraint:notblockingpairexample1} is the summation of the values of the variables that indicate whether $h_2$ is matched to the partners it prefers to $r_5$.

\begin{equation} \label{constraint:notblockingpairexample1} \displaystyle c_2 ( x_{5,2} +x_{5,3} ) \leq x_{2,1} + x_{6,2}+ x_{1,2} + x_{1,3} + x_{4,1} \end{equation}

In this case the LHS of Inequality \ref{constraint:notblockingpairexample1} equals 2 since $r_5$ is unmatched and the RHS of Inequality \ref{constraint:notblockingpairexample1} also equals 2 since $h_2$ has two assignees that it prefers to $r_5$. Hence Inequality \ref{constraint:notblockingpairexample1} is satisfied in $\mathbf{x} ^s$. A similar consideration of other possible blocking pairs of $M_s$ in $I$ shows that no constraint is violated by $\mathbf{x} ^s$ and thus $\mathbf{x} ^s$ is a feasible solution of $J$.

\subsection{An integer programming formulation for {\sc hrct}}

\label{section:IPModelsHRCT}

The \emph{Hospitals / Residents Problem with Couples and Ties} ({\sc hrct}) is a generalisation of {\sc hrc} in which hospitals (respectively residents) may find some subsets of their acceptable residents (respectively hospitals) equally preferable. Residents (respectively hospitals) that are found equally preferable by a hospital (respectively resident) are \emph{tied} with each other in the preference list of that hospital (respectively resident). 

In this section we show how to extend the IP model for {\sc hrc} as represented in Section \ref{section:IPModelsHRC} to the {\sc hrct} case. In order to do so we require some additional notation.

For an acceptable resident-hospital pair $(r_i, h_j)$, where $r_i$ is a single resident let $rank(r_i, $ $ h_j) = q$ denote the rank that resident $r_i$ assigns hospital $h_j$ where $1\leq q \leq l(r_i)$. Thus, $rank(r_i , h_j )$ is equal to the number of hospitals that $r_i$ prefers to $h_j$ plus one.

For an acceptable pair $( (r_{s}, r_{t}), (h_j , h_k))$ where $c=(r_s, r_t)$ is a couple, let $rank(c,$ $(h_j, h_k)) =q$ denote the rank that the couple $c=(r_{s}, r_{t})$ assigns the hospital pair $(h_j, h_k)$ where $1\leq q \leq l(c)$. Thus, $rank(c, (h_j, h_k))$ is equal to the number of hospital pairs that $( r_{s}, r_{t})$ jointly prefers to $(h_j , h_k)$ plus one.

For each single resident $r_i \in R$ and integer $p ~ (1\leq p\leq l(r_i))$ let
$$p^{+ } = \max \{ p^{\prime } : 1\leq p^{\prime } \leq l(r_i) \wedge rank (r_i, pref(r_i, p)) = rank(r_i , pref(r_i , p^{\prime }))  \}  $$

Similarly, in the case of a couple $c$ and integer $p ~ (1\leq p\leq l(c))$ let
$$p^{+ } = \max \{ p^{\prime } : 1\leq p^{\prime } \leq l(c) \wedge rank (c, pref(c, p)) = rank(c , pref(c , p^{\prime }))  \}  $$

Intuitively, for a single resident $r_i$, $p^{+ }$ is the largest position on $r_i$'s preference list of a hospital appearing in the same tie on $r_i$'s list as the hospital in position $p$ on $r_i$'s preference list. Also, for a couple $(r_i , r_j)$, $p^{+ }$ is the largest position on $(r_i , r_j)$'s joint preference list of a hospital pair appearing in the same tie on $(r_i , r_j)$'s preference list the as hospital pair in position $p$ on $(r_i ,r_j)$'s joint preference list.

To correctly construct an IP of model HRCT we must make the following alterations to the mechanism described in Section \ref{section:IPModelsHRC} for obtaining an IP model from an HRC instance. All constraints are as before unless otherwise noted. Since, a hospital $h_j$ may rank some members of $M(h_j)$ equally with $r_i$ in HRCT, the summations involving $q$ in Constraints \ref{constraint:HRC2_5} - \ref{constraint:HRC2_9} and \ref{constraint:HRC2_11} and the Inequalities \ref{definition:alpha} and \ref{definition:beta} must now range from $1$ to $q$.

Also, since a resident $r_i$ may rank $M(r_i)$ equally with $h_j$, the summations involving $p$ in Constraints \ref{constraint:HRC2_5}, \ref{constraint:HRC2_10},  \ref{constraint:HRC2_11} and \ref{constraint:HRC2_12} must now range from $p^{+}+1$ to $l(r_i)+1$. Further, we must extend the definition of $p_1$ and $p_2$ in Constraints \ref{constraint:HRC2_6} - \ref{constraint:HRC2_9} such that $1\leq p_1 \leq p^+_1 < p_2 \leq l(r_{s})$ where $r_s$ is the resident involved in each case.

To give an example of a modified constraint, a full description of Stability 1 within the {\sc hrct} context is:

\textbf{Stability 1} - In a stable matching $M$ in $I$, if a single resident $r_i\in R$ has a partner worse than some hospital $h_j\in H$ where $pref(r_i, p)=h_j$ and $rank(h_j, r_i) =q$ then $h_j$ must be fully subscribed with partners at least as good as $r_i$. Therefore, either $$\sum\limits_{p^{\prime}=p^{+ }+1}^{l(r_i)+1} x_{i,p^{\prime}}= 0$$ or $h_j$ is fully subscribed with partners at least as good as $r_i$, i.e. $$\sum\limits_{q^{\prime}=1}^{q} \{x_{i^{\prime },p^{\prime \prime}} \in X :   ( r_{i^{\prime }, p^{\prime \prime}}) \in R(h_{j}, q^{\prime }) \} = c_j$$

Thus, for each $i ~ (2c+1\leq i\leq n_1)$ and $p ~ (1\leq p\leq l(r_i))$ we obtain the following constraint where $pref(r_i, p) = h_j$ and $rank(h_j, r_i)=q$:

\begin{equation} \label{constraint:HRCT_1} \displaystyle c_j \sum\limits_{p^{\prime}=p^{+ }+1}^{l(r_i)+1} x_{i,p^{\prime}} \leq \sum\limits_{q^{\prime}=1}^{q} \{x_{i^{\prime },p^{\prime \prime}} \in X :   ( r_{i^{\prime }, p^{\prime \prime}}) \in R(h_{j}, q^{\prime }) \} \end{equation}

The other constraints in the {\sc hrct} model follow by adapting the remaining stability criteria in an analogous fashion. Using a proof analogous to that of Theorem \ref{IPHRC2}, the following result may be established.

\begin{theorem}\label{IPHRCT}
Given an instance $I$ of {\sc hr}, let $J$ be the corresponding IP model as defined in Section \ref{section:IPModelsHRCT}. A stable matching in $I$ is exactly equivalent to a feasible solution to $J$.
\end{theorem}

\section{Empirical Results}

\label{section:IPexperiments}

\subsection{Introduction}

We ran experiments on a Java implementation of the IP models as described in Section \ref{section:IPModels} applied to both randomly-generated and real data.  We present data showing (i) the average time taken to find a maximum cardinality stable matching or report that no stable matching exists, and (ii) the average size of a maximum cardinality stable matching where a stable matching did exist. Instances were generated with a skewed preference list distribution on both sides, taking into account that in practice some residents and hospitals are more popular than others (on both sides, the most popular agent was approximately 3 times as popular as the least popular agent).

All experiments were carried out on a desktop PC with an Intel i5-2400 3.1Ghz processor, with 8Gb of memory running Windows 7. The IP solver used in all cases was CPLEX 12.4 and the model was implemented in Java using CPLEX Concert. 

To test our implementation for correctness we used a brute force algorithm which recursively generated all possible matchings admitted by an {\sc hrc} instance and selected a maximum cardinality stable matching from amongst those matchings or reported that none of the generated matchings was stable. Due to the inefficiency of this algorithm it may only be realistically applied to relatively small instances. When solving several thousand {\sc hrc} instances involving up to 15 residents our implementation agreed with the brute force algorithm when reporting whether the instance admitted a stable solution and further our implementation returned a stable matching of the same size as a maximum cardinality stable matching output by the brute force algorithm.

For all the instances solved in these experiments, the minimum length of preference list for an individual resident is 5 and the maximum length is 10. In the SFAS application, the joint preference list for a couple $(r_i , r_j)$ is derived from the preference lists of the individual residents $r_i$ and $r_j$. The joint preference lists of the couples in the SFAS application are constructed as follows. For a couple $(r_i , r_j)$, let $s$ (respectively $t$) be the length of the individual preference list of $r_i$ (respectively $r_j$). Now, let $a$ and $b ~ (1\leq a\leq s, 1\leq b\leq t)$. The rank pair $(a, b)$ represents the $a^{th}$ hospital on resident $r_i$'s individual preference list and the $b^{th}$ hospital on resident $r_j$'s preference list. Couple $(r_i ,r_j)$ finds acceptable all pairs $(h_p ,h_q)$ where $r_i$ finds $h_p$ acceptable and $r_j$ finds $h_q$ acceptable ($st$ pairs in total). These pairs are ordered as follows.  Let $L=\max \{ s,t \}$.  Corresponding to every such acceptable pair $(h_p ,h_q)$, create an $L$-tuple whose $i^{th}$ entry is the number of residents in the couple who obtain their $i^{th}$ choice (when considering their individual lists) in the pair $(h_p ,h_q)$.  The acceptable pairs on the couple's list are then ordered according to a lexicographically increasing order on the reverse of the corresponding $L$-tuples. The preference lists of the couples in the randomly generated instances in the experiments that follow are constructed in a similar fashion.

\subsection{Experiments with randomly generated instances}

In the experiments which follow we consider the question of how the time taken to find maximum cardinality stable matchings or report that no stable matching exists in an instance of {\sc hrc} alters as we vary the parameters of the instance. Further, we consider how the size of a maximum cardinality stable matching supported by an instance changes as we vary the parameters of the instance.

\subsubsection{Experiment 1}

In this first experiment, we report on data obtained as we increased the number of residents while maintaining a constant ratio of couples, hospitals and posts to residents. For various values of $x ~ (100 \le x \le 1000)$ in increments of $30$, $1000$ randomly generated instances were created containing $x$ residents, $0.1x$ couples and $0.1x$ hospitals with $x$ available posts which were unevenly distributed amongst the hospitals. The mean time taken to find a maximum cardinality stable matching or report that no stable matching existed in each instance is plotted in Figure \ref{image:varyingsizeplots} for all values of $x$. Figure \ref{image:varyingsizeplots} also shows charts displaying the percentage of instances encountered which admit a stable solution and the mean size of the maximum cardinality stable solution for all values of $x$.

The data in Figure {\ref{image:varyingsizeplots} shows that the mean time taken to find a maximum cardinality stable matching or report that no stable matching existed increased as we increased the number of residents in the instance. Figure {\ref{image:varyingsizeplots} also shows that the percentage of {\sc hrc} instances admitting a stable matching did not appear to be correlated with the number of residents involved in the instance. Figure {\ref{image:varyingsizeplots} also shows that as the number of residents in the instances increased the mean size of the maximum cardinality stable matching supported by the instances increased.

\begin{figure}
\[
\includegraphics[width=\textwidth, height=4.55cm]{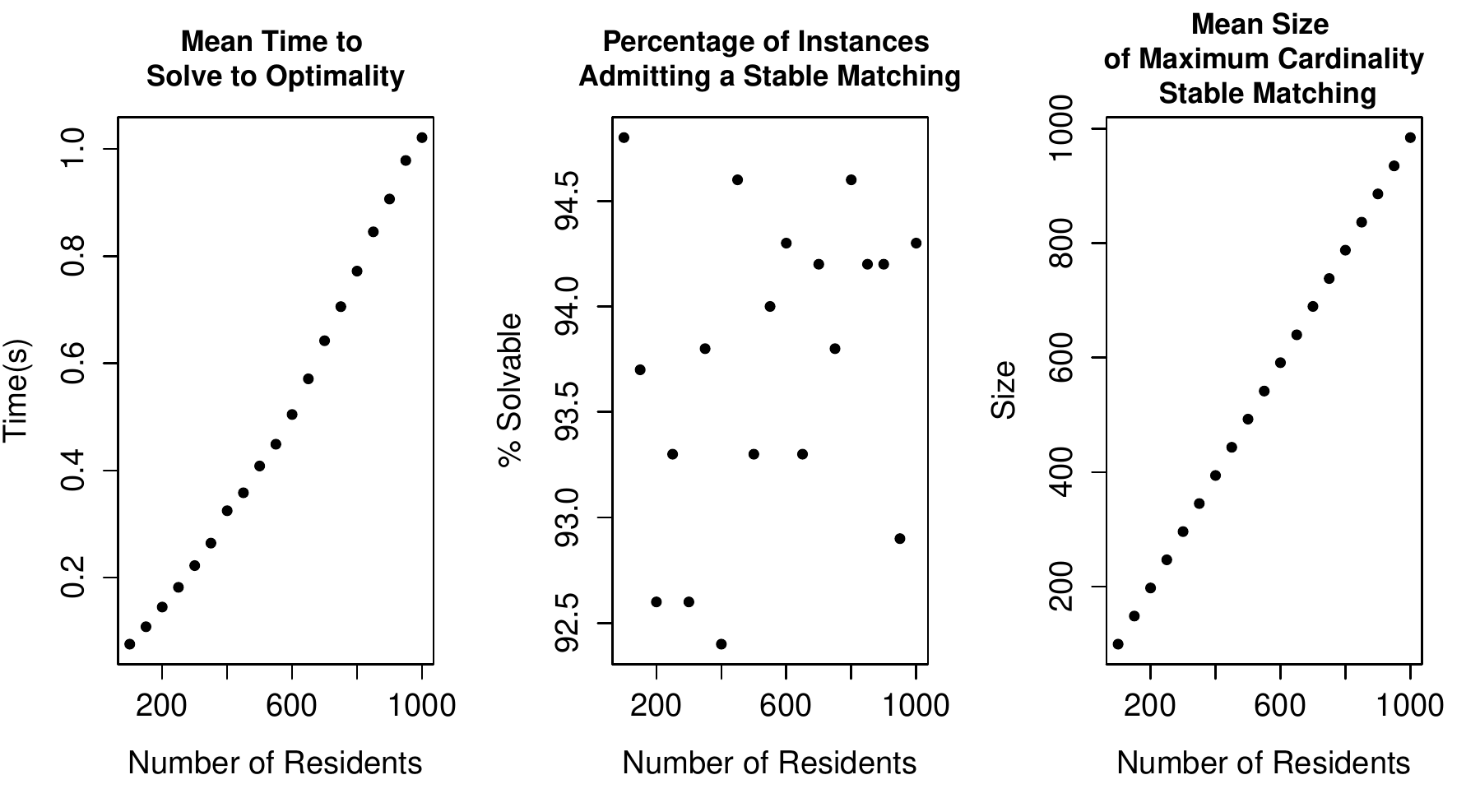}
\]
\caption{Empirical Results in Experiment 1.}
\label{image:varyingsizeplots}
\end{figure}

\subsubsection{Experiment 2}

In our second experiment, we report on results obtained as we increased the the percentage of residents involved in couples while maintaining the same total number of residents, hospitals and posts. For various values of $x ~ (0 \le x \le 250)$ in increments of $25$, $1000$ randomly generated instances were created containing $1000$ residents, $x$ couples (and hence $1000 - 2x$ single residents) and $100$ hospitals with $1000$ available posts which were unevenly distributed amongst the hospitals. The mean time taken to find a maximum cardinality stable matching or report that no stable matching existed in each instance is plotted in Figure \ref{image:varyingcouplesplots} for all values of $x$. In similar fashion to Figure \ref{image:varyingsizeplots}, Figure \ref{image:varyingcouplesplots} also shows charts displaying the percentage of instances encountered which admitted a stable matching and the mean size of the maximum cardinality stable solution for all values of $x$.

The data in Figure \ref{image:varyingcouplesplots} shows that the mean time taken to find a maximum cardinality stable matching tends to increased as we increased the number of residents in the instances involved in couples. Further, Figure {\ref{image:varyingcouplesplots} shows that the percentage of {\sc hrc} instances that admitted a stable matching fell as the percentage of the residents in the instances involved in couples increased. When $50\%$ of the residents in the instance were involved in a couple we found that $832$ of the $1000$ instances admitted a stable matching. Figure {\ref{image:varyingcouplesplots} also shows that as the percentage of the residents in the instances involved in couples increased the mean size of a maximum cardinality stable matching supported by the instances tended to decrease.

We conjecture that, as the number of couples increases while the number of residents remains the same, the cardinality of the set of stable matchings supported decreases. This is suggested by the fact that the number of instances where the set of stable matchings is of cardinality zero increases. Hence this in turn suggests that the range of sizes of the stable matchings, across the set of all stable matchings admitted, contracts, and thus the size of the largest matching in that set decreases

%The data in Figure \ref{image:varyingcouplesplots} shows that the mean time taken to find a maximum cardinality stable matching tends to increased as we increased the number of residents in the instances involved in couples. Further, Figure {\ref{image:varyingcouplesplots} shows that the percentage of {\sc hrc} instances that admitted a stable matching fell as the percentage of the residents in the instances involved in couples increased. When $50\%$ of the residents in the instance were involved in a couple we found that $832$ of the $1000$ instances admitted a stable matching. Figure {\ref{image:varyingcouplesplots} also shows that as the percentage of the residents in the instances involved in couples increased the mean size of a maximum cardinality stable matching supported by the instances tended to decrease.

\begin{figure}
\[
\includegraphics[width=\textwidth, height=4.55cm]{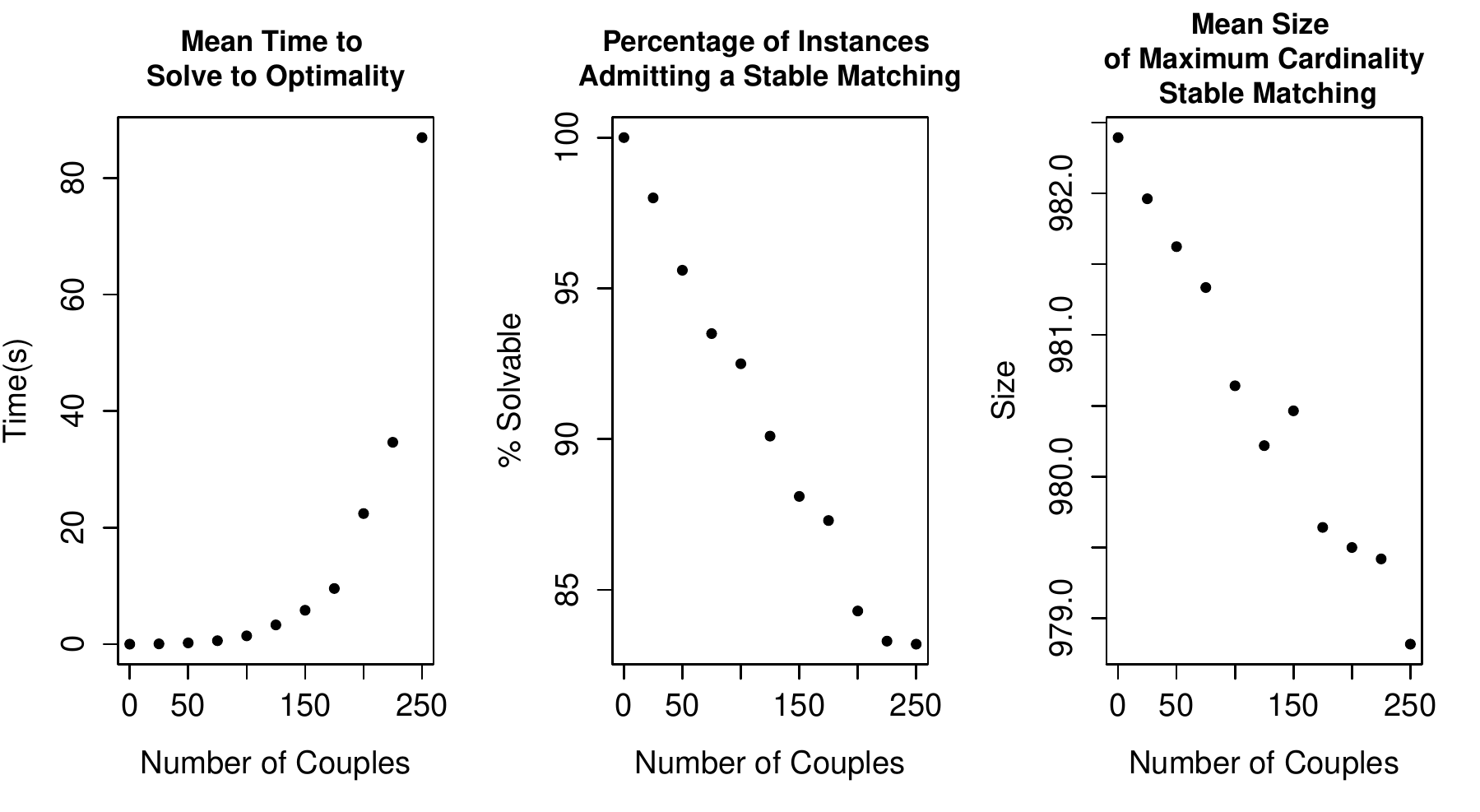}
\]
\caption{Empirical Results in Experiment 2.}
\label{image:varyingcouplesplots}
\end{figure}

\subsubsection{Experiment 3}

In our third experiment, we report on data obtained as we increased the number of hospitals in the instance while maintaining the same total number of residents, couples and posts. For various values of $x ~ (25 \le x \le 500)$ in increments of $25$, $1000$ randomly generated instances of size $1000$ were created consisting of $1000$ residents in total, $x$ hospitals, $100$ couples (and hence $800$ single residents) and $1000$ available posts which were unevenly distributed amongst the hospitals. The time taken to find a maximum cardinality stable matching or report that no stable matching existed in each instance is plotted in Figure \ref{image:varyinghospitalsplots} for all values of $x$. Again, in similar fashion to Figure \ref{image:varyingsizeplots}, Figure \ref{image:varyinghospitalsplots} also shows charts displaying the percentage of instances encountered which admitted a stable matching and the mean size of a maximum cardinality stable solution for all values of $x$.

\begin{figure}
\[
\includegraphics[width=\textwidth, height=4.55cm]{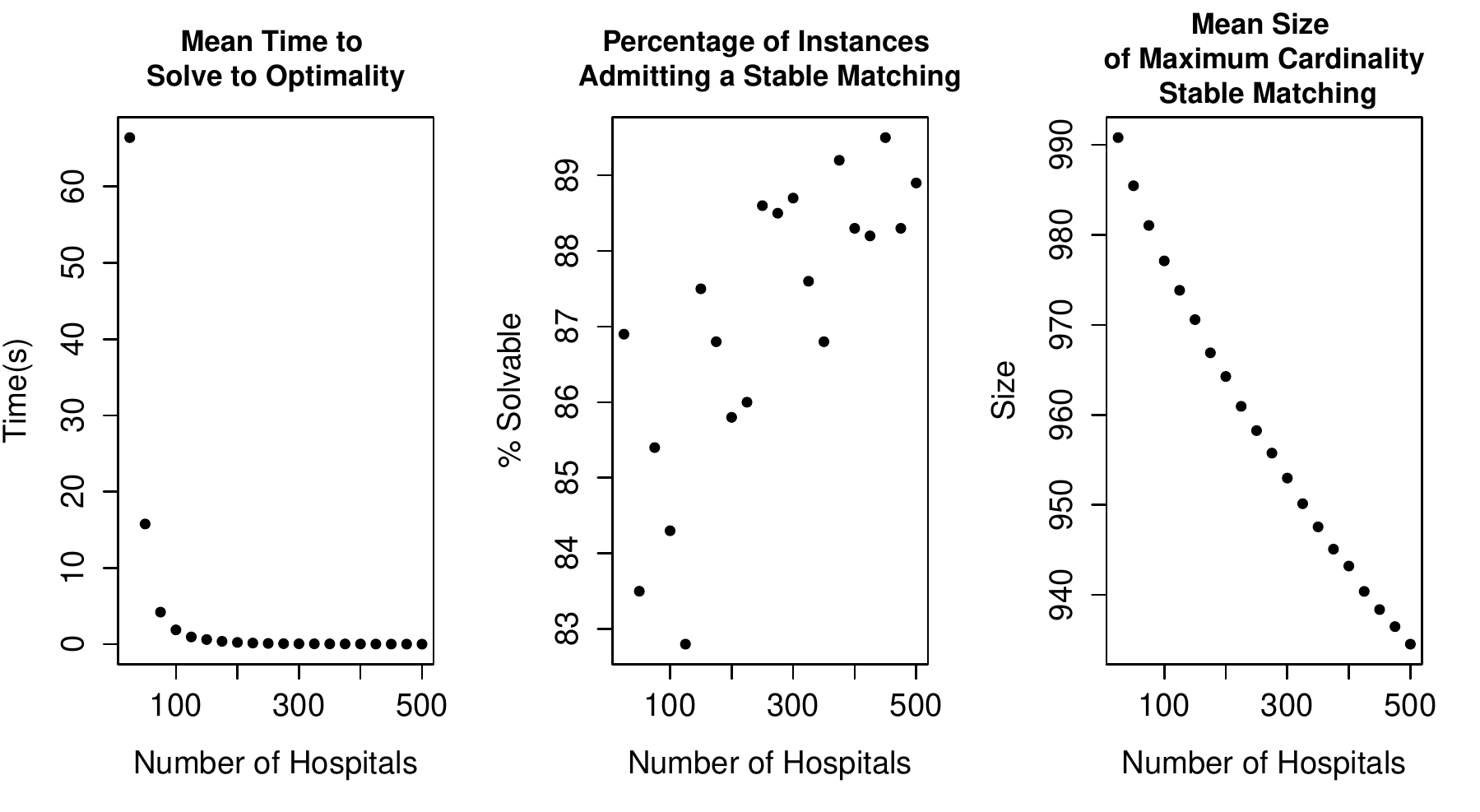}
\]
\caption{Empirical Results in Experiment 3.}
\label{image:varyinghospitalsplots}
\end{figure}

Figure \ref{image:varyinghospitalsplots} shows that the mean time taken to find a maximum cardinality stable matching tended to decrease as we increased the number of hospitals in the instances.  We believe that this is due to the hospitals' preference lists becoming shorter, thereby reducing the model's complexity. The data in Figure {\ref{image:varyinghospitalsplots} also shows that the percentage of {\sc hrc} instances admitting a stable matching appeared to increase with the number of hospitals involved in the instance.  We conjecture that this is because, as each hospital has a smaller number of posts, it is more likely to become full, and therefore less likely to be involved in a blocking pair due to being under-subscribed. Finally, the data shows that as the number of hospitals in the instances increased, the mean size of a maximum cardinality stable matching supported by the instances tended to decrease.  This can be explained by the fact that, as the number of hospitals increases but the residents' preference list lengths and the total number of posts remain constant, the number of posts per hospital decreases. Hence the total number of posts among all hospitals on a resident's preference list decreases.

%Figure \ref{image:varyinghospitalsplots} shows that the mean time taken to find a maximum cardinality stable matching tended to decrease as we increased the number of hospitals in the instances. The data in Figure {\ref{image:varyinghospitalsplots} also shows that the percentage of {\sc hrc} instances admitting a stable matching did not appear to be correlated with the number of hospitals involved in the instance and further that as the number of hospitals in the instances increased the mean size of a maximum cardinality stable matching supported by the instances tended to decrease.

\subsubsection{Experiment 4}

In our last experiment, we report on data obtained as we increased the length of the individual preference lists for the residents in the instance while maintaining the same total number of residents, couples, hospitals and posts. For various values of $x ~ (3 \le x \le 12)$ in increments of $1$, $1000$ randomly generated instances of size $1000$ were created consisting of $1000$ residents in total, $100$ hospitals, $100$ couples (and hence $800$ single residents) and $1000$ available posts which were unevenly distributed amongst the hospitals. The time taken to find a maximum cardinality stable matching or report that no stable matching existed in each instance is plotted in Figure \ref{image:varyingpreflistlengthplots} for all values of $x$. Again, in similar fashion to Figure \ref{image:varyingsizeplots}, Figure \ref{image:varyingpreflistlengthplots} also shows charts displaying the percentage of instances encountered admitting a stable matching and the mean size of a maximum cardinality stable solution for all values of $x$.

\begin{figure}
\[
\includegraphics[width=\textwidth, height=4.55cm]{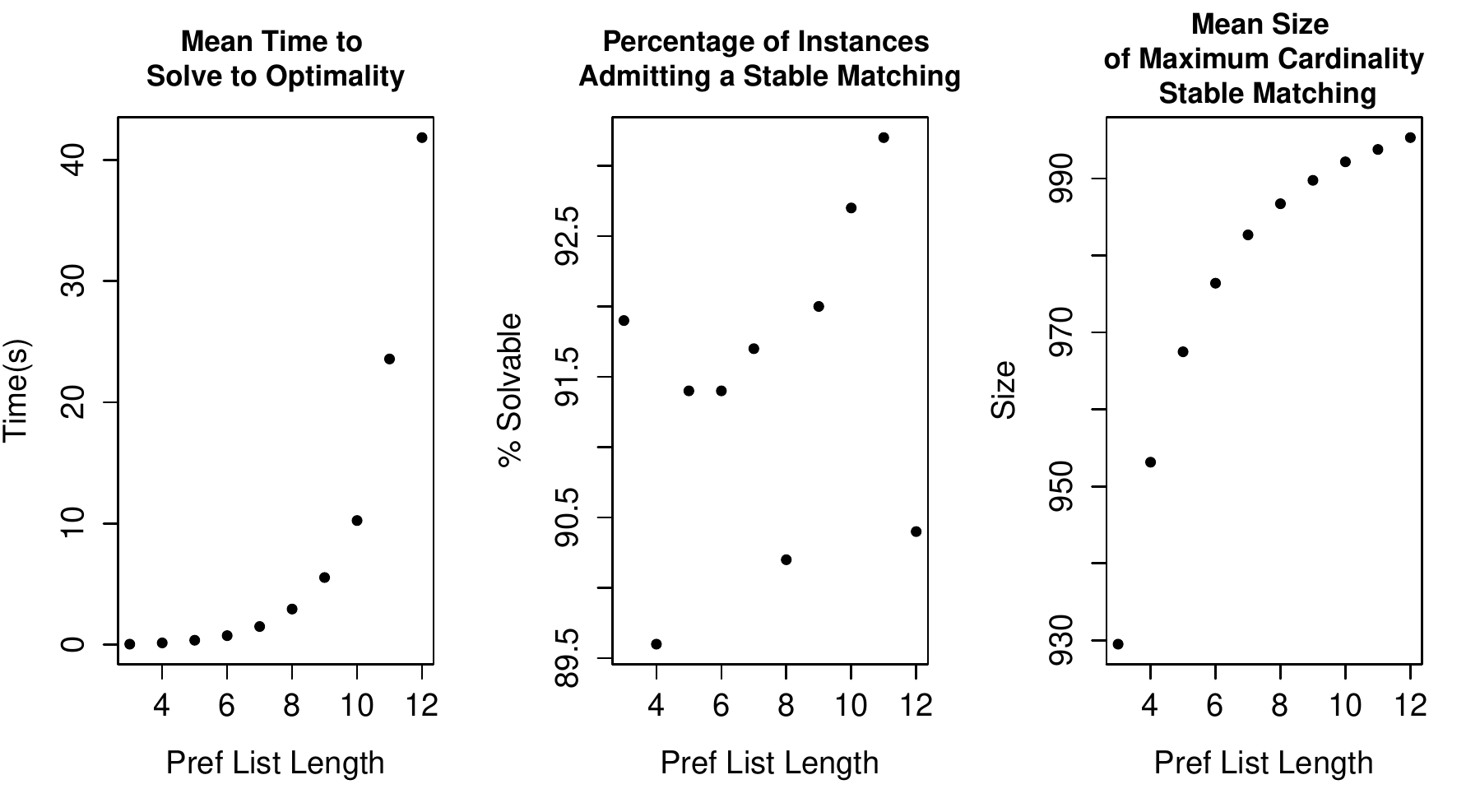}
\]
\caption{Empirical Results in Experiment 4.}
\label{image:varyingpreflistlengthplots}
\end{figure}

Figure \ref{image:varyingpreflistlengthplots} shows that the mean time taken to find a maximum cardinality stable matching increased as we increased the length of the individual residents' preference lists in the instances. The data in Figure {\ref{image:varyingpreflistlengthplots} also shows that the percentage of {\sc hrc} instances admitting a stable matching did not appear to be correlated with the length of the individual residents' preference lists in the instances and further that as the length of the individual residents' preference lists in the instances increased, the mean size of a maximum cardinality stable matching supported by the instances also tended to increase.  The first and third of these phenomena would seem to be explained by the fact that the underlying graph is simply becoming more dense.

%Figure \ref{image:varyingpreflistlengthplots} shows that the mean time taken to find a maximum cardinality stable matching increased as we increased the length of the individual residents' preference lists in the instances. The data in Figure {\ref{image:varyingpreflistlengthplots} also shows that the percentage of {\sc hrc} instances admitting a stable matching did not appear to be correlated with the length of the individual residents' preference lists in the instances and further that as the length of the individual residents' preference lists in the instances increased, the mean size of a maximum cardinality stable matching supported by the instances also tended to increase.

\subsection{Performance of the model with real world data}

In the context of the SFAS matching scheme, hospitals' preferences were derived from the  residents' \emph{scores}, where a junior doctor's score is derived from their previous academic performance. If two residents received the same score, they would be tied in a hospital's preference list. Thus, the underlying SFAS matching problem may be correctly modelled by {\sc hrct}. 

Hence, we further extended our implementation in the fashion described in Section \ref{section:IPModelsHRCT} to find maximum cardinality stable matchings in instances of {\sc hrct} and were able to find optimal solutions admitted by the real data obtained from the SFAS application. The worst case time and space complexity and the number of constraints in the {\sc hrc} model is the same as the {\sc hr} model.  The maximum cardinality stable matchings obtained in the SFAS application for the three years to 2012 are shown in Table \ref{table:SFASresults} alongside the time taken to find these matchings.

\begin{table}[ht]
\centering

\begin{tabulary}{\textwidth } { L  C | C | C | C | C | C } 
%
%\begin{tabular}{ p{0.1\textwidth} | p{0.125\textwidth} | p{0.125\textwidth} | p{0.125\textwidth} | p{0.125\textwidth} | p{0.2\textwidth} | p{0.125\textwidth} } 
% \hline
% \hline
 & Number of Residents & Number of Couples & Number of Hospitals & Number of Posts &Max Cardinality Stable Matching& Time to Solution \\
   
\hline 
2012 & 710 & 17 & 52 & 720 & 681 & 9.62s \\ %[0.5ex] 
2011 & 736 & 12 & 52 & 736 & 688 & 10.41s \\ %[0.5ex] 
2010 & 734 & 20 & 52 & 735 & 681 & 33.92s \\ %[0.5ex] 
%\hline 
\end{tabulary}
\caption{Results obtained from the previous 3 years SFAS data.} 
\label{table:SFASresults} 
\end{table}

\section{MM-stability and BIS-stability}
\label{section:stabilityComparison}

In Section \ref{section:MMnotequalBIS} we demonstrate that BIS-stability and MM-stability are not equivalent by means of a pair of example instances, the first of which admits an MM-stable matching, but no BIS-stable matching and the second of which admits a BIS-stable matching, but no MM-stable matching. Then, in Section \ref{section:cloning} we present a cloning methodology for {\sc hrc} that can be used to construct an instance of one-to-one {\sc hrc} from an instance of many-to-one {\sc hrc} such that the MM-stable matchings in the many-to-one instance are in correspondence to the MM-stable matchings in the one-to-one instance. We prove further that this cloning method is not applicable under BIS-stability.

\subsection{Distinction between MM-stability and BIS-stability}
\label{section:MMnotequalBIS}

MM-stability and the BIS-stability are not equivalent. We demonstrate this by means of the two instances shown in Figure \ref{instance:BISbutnoMM} and Figure \ref{instance:MMbutnostableBIS}. Consider the instance of HRC shown in Figure \ref{instance:BISbutnoMM} where $h$ has capacity 2. The matching $M=\{(r_3,h)\}$ is BIS-stable, but the instance admits no MM-stable matching.

\begin{figure}[h]
\begin{center}
\begin{tabular}[t]{lcrcccl}

\multicolumn{7}{c}{Residents} \\
\hline
\\

 & $(r_1, r_2)$		&$ ~~~	 ~~~ $:&	$(h, h)$ & 	& 	&				 	\\

& $r_{3}$		&	$ ~~~	 ~~~ $:&		$h$ & 	& 		&		 	\\

\\

\multicolumn{7}{c}{Hospitals} \\
\hline
\\

& $h$ 	&	:$ ~~~	2 ~~~ $: &		$r_1$		&		$r_3~~~$	&	 $r_2$ &  \\

\end{tabular}
\end{center}
\caption{An instance of HRC which admits a BIS-stable matching but admits no MM-stable matching.}
\label{instance:BISbutnoMM}
\end{figure}

\begin{figure}[h]
\begin{center}
\begin{tabular}[t]{lcrccccl}

\multicolumn{8}{c}{Residents} \\
\hline
\\

 & $(r_1, r_2)$		&	:		&	$(h_1, h_1)$ & 	& 					& &	\\
 
 & $(r_3, r_4)$		&	:		&	$(h_1, h_1)$ & 	 $(h_1, h_2)$		&			& &	\\

\\

\multicolumn{8}{c}{Hospitals} \\
\hline
\\

& $h_1$ 	&:$ ~~~	2 ~~~ $:& $r_3$	&	$r_1~~~~~$		&	$r_2~~~~~~$		& $r_4~~~~~$ &    \\

& $h_2$ 	&:$ ~~~	1 ~~~ $:& $r_4$				&			&	  &  & \\

\end{tabular}
\end{center}
\caption{An instance of HRC which admits an MM-stable matching that but admits no BIS-stable matching.}
\label{instance:MMbutnostableBIS}
\end{figure}

Consider the instance of {\sc hrc} shown in Figure \ref{instance:MMbutnostableBIS} due to Irving \cite{RI11}, where $h_1$ has capacity 2 and $h_2$ has capacity 1. The instance admits three distinct matchings, namely $M_1=\{(r_1,h_1),$ $ (r_2,h_1)\}$, $M_2=\{(r_3,h_1),$ $ (r_4,h_1)\}$ and $M_3=\{(r_3,h_1), $ $ (r_4,h_2)\}$. $M_2$ is MM-stable. However, $M_1$ is BIS-blocked by $(r_3, r_4)$ with $(h_1, h_2)$, $M_2$ is BIS-blocked by $(r_1, r_2)$ with $(h_1, h_1)$, and $M_3$ is BIS-blocked by $(r_3, r_4)$ with $(h_1, h_1)$.

\subsection{A hospital cloning method for {\sc hrc} under MM-stability}

\label{section:cloning}

\def\myline{
 
    %\vspace{-0.9em}
    \line(0,1){3.5}
    \hspace{-0.71em}
     
     \line(1,0){50}
 			\line(0,1){3.5}

}

\def\shortline{
 
    %\vspace{-0.9em}
    \line(0,1){3.5}
    \hspace{-0.71em}
     
     \line(1,0){30}
 			\line(0,1){3.5}

}

\def\Xline{
 
    %\vspace{-0.9em}
    \line(0,1){3.5}
    \hspace{-0.71em}
     \line(1,0){17.5}
     
     ~X
  
     \line(1,0){17.5}
 			\line(0,1){3.5}
}

\def\X2line{
 
    %\vspace{-0.9em}
    \line(0,1){3.5}
    \hspace{-0.71em}
     \line(1,0){9.5}
     
     \hspace{-0.3em}
     X
     
    \hspace{-0.55em}
     \line(1,0){12.5}

     \hspace{-0.30em}
     X
  
  	\hspace{-0.51em}
     \line(1,0){8.5}
 			\line(0,1){3.5}
}

For an arbitrary instance $I$ of {\sc hr} in which the hospitals may have capacity greater than one, Gusfield and Irving \cite{GI89} describe a method of constructing a corresponding instance, $I^{\prime }$ of {\sc hr} in which all of the hospitals have capacity one, such that a stable matching in $I$ corresponds to a stable matching in $I^{\prime }$ and vice versa. In this section we describe a method for producing an instance $I^{\prime }$ of {\sc hrc}, in which all of the hospitals have capacity one, from an arbitrary instance $I$ of {\sc hrc}, in which the hospitals may have capacity greater than one, such that an MM-stable matching in $I$ corresponds to an MM-stable matching in $I^{\prime }$ and vice versa. We show that this correspondence breaks down in the case of BIS-stability.

Let $I$ be an instance of {\sc hrc} with residents $R = \{r_1, r_2,\dots ,r_{n_1}\}$ and hospitals $H= \{h_1, h_2,\dots ,h_{n_2}\}$. Without loss of generality, suppose residents $r_1, r_2,\ldots ,r_{2c}$ are in couples. Again, without loss of generality, suppose that the couples are $(r_{2i-1}, r_{2i})$  $(1\leq i\leq c)$. Let the single residents be $r_{2c+1}, r_{2c+2}\ldots r_{n_1}$.

Suppose each single resident $r_i\in R$ has a preference list of length $l(r_i)$ consisting of individual hospitals $h_j\in H$. Suppose also that the joint preference list of a couple $c_i = (r_{2i-1}, r_{2i})$ is a list, of length $l(c_i)$, of hospital pairs. Assume each hospital $h_j\in H$ $(1\leq j \leq n_2)$  has a preference list of individual residents $r_i\in R$ of length $l(h_j)$. Let $c_j$ denote the capacity of hospital $h_j\in H$ $(1\leq j \leq n_2)$, the number of available posts it has to match with residents. 

To construct an equivalent instance $I^{\prime }$ of one-to-one {\sc hrc} from $I$ we create $c_j$ clones of each hospital $h_j\in H$, namely $h_{j,1}, h_{j,2}\ldots h_{j,{c_j}}$, each of unitary capacity and each representing one of the individual posts in $h_j$. In the preference list of a single resident $r_i$ we replace each incidence of $h_j$ with the following sequence of hospitals $h_{j,1}, h_{j,2}\ldots h_{j, {c_j}}$.

For each couple $c_i$,  we replace each $(h_{j_1}, h_{j_2})$ (where $j_1 \neq j_2$) in $(r_{2i-1}, r_{2i})$'s joint preference list with the following sequence of hospital pairs:

\noindent \small{
$$L_1 = (h_{j_1,1}, h_{j_2,1}), (h_{j_1, 2}, h_{j_2,1})\ldots (h_{j_1,c_{j_1}}, h_{j_2,1}), $$ $$(h_{j_1,1}, h_{j_2,2}), (h_{j_1,2}, h_{j_2,2})\ldots  (h_{j_1,c_{j_1}}, h_{j_2,2})\ldots (h_{j_1,c_{j_1}}, h_{j_2,c_{j_2}})
$$}
\normalsize
which contains all of the possible pairings of the individual clones of $h_{j_1}$ and $h_{j_2}$. Further in $I^{\prime }$ we replace each $(h_{j}, h_{j})$ in $(r_{2i-1}, r_{2i})$'s joint preference list with the following sequence of hospital pairs: 

\noindent \small{
$$L_2 = (h_{j, 2}, h_{j,1}), (h_{j, 3}, h_{j,1})\ldots (h_{j,c_j}, h_{j,1}), (h_{j,1}, h_{j,2}), (h_{j,3}, h_{j,2})\ldots (h_{j,c_j}, h_{j,2})\ldots (h_{j,c_j -1}, h_{j,c_j})
$$}
\normalsize
where $ \{ (h_{j,x} , h_{j,y} ) : x = y \} \cap L_2 = \emptyset$. Thus $L_2$ contains all possible pairings of distinct individual clones of $h_{j}$. We now show that MM-stable matchings are preserved under this correspondence.

%Since no preferences are expressed by a resident couple in $I^{\prime }$ for a hospital pair $(h_{j,k}, h_{j,k})$, only blocking assignments of the type described in 1, 2 and 3(a) are possible in $I^{\prime }$ under the MM stability definition.

\begin{lemma}
\label{lemma:cloningunderMMworks}

%$I^{\prime }$ admits a stable matching under the GI definition of stability if and only if $I$ admits a stable matching under the MM definition of stability.

$I$ admits an MM-stable matching if and only if $I^{\prime }$ does.

\end{lemma}
\begin{proof}

Let $M$ be an MM-stable matching in $I$. We construct an MM-stable matching $M^{\prime }$ in $I^{\prime }$ as follows.

Take any hospital $h_j \in H$ and list its assignees as $r_{i,1}, r_{i,2}\ldots r_{i,{t_j}}$ where $t_j\leq c_j$. Assume without loss of generality that $rank (h_j, r_{i,1}) < rank (h_j, r_{i,2})\ldots < rank (h_j, r_{i,{t_j}})$. For each $k ~ (1\leq k\leq t_j)$ add $(r_{i,k}, h_{j,k})$ to $M^{\prime }$.

All single residents $r_{i}$ who are assigned to a hospital $h_j$ in $M$ are assigned to an acceptable hospital clone $h_{j,k}$ in $M^{\prime }$, for some $k ~ (1\leq k\leq c_j)$. All couples $(r_{2i-1}, r_{2i})$ jointly assigned to some $(h_{j_1}, h_{j_2})$ in $M$ are jointly assigned in $M^{\prime }$ to $(h_{j_1, k_1}, h_{j_2, k_2})$ for some $k_1, k_2 ~ (1\leq k_1 \leq c_{j_1}, 1\leq k_2\leq c_{j_2})$. Since $(h_{j_1}, h_{j_2})$ is an acceptable pair of hospitals for $(r_{2i-1}, r_{2i})$ in $M$, $(h_{j_1, k_1}, h_{j_2, k_2})$ must be an acceptable pair of hospital clones for $(r_{2i-1}, r_{2i})$ in $M^{\prime }$ (note that if $j_1 = j_2$ then $k_1\neq k_2$). Therefore $M^{\prime }$ is a matching in $I^{\prime }$.

We now require to prove that $M^{\prime }$ is MM-stable in $I^{\prime }$. Suppose not. Then there is some MM-blocking pair for $M^{\prime }$ in $I^{\prime }$. Since no preferences are expressed by a couple in $I^{\prime }$ for a hospital pair $(h_{j,k}, h_{j,k})$ consisting of two identical clones, only blocking pairs of Types 1, 2 and 3(a) shown in Definition \ref{stability:MM} are possible in $I^{\prime }$.

\noindent\emph{Case (1)}: A single resident $r_i$ and hospital clone $h_{j,k}$ MM-block $M^{\prime }$ in $I^{\prime }$. 
%
%\begin{center}
%\begin{tabular}[t]{lcrcccl}
%
%
%$r_{i}$ 	&	:	&	\myline	&	$h_{j,k}$		&		\myline	&	$M^{\prime }(r_i)$ & \myline \\
%
%\\
%
%$h_{j,k}$		&	:	&	\myline	&	$r_{i}$ &	 	\myline	&	 $M^{\prime }(h_{j,k})$		& \myline	\\
%
%\\
%
%\end{tabular}
%\end{center}

Hence, in $I^{\prime }$ resident $r_i$ is unassigned or prefers $h_{j,k}$ to $M^{\prime }(r_i)$ and also $h_{j,k}$ is under-subscribed or prefers $r_i$ to $M^{\prime }(h_{j,k})$. It is either the case that $M^{\prime }(r_i) = h_{j,l}$ for some $l ~ (1\leq l\leq c_j)$ or $M^{\prime }(r_i) \neq h_{j,l}$ for all $l ~ (1\leq l\leq c_j)$ in $I^{\prime }$. 

\emph{(i)} If $M^{\prime }(r_i) \neq h_{j,l}$ for all $l ~ (1\leq l\leq c_j)$, then, by construction, this would imply that in $M$, $r_i$ is unassigned or prefers $h_{j}$ to $M(r_i)$ and $h_{j}$ is also either under-subscribed or prefers $r_i$ to some member of $M(h_{j})$, and thus $(r_i, h_j)$ forms an MM-blocking pair of $M$ in $I$, a contradiction. 

\emph{(ii)} If $M^{\prime }(r_i) = h_{j,l}$ for some $l ~ (1\leq l\leq c_j)$, then $r_i$ must prefer $h_{j,k}$ to $h_{j,l}$ and therefore $k<l$. So $h_{j,k}$ is assigned in $M^{\prime }$ to some $r_p$ such that $h_{j,k}$ prefers $r_p$ to $r_i$. Hence $(r_i, h_{j,k})$ cannot MM-block $M^{\prime }$ in $I^{\prime }$. 

\noindent\emph{Case (2)}: Resident couple $(r_{2i-1}, r_{2i})$ MM-blocks $M^{\prime }$ with $(h_{j_1, k_1}, h_{j_2, k_2})$ in $I^{\prime }$ for some $k_1, k_2 ~ ( 1\leq k_1 \leq c_{j_1}, 1\leq k_2 \leq c_{j_2}) $ where $M^{\prime }(r_{2i-1}) = h_{j_1, k_1}$ or $M^{\prime }(r_{2i}) = h_{j_2, k_2}$

%
%\begin{center}
%\begin{tabular}[t]{lcrcccl}
%
%
%$(r_{2i-1}, r_{2i})$ 	&	:	&	\myline	&	$(h_{j_1, k_1}, h_{j_2, k_2})$		&		\myline	&	($M^{\prime }(r_{2i-1}), h_{j_2, k_2})$	&	\myline  \\
%
%\\
%
%
%\end{tabular}
%\end{center}
%
%
%\begin{center}
%\begin{tabular}[t]{llcrcccl}
%
%
% & $h_{j_1, k_1}$		&	:	&	\myline	&	$r_{2i-1}$ 	&	 	\myline	&	$M^{\prime }(h_{j_1, k_1})$	& \myline 	\\
%
%
%
%\\
%
%
%\end{tabular}
%\end{center}

In this case either $j_1 \neq j_2$ or $j_1 = j_2$. Consider first the case where $j_1 \neq j_2$ and $(r_{2i-1}, r_{2i})$ prefers $(h_{j_1, k_1}, h_{j_2, k_2})$ to $( h_{j_1, k_1}, M^{\prime }(r_{2i}) )$ and $h_{j_1, k_1}$ is either under-subscribed in $M^{\prime }$ or prefers $r_{2i-1}$ to $M^{\prime }(h_{j_1, k_1})$. It is either the case that $M^{\prime }(r_{2i-1}) = h_{j_1,l}$ for some $l ~ (1\leq l\leq c_{j_1})$ or $M^{\prime }(r_{2i-1}) \neq h_{j_1,l}$ for all $l ~ (1\leq l\leq c_{j_1})$.

\emph{(i)} If $M^{\prime }(r_{2i-1}) \neq h_{j_1,l}$ for all $l ~ (1\leq l\leq c_{j_1})$ then if $h_{j_1, k_1}$ is under-subscribed in $M^{\prime }$, $h_{j_1}$ must be under-subscribed in $M$ and $(r_{2i-1},r_{2i})$ must MM-block $M$ in $I$ with $(h_{j_1}, h_{j_2})$, a contradiction. Also, if $h_{j_1, k_1}$ prefers $r_{2i-1}$ to $M^{\prime }(h_{j_1, k_1})$ then $h_{j_1}$ must prefer $r_{2i-1}$ to some member of $M(h_{j_1})$ in $M$. Hence $(r_{2i-1}, r_{2i})$ must MM-block $M$ with $(h_{j_1}, h_{j_2})$ in $I$ , a contradiction.

\emph{(ii)} If $M^{\prime }(r_{2i-1}) = h_{j_1,l}$ for some $l ~ (1\leq l\leq c_{j_1})$ then $(r_{2i-1}, r_{2i})$ jointly prefers $(h_{j_1, k_1}, h_{j_2, k_2})$ to $(h_{j_1, l}, $ $ h_{j_2, k_2})$ and therefore $k_1 < l$. So $h_{j_1, k_1}$ is matched in $M^{\prime }$ to some $r_p$ such that $h_{j_1, k_1}$ prefers $r_p$ to $r_{2i-1}$ and therefore $(r_{2i-1}, r_{2i})$ cannot MM-block $M^{\prime }$ with $(h_{j_1, k_1}, h_{j_2, k_2})$ in $I^{\prime }$, a contradiction.

Now consider the alternate case where $j_1 = j_2$ and $(r_{2i-1}, r_{2i})$ prefers $(h_{j_1, k_1}, h_{j_1, k_2})$ where $k1 \neq k2$ to $( h_{j_1, k_1}, M^{\prime }(r_{2i}) )$ and $h_{j_1, k_1}$ is either under-subscribed in $M^{\prime }$ or prefers $r_{2i-1}$ to $M^{\prime }(h_{j_1, k_1})$. It is either the case that $M^{\prime }(r_{2i-1}) = h_{j_1,l}$ for some $l ~ (1\leq l\leq c_{j_1})$ or $M^{\prime }(r_{2i-1}) \neq h_{j_1,l}$ for all $l ~ (1\leq l\leq c_{j_1})$. 

\emph{(i)} If $M^{\prime }(r_{2i-1}) \neq h_{j_1,l}$ for all $l ~ (1\leq l\leq c_{j_1})$ then if $h_{j_1, k_1}$ is under-subscribed in $M^{\prime }$, $h_{j_1}$ must be under-subscribed in $M$ and $(r_{2i-1},r_{2i})$ must MM-block $M$ in $I$ with $(h_{j_1}, h_{j_1})$, a contradiction. Further, if $h_{j_1, k_1}$ prefers $r_{2i-1}$ to $M^{\prime }(h_{j_1, k_1})$ then $h_{j_1}$ must prefer $r_{2i-1}$ to some member of $M(h_{j_1})$ other than $r_{2i-1}$ in $M$. Hence $(r_{2i-1}, r_{2i})$ must MM-block $M$ with $(h_{j_1}, h_{j_1})$ in $I$ , a contradiction.

\emph{(ii)} If $M^{\prime }(r_{2i-1}) = h_{j_1,l}$ for some $l ~ (1\leq l\leq c_{j_1})$ then $(r_{2i-1}, r_{2i})$ jointly prefers $(h_{j_1, k_1}, h_{j_1, k_2})$ where $k_1 \neq k_2$ to $(h_{j_1, l}, $ $ h_{j_1, k_2})$ and therefore $k_1 < l$. So $h_{j_1, k_1}$ is assigned in $M^{\prime }$ to some $r_p$ such that $h_{j_1, k_1}$ prefers $r_p$ to $r_{2i-1}$ and therefore $(r_{2i-1}, r_{2i})$ cannot MM-block $M^{\prime }$ with $(h_{j_1, k_1}, h_{j_1, k_2})$ in $I^{\prime }$, a contradiction.

A similar argument may be applied when considering the case that $(r_{2i-1}, r_{2i})$ prefers $(h_{j_1, k_1},$ $ h_{j_2, k_2})$ to $(M^{\prime }(r_{2i-1}),$ $h_{j_2, k_2})$ in $I^{\prime }$.

\noindent\emph{Case (3a)}: Resident couple $(r_{2i-1}, r_{2i})$ MM-blocks $M^{\prime }$ in $I^{\prime }$ in $I^{\prime }$ with $(h_{j_1, k_1}, h_{j_2, k_2})$ (where $j_1 \neq j_2$)  for some $k_1, k_2 ~ ( 1\leq k_1 \leq c_{j_1}, 1\leq k_2 \leq c_{j_2}) $ where $M^{\prime }(r_{2i-1}) \neq h_{j_1 , k_1}$ and $M^{\prime }(r_{2i}) \neq h_{j_2 , k_2}$.
%
%\begin{center}
%\begin{tabular}[t]{lcrcccl}
%
%
%$(r_{2i-1}, r_{2i})$ 	&	:	&	\myline	&	$(h_{j_1, k_1}, h_{j_2, k_2})$		&		\myline	&	$(M^{\prime }(r_{2i-1}), M^{\prime }(r_{2i}))$	&	\myline  \\
%
%\\
%
%\end{tabular}
%\end{center}
%
%
%\begin{center}
%\begin{tabular}[t]{llcrcccl}
%
%
%
%
% & $h_{j_1, k_1}$		&	:	&	\myline	&	$r_{2i-1}$ 	&	 	\myline	&	$M^{\prime }(r_{2i-1})$	&  \myline	\\
%
%
%\textbf{and} & $h_{j_2, k_2}$		&	:	&	\myline	&	$r_{2i}$ 	&	 	\myline	&		$M^{\prime }(r_{2i})$		& \myline	\\
%
%\\
%
%
%\end{tabular}
%\end{center}
Hence, resident couple $(r_{2i-1}, r_{2i})$ is either jointly unassigned in $M^{\prime }$ or jointly assigned in $M^{\prime }$ to a worse hospital pair than $(h_{j_1, k_1}, h_{j_2, k_2})$ (where $j_1 \neq j_2$) for some $k_1, k_2 ~ ( 1\leq k_1 \leq c_{j_1}, 1\leq k_2 \leq c_{j_2}) $ and also each of hospitals $h_{j_1, k_1}$ and $ h_{j_2, k_2}$ is either under-subscribed in $M^{\prime }$ or assigned in $M^{\prime }$ to a worse partner than $r_{2i-1}$ and $r_{2i}$ respectively. 

Assume that $M^{\prime }(r_{2i-1}) \neq h_{j_1,l}$ for all $l ~ (1\leq l \leq c_{j_1})$ and also $M^{\prime }(r_{2i}) \neq h_{j_2,l}$ for all $l  ~ (1\leq l \leq c_{j_2})$. From the construction, this means that in $M$, $(r_{2i-1}, r_{2i})$ is unassigned or jointly prefers $(h_{j_1}, h_{j_2})$ to $M(r_{2i-1}, r_{2i})$ and also each of $h_{j_1}$ and $h_{j_2}$ is either under-subscribed or prefers $r_{2i-1}$ and $r_{2i}$ respectively to some member of $M(h_{j_1})$ and $M(h_{j_2})$ respectively. Hence $(r_{2i-1}, r_{2i})$ MM-blocks $M$ in $I$ with $(h_{j_1}, h_{j_2})$, a contradiction. 

Now, assume that $M^{\prime }(r_{2i-1}) = h_{j_1,l}$ for some $l ~ (1\leq l \leq c_j)$. Then $(r_{2i-1}, r_{2i})$ jointly prefers $(h_{j_1, k_1}, h_{j_2, k_2})$ to $(h_{j_1, l}, M^{\prime }(r_{2i}))$ and therefore $k_1 < l$. So $h_{j_1, k_1}$ is assigned in $M^{\prime }$ to some $r_p$ such that $h_{j_1, k_1}$ prefers $r_p$ to $r_{2i-1}$ and therefore $(r_{2i-1}, r_{2i})$ cannot MM-block $M^{\prime }$ in $I^{\prime }$ with $(h_{j_1, k_1}, h_{j_2, k_2})$. A similar argument may be applied in the case that $M^{\prime }(r_{2i}) = h_{j_2,l}$ for some $l  ~ (1\leq l \leq c_{j_2})$. 

We have therefore shown that an MM-stable matching in $M$ corresponds to an MM-stable matching in $M^{\prime }$. 

\medskip 

Conversely, let $M^{\prime} $ be an MM-stable matching in $I^{\prime} $. We construct an MM-stable matching $M$ in $I$ from $M^{\prime }$ as follows. For all pairs $(r_i, h_{j,k})$ in $M^{\prime } $ such that $h_{j} \in H$ and $1\leq k\leq c_j$, add $(r_{i}, h_{j})$ to $M$. Since there are at most $c_j$ clones of $h_j$ in $I^{\prime }$, $h_j$ cannot be oversubscribed.

All single residents $r_{i}$ who are assigned in $M^{\prime }$ to a hospital clone $h_{j,k}$ are assigned to an acceptable hospital $h_{j}$ in $M$. All couples $(r_{2i-1}, r_{2i})$ jointly assigned to some $(h_{j_1, k_1}, h_{j_2, k_2})$ in $M^{\prime }$ are jointly assigned in $M$ to $(h_{j_1}, h_{j_2})$ (note that possibly $j_1 = j_2$). Since $(h_{j_1, k_1}, h_{j_2, k_2})$ is an acceptable pair of hospital clones for $(r_{2i-1}, r_{2i})$ in $I^{\prime }$, $(h_{j_1}, h_{j_2})$ must be an acceptable pair of hospitals for $(r_{2i-1}, r_{2i})$ in $I$. Therefore $M$ is a matching in $I$.

We now require to prove that $M$ is MM-stable in $I$. Suppose not. Then there is an MM-blocking pair of $M$ in $I$.

\noindent\emph{Case (1)}: A single resident $r_i$ and hospital $h_{j}$ MM-block $M$ in $I$. 
%
%\begin{center}
%\begin{tabular}[t]{lcrcccl}
%
%
%$r_{i}$ 	&	:	&	\myline	&	$h_{j}$		&		\myline	& $M(r_i)$	& \myline \\
%
%
%
%\\
%
%
%$h_{j}$		&	:	&	\myline	&	$r_{i}$ &	 	\myline	&	$r_p \in M(h_j)$		& \myline	\\
%
%
%	
%\\
%
%
%\end{tabular}
%\end{center}
Then resident $r_i$ is unassigned or prefers $h_{j}$ to $M(r_i)$ and also $h_{j}$ is under-subscribed or prefers $r_i$ to some $r_p \in M(h_j)$ in $M$. If $h_j$ is under-subscribed in $M$ then some $h_{j,k}$ is under-subscribed in $M^{\prime }$. Otherwise, if $h_j$ prefers $r_i$ to some $r_p\in M(h_j)$ then let $k ~ (1\leq k\leq c_j)$ be such that $r_p = M^{\prime }(h_{j,k})$  From the construction, this would imply that in $I^{\prime }$, $r_i$ is unassigned or prefers $h_{j,k}$ to $M^{\prime }(r_i)$ and $h_{j,k}$ is also either under-subscribed or prefers $r_i$ to $r_p =M^{\prime }(h_{j,k})$, and thus $(r_i, h_{j,k})$ MM-blocks $M^{\prime }$ in $I^{\prime }$, a contradiction.

\noindent\emph{Case (2)}: Resident couple $(r_{2i-1}, r_{2i})$ MM-blocks $M$ in $I$ with $(h_{j_1}, h_{j_2})$ and either $r_{2i-1} \in M(h_{j_1})$ or $r_{2i} \in M(h_{j_2})$
%
%\begin{center}
%\begin{tabular}[t]{lcrcccl}
%
%
%$(r_{2i-1}, r_{2i})$ 	&	:	&	\myline	&	$(h_{j_1}, h_{j_2})$		&		\myline	&	($M(r_{2i-1}), h_{j_2})$	&	\myline  \\
%
%\\
%
%
%\end{tabular}
%\end{center}
%
%
%\begin{center}
%\begin{tabular}[t]{llcrcccl}
%
%
% & $h_{j_1}$		&	:	&	\myline	&	$r_{2i-1}$ 	&	 	\myline	&	$r_p\in M(h_{j_1})$	& \myline 	\\
%
%
%
%\\
%
%
%\end{tabular}
%\end{center}
Consider the case that $M(r_{2i})=h_{j_2}$. Then $(r_{2i-1},r_{2i})$ prefers $(h_{j_1}, h_{j_2})$ to $(M(r_{2i-1}), h_{j_2})$. Hence, $h_{j_1}$ is either under-subscribed in $M$ or prefers $r_{2i-1}$ to some $r_p\in M(h_{j_1})$ in $M$. If $h_{j_1}$ is under-subscribed in $M$ then there must be a hospital clone $h_{j_1, k_1}$ for some $k_1 ~ (1\leq k_1\leq c_{j_1})$ which is under-subscribed in $M^{\prime }$. Further, since $r_{2i}$ is assigned to $h_{j_2}$ in $M$ there must be a hospital clone $h_{j_2, k_2}$ for some $k_2 ~ (1\leq k_2\leq c_{j_2})$ assigned to $r_{2i}$ in $M^{\prime }$. Therefore, $(r_{2i-1},r_{2i})$ must MM-block $M^{\prime }$ in $I^{\prime }$ with $(h_{j_1, k_1}, h_{j_2, k_2})$, a contradiction. 

Hence $h_{j_1}$ prefers $r_{2i-1}$ to some $r_p\in M(h_{j_1})$ in $M$. Let $k_1 ~ (1\leq k_1\leq c_{j_1})$ be such that $(r_p , h_{j_1 , k_1} \in M^{\prime }$. Thus $h_{j_1, k_1}$ prefers $r_{2i-1}$ to $r_p$. Also, as before there is some $k_2 ~ (1\leq k_1\leq c_{j_2})$ such that $(r_{2i} , h_{j_2, k_2}) \in M^{\prime }$. Hence $(r_{2i-1}, r_{2i})$ MM-blocks $M^{\prime }$ with $(h_{j_1 , k_1} , h_{j_2 , k_2 })$ in $I^{\prime }$, a contradiction. A similar argument may be applied in the case that $M(r_{2i-1}) = h_{j_1})$.

\noindent\emph{Case (3a)}: Resident couple $(r_{2i-1}, r_{2i})$ MM-blocks $M$ with $(h_{j_1}, h_{j_2})$ in $I$ (where $j_1 \neq j_2$).
%
%\begin{center}
%\begin{tabular}[t]{lcrcccl}
%
%
%$(r_{2i-1}, r_{2i})$ 	&	:	&	\myline	&	$(h_{j_1}, h_{j_2})$		&		\myline	&	$(M(r_{2i-1}), M(r_{2i}))$	&	\myline  \\
%
%
%
%
%\end{tabular}
%\end{center}
%
%
%\begin{center}
%\begin{tabular}[t]{llcrcccl}
%
%
%
%
% & $h_{j_1}$		&	:	&	\myline	&	$r_{2i-1}$ 	&	 	\myline	&	$r_p \in M(h_{j_1})$		& \myline \\
%
%
%
%
%\textbf{and} & $h_{j_2}$		&	:	&	\myline	&	$r_{2i}$ 	&	 \myline	&	$r_q \in M(h_{j_2})$		& \myline\\
%
%\\
%
%
%\end{tabular}
%\end{center}
%
Then $(r_{2i-1}, r_{2i})$ is either jointly unassigned or jointly assigned to a worse pair in $M$ than $(h_{j_1}, h_{j_2})$ and also each of $h_{j_1}$ and $ h_{j_2}$ is either under-subscribed or has a worse partner than $r_{2i-1}$ and $r_{2i}$ respectively amongst their assignees in $M$. 

If $h_{j_1}$ (respectively $h_{j_2}$) is under-subscribed in $M$ then there must be a hospital clone $h_{j_1, k_1}$ for some $k_1 ~ (1\leq k_1\leq c_{j_1})$ (respectively $h_{j_2, k_2}$ for some $k_2 ~ (1\leq k_2\leq c_{j_2})$) that is under-subscribed in $M^{\prime }$. Further, if $h_{j_1}$ (respectively $h_{j_2}$) has amongst its assignees in $M$ a worse partner than $r_{2i-1}$ (respectively $r_{2i}$) then there must be some $h_{j_1, k_1} ~ (1\leq k_1\leq c_{j_1})$ (respectively $h_{j_2, k_2} ~ (1\leq k_2\leq c_{j_1})$) that is assigned to a worse partner than $r_{2i-1}$ (respectively $r_{2i}$) in $M^{\prime }$.

From the construction, this means that in $M^{\prime }$, $(r_{2i-1}, r_{2i})$ is either jointly unassigned or jointly prefers $(h_{j_1, k_1}, h_{j_2, k_2})$  to $M^{\prime }(r_{2i-1}, r_{2i})$ and also each of $h_{j_1, k_1}$ and $h_{j_2, k_2}$ is also either under-subscribed in $M^{\prime }$ or prefers $r_{2i-1}$ and $r_{2i}$ to $M^{\prime }(h_{j_1, k_1})$ and $M^{\prime }(h_{j_2, k_2})$ respectively. Hence $(r_{2i-1}, r_{2i})$ MM-blocks $M^{\prime }$ in $I^{\prime }$ with $(h_{j_1, k_1}, h_{j_2, k_2})$, a contradiction.

\noindent\emph{Case (3b)}: Resident couple $(r_{2i-1}, r_{2i})$ MM-blocks $M$ in $I$ with $(h_{j}, h_{j})$ and $h_{j}$ has two free posts.
%
%\begin{center}
%\begin{tabular}[t]{lcrcccl}
%
%
%$(r_{2i-1}, r_{2i})$ 	&	:	&	\myline	&	$(h_{j_1}, h_{j_1})$		&		\myline	&	$(M(r_{2i-1}), M(r_{2i}))$	&	\myline  \\
%
%
%\end{tabular}
%\end{center}
Then $(r_{2i-1}, r_{2i})$ is either jointly unassigned or jointly assigned in $M$ to a worse hospital pair than $(h_{j}, h_{j})$ in $M$. From the construction this means that in $M^{\prime }$ there are two hospital clones $h_{j, k_1}$ and $h_{j, k_2}$ for some $1\leq k_1 \leq c_j, 1\leq k_2 \leq c_j, k_1 \neq k_2$ that are under-subscribed in $M^{\prime }$. Since $(r_{2i-1}, r_{2i})$ is jointly matched in $M^{\prime }$ to a worse partner than $(h_{j, k_1}, h_{j, k_2})$, $(r_{2i-1}, r_{2i})$ MM-blocks $M^{\prime }$ in $I^{\prime }$ with $(h_{j, k_1}, h_{j, k_2})$, a contradiction.

\noindent\emph{Case (3c)}: Resident couple $(r_{2i-1}, r_{2i})$ MM-blocks $M$ in $I$ with $(h_{j}, h_{j})$ and $h_{j}$ has one free post and prefers at least one of $r_{2i-1}$ or $r_{2i}$ to some $r_p\in M(h_{j_1})$.
%
%\begin{center}
%\begin{tabular}[t]{lcrcccl}
%
%
%$(r_{2i-1}, r_{2i})$ 	&	:	&	\myline	&	$(h_{j_1}, h_{j_1})$		&		\myline	&	$(M(r_{2i-1}), M(r_{2i}))$	&	\myline  \\
%
%
%	
%\\
%
%
%\end{tabular}
%\end{center}
%
%\begin{center}
%\begin{tabular}[t]{llcrcccl}
%
%
%
%
% & $h_{j_1}$		&	:	&	\myline	&	$r_{2i-1}$ 	&	 	\myline	&	$r_p \in M(h_{j_1})$		& \myline \\
%
%
%
%
%\textbf{or} & $h_{j_1}$		&	:	&	\myline	&	$r_{2i}$ 	&	 \myline	&	$r_q \in M(h_{j_2})$		& \myline\\
%
%
%
%	
%\\
%
%
%\end{tabular}
%\end{center}
Also $(r_{2i-1}, r_{2i})$ is either unassigned or jointly assigned to a worse hospital pair than $(h_{j}, h_{j})$.

If $h_j$ prefers $r_{2i-1}$ to $r_p$ in $I$ then $(r_{2i-1}, r_{2i})$ MM blocks $M^{\prime }$ with $(h_{j , k_2} , h_{ j, k_1})$ in $I^{\prime }$, a contradiction. Otherwise $h_j$ prefers $r_{2i}$ to $r_p$ in $I$ so that $(r_{2i-1} , r_{2i})$ MM-blocks $M^{\prime }$ with $(h_{j , k_1} , h_{ j, k_2})$ in $I^{\prime }$, a contradiction.

%From the construction this means that in $M^{\prime }$ there is a hospital clone $h_{j, k_1}$ that is unmatched and a further hospital clone $h_{j, k_2}$ that prefers $r_{2i}$ to $r_p = M^{\prime }(h_{j_1, k_2})$ . Hence $(r_{2i-1}, r_{2i})$ MM-blocks $M^{\prime }$ in $I^{\prime }$ with $(h_{j_1, k_1}, h_{j_1, k_2})$, a contradiction. A similar argument may be applied for the case that $h_{j_1}$ prefers $r_{2i}$ to some $r_p\in M(h_{j_1})$. 

\noindent\emph{Case (3d)}: Resident couple $(r_{2i-1}, r_{2i})$ MM-blocks $M$ with $(h_{j}, h_{j})$ in $I$ where $h_{j}$ is full and $h_{j}$ prefers $r_{2i-1}$ to some $r_p \in M(h_{j})$ and also prefers $r_{2i}$ to some $r_q \in M(h_{j}) \setminus \{ r_p \}.$
%
%
%\begin{center}
%\begin{tabular}[t]{lcrcccl}
%
%$(r_{2i-1}, r_{2i})$ 	&	:	&	\myline	&	$(h_{j_1}, h_{j_1})$		&		\myline	&	$(M(r_{2i-1}), M(r_{2i}))$	&	\myline  \\
%
%
%\\
%
%
%
%\end{tabular}
%\end{center}
%
%
%
%\begin{center}
%\begin{tabular}[t]{llll}
%
%
%
%
%& $h_{j_1}$		&	:	&	\shortline		$r_{2i-1}$ 	 	\shortline	 $M(r_{2i-1})$  \shortline 	$r_{2i}$	\shortline		 $M(r_{2i})$  \shortline	\\
%\\
%
%\textbf{or} & $h_{j_1}$		&	:	&	\shortline		$r_{2i-1}$ 	 	\shortline	$r_{2i}$   \shortline 	$M(r_{2i-1})$	\shortline		 $M(r_{2i})$  \shortline	\\
%
%
%\end{tabular}
%\end{center}
Then, in $I$, $(r_{2i-1}, r_{2i})$ is either unassigned or jointly assigned to a worse hospital pair than $(h_{j}, h_{j})$

%and $h_{j_1}$ prefers $r_{2i-1}$ to some member of $M(h_{j_1})$ and $h_{j_1}$ prefers $r_{2i}$ to some member of $M(h_{j_1}) \setminus \{r_{2i-1}\}$ in $M$.

Since $h_{j}$ prefers $r_{2i-1}$ in $I$ to one of its assignees in $M$ there must be a hospital clone $h_{j, k_1}$ for some $k_1 ~ (1\leq k_1\leq c_{j})$ that is assigned to a worse partner, $r_p$ than $r_{2i-1}$ in $M^{\prime }$. Further, since $h_{j}$ prefers $r_{2i}$ to some member of $M(h_{j}) \setminus \{r_{p}\}$ there must be a hospital clone $h_{j, k_2}$ for some $k_2 ~ (1\leq k_2 \leq c_{j})$ assigned to a partner worse than $r_{2i}$.

From the construction this means that in $I^{\prime }$, there are two distinct hospital clones $h_{j, k_1}$ and $h_{j,k_2}$ such that $(r_{2i-1}, r_{2i})$ jointly prefers $(h_{j, k_1}, h_{j,k_2})$ to $(M^{\prime }(h_{j,k_1}), M^{\prime }(h_{j, k_2}))$. Also, $h_{j, k_1}$ prefers $r_{2i-1}$ to $M^{\prime }(h_{j, k_1})$ and $h_{j, k_2}$ prefers $r_{2i}$ to $M^{\prime }(h_{j , k_2})$. Hence $(r_{2i-1}, r_{2i})$ MM-blocks $M^{\prime }$ in $I^{\prime }$ with $(h_{j , k_1}, h_{j , k_2})$, a contradiction.
\end{proof}

\begin{corollary}
\label{corollary:cloningunderBISdoesntwork}

If $I^{\prime }$ admits a BIS-stable matching it need not be the case that $I$ admits a BIS-stable matching.

\end{corollary}
\begin{proof}

Let $I$ be an instance of {\sc hrc} as shown in Figure \ref{instance:MMbutnostableBIS}. Clearly, the instance admits no BIS-stable matching. We construct the instance $I^{\prime }$ from $I$ exactly as in Lemma \ref{lemma:cloningunderMMworks} by creating two distinct hospital clones $h_{1 , 1}$ and $h_{1 , 2}$ to represent the two posts in $h_1$ and amending the preference lists of the couples as shown in Figure \ref{instance:cloningbreaksdown}.

The instance $I^{\prime }$ supports the matching $M^{\prime } = \{ (r_3 , h_{1, 1}) , (r_4 , h_{1,2}) \}$ which is BIS-stable. Since $I$ admits no BIS-stable matching, $M^{\prime }$ clearly has no corresponding BIS-stable matching in $I$ and the cloning method described above breaks down under BIS-stability.

\begin{figure}[h]
\begin{center}
\begin{tabular}[t]{lcrccccl}

\multicolumn{8}{c}{Residents} \\
\hline
\\

 & $(r_1, r_2)$		&	:		&	$(h_{1 ,1}, h_{1, 2})$ & $(h_{1, 2}, h_{1 , 1})$		& & &	\\
 
 & $(r_3, r_4)$		&	:		&	$(h_{1 ,1}, h_{1, 2})$ & $(h_{1, 2}, h_{1 , 1})$  &	 $(h_{1, 1}, h_2)$		&			 $(h_{1, 2}, h_2)$ & 	\\

\\

\multicolumn{8}{c}{Hospitals} \\
\hline
\\

& $h_{1, 1}$ 	&:$ ~~~	1 ~~~ $:& $r_3$	&	$r_1~~~~~$		&	$r_2~~~~~~$		& $r_4~~~~~$ &    \\

& $h_{1, 2}$ 	&:$ ~~~	1 ~~~ $:& $r_3$	&	$r_1~~~~~$		&	$r_2~~~~~~$		& $r_4~~~~~$ &    \\

& $h_2$ 	&:$ ~~~	1 ~~~ $:& $r_4$				&			&	  &  & \\

\end{tabular}
\end{center}
\caption{An instance of HRC that shows that the cloning method breaks down under BIS-stability.}
\label{instance:cloningbreaksdown}
\end{figure}
\end{proof}

\section{Conclusions}
\label{section:conclusions}

The new NP-completeness results presented in this work suggest that an efficient algorithm for \small $(2,2)$\normalsize {\sc -hrc} is unlikely. However, we conclude that the IP model presented in this paper performs well when finding a maximum cardinality stable matching in {\sc hrc} instances that are similar to those arising in the SFAS application. 

It remains open to investigate the performance of the model as we increase the size of the instance substantially beyond that of the SFAS application. It would also be of further interest to investigate other modelling frameworks, for example involving pseudoboolean solvers or CP strategies.

The IP model for {\sc hrc} presented here might also be updated to produce maximum cardinality stable matchings under other stability definitions, most obviously BIS-stability. It would be of interest to compare an IP model producing exact maximum cardinality BIS-stable matchings against {\sc hrc} heuristics such as those compared and contrasted by Bir\'o et al \cite{BIS11}.

%It also remains open to establish the frontier between polynomial time solvability and NP-completeness for further restricted variants of {\sc hrc} where the length of the preference list of each single resident, couple and hospital is 1 or 2.

\subsubsection*{Acknowledgement.}
We would like to thank the anonymous reviewers of submitted versions of this work for their valuable comments.

\bibliographystyle{plain}
\bibliography{matching}
\end{document}